\documentclass[preprint,12pt]{elsarticle}

\usepackage{amssymb}
\usepackage{graphicx}
\usepackage{wrapfig}
 \usepackage{mathtools} 
\usepackage{float}
\usepackage{subcaption}
\usepackage{parboxx}

\usepackage{fix-cm}
\usepackage{ulem}
\usepackage{soul}
\usepackage{tikz}
\usetikzlibrary{patterns}
\usepackage{mathrsfs}
\usepackage[ruled,vlined]{algorithm2e}

\usepackage{amssymb}
\usepackage{amsmath,dsfont}
\usepackage{xcolor}
\usepackage{nicefrac}
\usepackage{mdframed}

\usepackage{float}

\newcommand{\extrawork}[1]{}

\usepackage{hyperref}
\newcommand{\Canticipate}[1]{}
\newcommand{\ignore}[1]{}
\newcommand{\Ilextra}[1]{}
\newcommand{\ENMremoved}[1]{}
\newcommand{\Details}[1]{}
\newcommand{\extra}[1]{}
\newcommand{\jlt}[1]{}

\renewcommand{\nu}{w}

\newcommand{\manu}{manufacturer}
\newcommand{\spcialComment}[1]{}

\newtheorem{cor}{Corollary}

\newcommand{\eop}{{\hfill~$\Box$}}

\renewcommand{\S}{{\mathcal S}}

\newcommand{\M}{{\mathbb M}}

\newcommand{\V}{{\mathbb V}}
\renewcommand{\S}{{\mathbb S}}

\newcommand{\Me}{M_e}
\newcommand{\Mi}{M}

\newcommand{\bMe}{{\bf Me}}
\newcommand{\bMi}{{\bf M}}
\newcommand{\bS}{{\bf S}}

\newcommand{\pe}{p_e}

\newcommand{\Vi}{{{\mathbb V}}}

\newcommand{\dbar}{\bar{d}}

\newcommand{\indc}[1]{\mathds{1}_{\left\{ #1 \right \}} }

\newcommand{\F}{{\cal F}}

\newcommand{\sM}{{\mbox{\fontsize{4.7}{5}\selectfont{$\M$}}}}
\newcommand{\sMi}{{\mbox{\fontsize{4.7}{4.5}\selectfont{$\bMi$}}}}
\newcommand{\sMe}{{\mbox{\fontsize{4.7}{4.5}\selectfont{$\bMe$}}}}

\newcommand{\sS}{{\mbox{\fontsize{5.2}{5.2}\selectfont{$\bS$}}}}

\newcommand{\sV}{{\mbox{\fontsize{5.2}{5.2}\selectfont{$\V$}}}}

\newcommand{\I}{{\mathbb I}}

\newcommand{\RB}[1]{}

\newcommand{\Sco}{\eta_{co}}
\newcommand{\Sse}{\eta_{se}}

\newcommand{\pmax}{p_{mx}}

\newtheorem{lemma}{Lemma}
\newtheorem{theorem}{Theorem}

\vspace{-2mm}

\journal{International Journal of Production Economics}

\begin{document}

\begin{frontmatter}



\title{What should the encroaching supplier do?: A Stackelberg Game Approach}


\author{Gurkirat Wadhwa, Veeraruna Kavitha}

\affiliation{organization={IEOR,IIT Bombay},
            addressline={Powai}, 
            city={Mumbai},
            postcode={400076}, 
            state={Maharashtra},
            country={India}}

\begin{abstract}
Suppliers often encroach downstream by operating in-house production-units while continuing to supply independent production-units. 
We study the optimal configuration, including optimal pricing, for an encroaching supplier that balances these dual roles through a Stackelberg game. The integrated supplier  determines the wholesale price charged to the outsourced production unit and the retail price of its own product, while the outsourced unit responds optimally.
 Customer demand-response incorporates both price-based substitutions (of the two production-units) and loyalty (towards individual units).
With strong customer loyalty and luxury products, at the  optimal choice for the coalition,   both units  co-exist profitably. In contrast, when the products become essential, the optimal strategy depends upon customer-fallback rates (fraction of the exiting production-unit's market that falls-back to other). Under low fallback, the coalition either sustains  co-existence at maximum prices or disciplines the out-house to operate at break-even---with high fallback it  is optimal to shut-down the in-house or eliminate the out-house---we derive two factors that identify the above.  We further develop a numerical procedure to identify the optimal regime for any given set of parameters. Two surprising results are---higher market potential of the out-house can become a reason for it to operate at break-even---and the coalition may find it beneficial to operate its in-house at losses, particularly  for products that are
neither highly essential nor in the luxury category. 
\end{abstract}




\end{frontmatter}


\vspace{-2mm}

\section{Introduction}
\label{sec_intro}

Supplier encroachment refers to a strategic move where suppliers bypass traditional distribution channels, set-up an in-house production unit
to sell directly to end consumers while continuing supplying to downstream  or  lower echelon manufacturers. Such strategies have become increasingly prevalent as firms seek to capture larger margins, strengthen brand identity,  reduce dependence on downstream intermediaries and thereby gain more control over the supply chain   (see e.g., \cite{arya2007bright,ha2022supplier,yoon2016supplier}).
This trend is evident across various industries, transforming conventional supply chain (SC) dynamics. For instance, in the automotive industry, major suppliers such as \href{https://www.boschautoservice.com/} Bosch and \href{https://www.continental-aftermarket.com/us-en} Continental now market parts and services directly to consumers, thereby enhancing their brand visibility and fostering closer customer relationships. 
Acer Inc., initially a supplier for IBM and Apple, leveraged this strategy to become one of the big computer manufacturers worldwide by 2007 (\cite{nystedt2007acer}).

The arrangement is closely linked to another aspect namely  `vertical integration'  studied in SC literature (e.g., \cite{ursino2015supply,simchi1999designing}). Vertical integration typically implies the integration of various units (across various echelons) into a single unit that   controls    multiple stages of production and distribution   (e.g., \cite{simchi1999designing,wadhwapartition,zheng2021willingness}).  The idea in most of this literature is to illustrate  the advantages of a centralized SC formed by complete integration of all  the manufacturers and the supplier.  
%
%
Recently in \cite{wadhwapartition}, we showed that, for essential products, a partially integrated supplier--manufacturer coalition facing competition from another manufacturer is more stable than the fully centralized SC (a structure that is not   opposed by other collaborative arrangements). Extending this analysis to more general settings, including non-essential products, remains an open problem.
A key challenge is to characterize the worths of coalitions under different partition structures, with the main hurdle being the characterization of worths of the individual coalitions in a  partition involving partial vertical integration (see \cite{wadhwa2025should,wadhwapartition}). 

The common feature in both the aspects mentioned above,  is a single unit that has capacity  spanning across multiple echelons. 
In 
 this study, we investigate one such  SC with one supplier and two manufacturers, where the supplier collaborates with one of the manufacturers resulting in a partial vertical integration, while competing with the other.  We  explore the optimal operating strategies for the vertical collaborating unit and thereby derive it's worth---when it acts as a leader by setting the wholesale price for the raw  materials  (to the out-house manufacturer) and by quoting another price to the end customers---while anticipating    the optimal response of the out-house, that acts as a follower. We also derive the worth of the out-house manufacturer.

We study the partially integrated SC using a Stackelberg model that incorporates market reputation, product essentialness, demand substitution or fallback rate  (when a  manufacturer  exits the  market), production costs, and non-negligible operating costs. Such features are often ignored in existing encroachment models (e.g., \cite{arya2007bright,ha2022supplier}, which do not consider operating costs or the option to cease operations). Accounting for these factors reveals several surprising market-dependent optimal configurations for the supplier--manufacturer duo---including loss-making in-house operation, forcing out-house to operate at break-even, shutdown of one of the two units, and the conventional regime of profitable co-existence.

\ignore{As already mentioned, one interpretation of our SC is related to supplier encroachment---one can view the above arrangement as  a supplier  with an in-house production unit that also outsources   materials to an independent out-house production unit. Elaborate study of such an arrangement, as mentioned above,  can  answer one central managerial question:}

As noted earlier, our SC can also be interpreted as a supplier encroachment model, where the supplier operates an in-house production unit while simultaneously supplying materials to an independent outsourced production unit. Analyzing this arrangement in detail, as mentioned above,  helps address a central managerial question:

\begin{quote}
\textit{Under what market conditions should a supplier sustain both roles---supplying an out-house manufacturer while operating an in-house unit---and when should it instead discipline, withdraw from, or eliminate one of them?}
\end{quote}

\ignore{On the other hand,  the current   paper also derives the  the worths of  the coalitions in a partition, involving partial vertical integration, under  fairly general conditions  compared to that in \cite{wadhwapartition}, which  primarily   focuses  on essential products. This would enable a more complete future analysis of coalition formation game in a one-supplier, two-manufacturer SC.}

On the other hand, this paper also derives the coalition worths for partitions involving partial vertical integration under substantially more general conditions than \cite{wadhwapartition}, which primarily focuses on essential products. The results of this paper hence can facilitate a more comprehensive future analysis of coalition formation game in a one-supplier, two-manufacturer SC.

A preliminary conference version of this work \cite{wadhwa2025should}, that restricted attention only to co-existence possibilities,  illustrated  that strong customer loyalty (towards individual manufacturers) sustains profitable co-existence of both units; while for more essential products (where customers  are desperate to buy from any) the coalition can push the out-house manufacturer to break-even, or even the in-house unit into a loss. The analysis in \cite{wadhwa2025should} stopped short of asking whether co-existence is optimal at all, relative to shutting down the in-house unit or eliminating the out-house manufacturer entirely.

The present paper closes  the gaps in  analysis of \cite{wadhwa2025should} by investigating the coalition's choice across all major operating regimes. Our contributions are as follows.

\begin{enumerate}
    \item[i)] \textbf{A complete regime taxonomy:} We formalize various operating regimes available to the coalition---operate both units profitably (Bp), in-house at   loss (I$\ell$), out-house forced to operate at par (Op), shutdown  the in-house unit (Sh), and elimination of the out-house  (E$\ell$); we also  derive  optimal prices and utilities for each regime.

    \item[ii)] \textbf{Optimal regime in two asymptotic scenarios:} We prove that when products are weak substitutes and the customer loyalty is high (the low-essentialness  scenario), profitable co-existence (Bp) is optimal. Thus for luxury  or   weakly substitutable products, the coalitional  advantage of the encroaching supplier does not create a significantly adverse outcome for the outsourced production unit.  

    \medskip
        Conversely, in the high-essentialness scenario, we show  that the optimal regime depends upon the fallback rates. When the rates are small (less than some  fallback-rate threshold $\bar{r}$),  the optimal choice reduces to a choice between forcing the out-house to operate at par (Op) or  profitable co-existence  (Bp)---we derive a closed-form expression for a certain relative-strength score, constructed using the parameters of the SC,   that dictates this optimal choice. For high fallback, the optimal choice is between shutting down the in-house unit (Sh) or eliminating the out-house manufacturer (E$\ell$)---this is governed by a second relative-strength indicator score.

    \item[iii)] \textbf{A numerical procedure:} 
    We further provide a numerical procedure that computes the optimal regime   for any given set of parameters---by comparing the closed-form sub-optimal utilities of each possible operating regime.  
    
    We   plot these optimal regimes to obtain `regime maps', as essentialness factor ($\varepsilon$) and fallback rate ($r$) are varied continually, while keeping the rest of the parameters fixed---the resulting optimal regime maps show how the two asymptotic results extend—or, in some cases, fail to extend—for intermediate values of essentialness factor. 

  We have some surprising  observations:  (a) when the in-house unit is inferior, operating it  at loss  (I$\ell$) becomes optimal for a non-trivial intermediate range of essentialness factor, when the fall-back rates are small;  and (b) 
    when the out-house is significantly superior---both in price-sensitivity or market reputation and production costs---its higher potential can paradoxically become a disadvantage, particularly under low fallback rates---the coalition can push it to operate at break-even point, by quoting a sufficiently large wholesale price.
\end{enumerate}

Together, these results show that profitable co-existence, far from being the generic outcome of partial vertical integration, is confined to conditions of high customer loyalty (and for weakly-substitutable products); when products become more essential, the supplier optimally moves through a sequence of increasingly aggressive strategies---disciplining the out-house manufacturer to break-even, sustaining its own in-house unit at a loss to forestall downstream monopoly, or, at the extreme, shutting down one of the two units altogether---with the specific choice pinned down by simple, computable market-strength scores at the two ends of the essentialness spectrum and by direct numerical comparison in between.
Overall, our results demonstrate that partial vertical integration is not merely an intermediate organizational form between outsourcing and full integration, but an active strategic lever for shaping downstream market structure.

\subsection*{Literature Review}

Supplier encroachment, whereby an upstream supplier sells directly to end customers while continuing to supply downstream firms, has been extensively studied in the supply chain literature. Early studies primarily emphasized its adverse effects, arguing that encroachment intensifies channel competition and may reduce downstream profitability \cite{frazier1996determinants,fein1997patterns}. Subsequent research, however, demonstrated that supplier encroachment can also improve channel performance by mitigating double marginalization and better aligning upstream and downstream incentives. For example, \cite{chiang2003direct} showed that introducing a direct sales channel may reduce wholesale prices and improve retailer profitability when consumers exhibit sufficient acceptance of the direct channel, while \cite{arya2007bright} identified conditions under which supplier encroachment benefits both suppliers and retailers. More recent studies have extended this literature by incorporating information asymmetry, digital channels, and platform-mediated competition, highlighting the influence of market structure and information availability on the profitability of encroachment \cite{ha2022supplier}. Collectively, these studies show that supplier encroachment simultaneously involves upstream cooperation and downstream competition, making it a fundamental strategic issue in modern supply chains.

Despite these advances, much of the analytical literature focuses primarily on pricing and channel coordination within a predetermined operating structure. In many existing models, firms optimize prices while the market configuration itself is assumed to remain fixed throughout the game. Comparatively less attention has been devoted to settings in which an integrated supplier can strategically determine the operating structure by choosing whether to sustain co-existence, maintain a downstream competitor at break-even, temporarily accept losses, shut down its own production unit, or eliminate downstream competition. Furthermore, the joint analytical treatment of dedicated customer bases, price-based demand substitution, and non-negligible production and operating costs remains limited. In our formulation, incorporating these features results in piecewise-defined payoff functions with multiple feasible operating regimes, motivating a regime-by-regime analytical framework rather than a conventional single-regime pricing analysis.

A related stream of research investigates the trade-off between outsourcing and in-house production. For example, \cite{kaya2011outsourcing} studies sourcing decisions under effort-dependent demand, while \cite{wang2013advantage} analyzes competition between an original equipment manufacturer and a contract manufacturer within a Stackelberg framework. These studies demonstrate that coexistence between outsourcing and internal production can emerge under suitable contractual arrangements. However, the production structure is generally specified exogenously, with relatively little attention devoted to settings in which an upstream supplier simultaneously supplies and competes with a downstream manufacturer while optimally determining its mode of operation.

Another closely related stream considers vertical integration as a mechanism for improving coordination and mitigating double marginalization \cite{simchi1999designing,ursino2015supply,arora2025vertical}. While this literature demonstrates the efficiency benefits of integration, many models either assume complete integration or treat the integration structure as fixed. Consequently, they provide limited insight into how a partially integrated coalition strategically interacts with an independent downstream manufacturer when multiple operating regimes are feasible.

Game-theoretic models, particularly Stackelberg games, have been widely employed to analyze leader--follower interactions in supply chains \cite{li2006channel,yan2011managing,lin2011dual,das2022integration,taleizadeh2016pricing}. These studies have generated important insights into pricing, channel power, and competitive behaviour. Nevertheless, supplier encroachment, outsourcing, and partial vertical integration have largely developed as related but distinct research streams, and comparatively few analytical frameworks integrate these decisions while allowing the operating regime itself to be determined endogenously.

Recent behavioural studies further recognize that firms may deliberately sacrifice short-term profitability to preserve long-term strategic objectives \cite{zheng2021willingness}. Although these studies provide valuable insights into strategic behaviour, they do not explicitly examine such decisions within a supplier encroachment framework that simultaneously captures partial vertical integration, asymmetric dependence, and endogenous operating regimes.

Motivated by these observations, this paper develops a unified Stackelberg framework for supplier encroachment under partial vertical integration. The model jointly incorporates customer loyalty, demand substitution through fallback behaviour, production and operating costs, and endogenous regime selection. Rather than optimizing prices within a predetermined operating structure, the integrated supplier optimally chooses among multiple operating regimes, including profitable coexistence, break-even operation of the downstream manufacturer, loss-making operation of the in-house unit, shutdown of the in-house unit, and elimination of the downstream manufacturer.

The analysis contributes in four directions. First, it characterizes the conditions under which different operating regimes become optimal. Second, it derives closed-form analytical characterizations for high- and low-essentialness scenarios together with a numerical procedure for determining the optimal regime in intermediate parameter regions. Third, it demonstrates how customer loyalty, product essentialness, demand fallback, and market asymmetry jointly influence the supplier's optimal operating strategy. Finally, numerical experiments over economically meaningful parameter ranges illustrate how optimal operating regimes evolve as product essentialness and demand transfer intensify. Overall, the paper contributes to the supplier encroachment and partial vertical integration literature by explicitly treating the operating regime as an endogenous strategic decision rather than assuming a predetermined operating structure.
\begin{figure}[h]
\vspace{-5mm}
\centering
    \includegraphics[width = \linewidth, height= 6 cm]{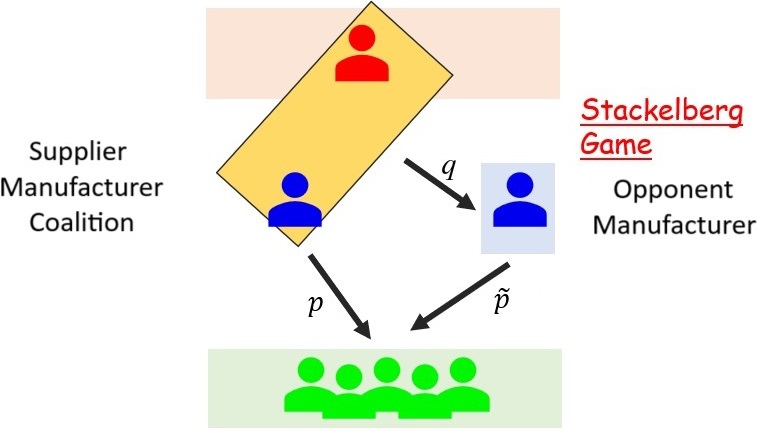}
    \caption{Model Description}
    \label{fig:model}
\vspace{-4mm}
\end{figure}

\section{Model Description}
\label{sec:model}
Consider  a partially integrated SC with one   supplier $S$  that collaborates with  a manufacturer  $\Mi$,  by forming the coalition $\Vi = \{\Mi, S\}$. The coalition  competes with an  out-house manufacturer,  referred by $\Me$. 

We consider a Stackelberg framework, where the members of  the  coalition $\Vi$ quote their prices  first: (a) the supplier $S$ quotes wholesale price $q$   to $\Me$ for the raw materials; and (b) the manufacturer $M$ quotes price $p$ for its final product to the end-customers. Thus  the coalition $\V$  forms the leader of the 
  Stackelberg game, while the manufacturer  $\Me$  is the follower and responds to the quoted prices $(q,p)$.   The latter  
 can choose to operate by quoting a price $\pe$ to the end customers 
 for its own finished product. It can  also choose not to operate represented by action $n_o$,  if  the resultant market response is not conducive. Any unit can choose $n_o$ and with such a choice, the corresponding unit is completely shut-down and incurs zero utility (zero profit and    cost). 
  The out-house manufacturer $\Me$ uses  the  raw material supplied by $\V$
  for producing the  finished products,  see Figure \ref{fig:model} for the flow of materials and the prices.

\subsection{Market Response}  The demand attracted by any manufacturer depends upon the price quoted for the finished product, for example, that attracted by
 manufacturer $\Mi$ is given by  (see \cite{wadhwapartition,zheng2021willingness} for similar models):
 \vspace{-2mm}
\begin{eqnarray}\label{eqn_demand_vc_coal}
D_\sMi =  (\dbar_\sMi - \alpha_\sMi p +\varepsilon\alpha_\sMe \pe)^{+},
\end{eqnarray}
where the different influencing factors are as  below:
\begin{itemize}
\item $\dbar_\sMi$  is the dedicated market potential of manufacturer~$\Mi$,
\item $\alpha_\sMi p$ is the fraction of demand lost by $\Mi$ due to
 its quoted   price $p$, sensitized by  parameter $\alpha_\sMi$ (here $\alpha_\sMi$ can be  a representative of the reputation of $\Mi$),
 
 \item The demand is positive as long as the term inside
 $(.)^+$ is positive; else, the demand is zero.

 \item  \ignore{the last component $\varepsilon g_j(\pe) $ is due an `unhappy' fraction of the loyal customer base of the out-house manufacturer $\Me$ and is explained in the immediate following.} $\varepsilon \alpha_\sMe
 \pe$ is the fraction of customer base of $\Me$
 that rejected $\Me$ (due to its quoted price $\pe$) and got converted as customers of $\Mi$.
\end{itemize}

\ignore{
The demand $D_\sMe$ attracted by manufacturer $\Me$ 
has exactly similar structure with respective parameters, ($\dbar_\sMe, \alpha_\sMe $), the fraction of the loyal customers of $\Me$ that are unhappy with price $\pe$  exactly equal $g_j(\pe) = 
\varepsilon \alpha_\sMe \pe$ and an $\varepsilon$  fraction among these unhappy customers seek service from $\Mi$; additionally observe $g_j(\pe)$ can at maximum be $\dbar_\sMe$, see  \eqref{eqn_demand_vc_coal}.}

The parameter $\varepsilon$ represents the essentialness of the product. 
When $\varepsilon \approx 1$, the product is essential implying that the manufacturers (or their products) are substitutable and the customers can buy the product from any of the manufacturers. On the other hand, when $\varepsilon \approx 0$, the product is not essential, i.e.,  the customers are loyal and choose  to buy the product only from `their'   manufacturers. 

\subsection{Utilities} 
We begin with the utility of out-house manufacturer $\Me$. When it does not operate, represented by indicator  $\F^c_\sMe = \indc{\pe= n_{o}}$, it derives zero utility. 
When it operates (represented by $\F_\sMe$), it attracts demand as in \eqref{eqn_demand_vc_coal} and then the revenue derived equals the demand times the price minus the expenses (the raw material price $q$ plus the production cost). 
 Thus the utility of the manufacturer $\Me$  equals:
\begin{eqnarray}
 U_\sMe (\pe; p, q) &=& \left(D_\sMe \left(\pe- q - C_\sMe \right) {\cal F}_{\sV} - O_\sMe \right)\F_\sMe, \mbox{ with, } \label{Eqn_Umj}\\  D_\sMe &=&  (\dbar_\sMe - \alpha_\sMe \pe+\varepsilon\alpha_\sMi p)^{+},   \nonumber
\end{eqnarray}
where    $C_\sMe$ represents the  production cost per unit and $O_\sMe$ represents the operating cost. The profit of manufacturer
 $\Me$ is zero   when  the
 supplier does not operate (represented by indicator ${\cal F}_{\sV}^c$).  Some dependencies are suppressed, when there is clarity,  to keep the notations simple. 

The utility of $\V$ due to  demand $D_\sMi $  attracted by its manufacturer will have similar structure.   
Additionally, the demand  $D_\sMe$ attracted by $\Me$ also contributes towards the revenue of $\V$  (as it  supplies raw material).
  In all,  the utility of the coalition $\V$ is given by,

\vspace{-4mm}
{\small\begin{eqnarray}
 U_{\sV} (p, q; \pe) &=&  \big( U_\sM  + D_\sMe \left(q-C_\sS \right)\F_\sMe
 - O_\sS - O_\sMi \big)\F_{\sV} \label{Eqn_Util_V}, \\
 U_\sM (p; \pe)  &=& D_\sMi\left(p-C_\sMi -C_\sS \right)\F_\sMi \label{eqn_Util_M} 
\end{eqnarray}}%
where  $C_\sS$ represents the raw material procurement cost  (per unit) and $C_\sMi$, $O_\sS$ and $O_\sMi$ have similar interpretations. The coalition $\V$ can choose to shut    in-house production (or it's   manufacturer $\Mi$) if it deems   advantageous,   represented by action $(p, q)$ with $p=n_o$, and hence the inclusion of the  flag   $\F_\sMi := \indc{p \ne n_o}$ in  \eqref{Eqn_Util_V}; alternatively it might find it beneficial to not operate at all, indicated by  
$ \F^c_{{\sV}} = 1-\indc{q  \ne n_o}$.

We need to choose an upper bound for the prices without loosing generality, as compact domains significantly simplify the analysis. Towards this, first observe that the demand attracted by any manufacturer (say $m$) gets zero, even after considering that  maximum possible fraction {\small$\varepsilon \dbar_{-m}$} is received from the other manufacturer (say $-m$), if 
$
p_m >  \nicefrac{ (\dbar_m + \varepsilon \dbar_{-m})} {\alpha_m}. 
$
Thus we set $\pmax =\nicefrac{(\dbar_\sMi + \varepsilon\dbar_\sMe)}{\alpha_\sMi}$  and ${\pe}_{_{mx}} =\nicefrac{(\dbar_\sMe + \varepsilon\dbar_\sMi)}{\alpha_\sMe}$ as the maximum prices respectively  for $\Mi$ and $\Me.$ 
\textit{We assume that if any agent is indifferent between the action~$a = n_o$ and an $a \ne n_o$,   the agent prefers to operate.} This consideration is inspired from the practical scenarios (see \cite{wadhwapartition}). We further assume   the following   as in \cite{wadhwapartition}, which ensures all the   agents find it `beneficial   to operate':
\begin{itemize}
    \item [{\bf A.1}] Assume the market potentials are sufficiently high, i.e., 

    \vspace{-4mm}
{\small\begin{eqnarray*}
    \dbar_\sMi  &\ge& \alpha_\sMi(C_\sS + C_\sMi) +2\sqrt{\alpha_\sMi(O_\sS + O_\sMi)}  \mbox{ and }\\
\dbar_\sMe  &\ge& \alpha_\sMe(C_\sS + C_\sMe) + 2\max\{\sqrt{2\alpha_\sMe (O_\sS + O_\sMi)},2\sqrt{ \alpha_\sMe O_\sMe}\} .    
    \end{eqnarray*}}

  \item [{\bf A.2}] Assume 
 $\varepsilon \le 2 C_{\sMe} + 4 \sqrt{\alpha_\sMe O_\sMe}$. 
\end{itemize}
Assumption  {{\bf A.1}} ensures that the market potentials of both the manufacturers are sufficiently high compared to production, procurement and the operating costs  (see \cite{wadhwapartition} for similar details). We will observe that    $\V$   finds it optimal to operate (i.e $\F^*_{\sV} =1$) under {{\bf A.1}},   which is important for meaningful analysis.
 Assumption  {{\bf A.2}} is required for some technical reasons in the proof of Theorem~\ref{thm_Fco_positive}  mentioned in \ref{sec_Appendix_AA}; besides, in general the operating and the production costs are significantly large
 and hence the assumption would automatically be satisfied (note here $\varepsilon\le 1$).

\subsection{Preliminary analysis and discussions}
  
\subsubsection { Best response of   $\Me$}
We  begin  by obtaining the best response of the follower, the out-house manufacturer $\Me$,  when the Stackelberg leader (coalition $\Vi$)  declares $  (p, q)$. In particular we consider the case with ${\cal F}_{\sV} =1$, or when $\Vi$ decides to operate. This  response of $\Me$  is governed by the  following optimization problem (observe from   \eqref{Eqn_Umj} that $D_\sMe$ depends upon $(p,q)$):

\vspace{-4mm}
{\small\begin{eqnarray}
U^{*}_\sMe(p,q) = \sup_{\pe\in \{n_o, [0, {\pe}_{_{mx}}] \}} \Big( D_\sMe (\pe- C_\sMe - q) - O_\sMe \Big) \indc{\pe\ne n_o}.\label{eqn_opt_util_out-house}
\end{eqnarray}}%
Such a problem is considered in \cite[Lemma 4]{wadhwapartition}.  By similar concavity arguments, the best response exists and equals:

\vspace{-4mm}
{\small\begin{eqnarray}
\pe^*(p,q) = \min\left\{ \frac{\dbar_\sMe + \varepsilon \alpha_\sMi p}{2 \alpha_\sMe} + \frac{C_\sMe + q}{2}, {\pe}_{_{mx}} \right\} \indc{q \le \theta(p)}   + n_o \indc{q > \theta(p)}, \mbox{ with, } \hspace{3mm} 
\label{Eqn_opt_policy_Mj} 
\end{eqnarray}
\begin{eqnarray}
 \theta(p) := \left\{ 
\begin{array}{lll}
\frac{\dbar_\sMe +\varepsilon\alpha_\sMi p -\alpha_\sMe C_\sMe - 2\sqrt{\alpha_\sMe O_\sMe}}{\alpha_\sMe} & \mbox{if } p < p_{sw}, \\
\frac{\dbar_\sMe + \varepsilon\dbar_\sMi - \alpha_\sMe C_\sMe}{\alpha_\sMe} - \frac{\alpha_\sMe O_\sMe}{\alpha_\sMe (\varepsilon\alpha_\sMi p - \varepsilon\dbar_\sMi)} & \mbox{else,}
\end{array}
\right.  \hspace{-85mm}   \label{Eqn_feasible_Regioin_Mj}\\ 
& p_{sw} := \frac{\dbar_\sMi}{\alpha_\sMi} + \frac{\sqrt{\alpha_\sMe O_\sMe}}{\varepsilon\alpha_\sMi}, \mbox{ and recall, } p_{mx} = \frac{\dbar_\sMi + \varepsilon \dbar_\sMe}{\alpha_\sMi}.  \hspace{2mm} \label{Eqn_psw}
\end{eqnarray}}%
In the above $p_{sw}$ represents a switching point---if the price of in-house $\Mi$ is above $p_{sw}$,
the optimal price of the out-house $\Me$ is clamped at the maximum possible value ${\pe}_{_{mx}}$.
Further,
  $\Me$ may not find it beneficial even to operate if the price $q$ quoted for raw materials is high (this happens when $q > \theta (p) $ in \eqref{Eqn_opt_policy_Mj}). Interestingly,  this  also depends upon the price $p$ quoted by the in-house manufacturer $\Mi$   towards the end-product. More interestingly $\Me$ can tolerate a larger $q$ if the price $p$ is higher (observe $\theta(p)$ increases with $p$)---a large part of loyal customer-base of in-house $\Mi$ can improve market opportunities for $\Me$ 
\ignore{(observe $\varepsilon g_i (p)$ seek products from $\Me$, and that $p \mapsto g_i (p)$ is increasing).} (observe from  \eqref{Eqn_Umj} that $\varepsilon \alpha_\sMi p$  fraction of customers seek products from $\Me$, and it is increasing in $p$).

\subsubsection {Choices of   coalition $\V$}    The  coalition $\V$ comprising of in-house manufacturer $\Mi$ and supplier $S$  has several advantages, as the vertical cooperation (VC) provides it  multiple choices as discussed below.  
 
 \noindent$ \bullet${\bf [Eliminate downstream competition (E$\ell$)]} The existence of in-house manufacturer in $\V$ provides it an option to operate in monopolistic manner when it is possible to attract a large fraction of `unhappy' loyal customers of the out-house $\Me$; this is possible probably when the manufacturers are substitutable to a good extent, i.e., if   $\varepsilon$ is large.  In  this case, it can completely eliminate out-house \manu 
    $\Me$ (by quoting exorbitantly large $q$)   and operate in the monopolistic manner in the downstream market with the combined market potential,  $\dbar_\sMi + \epsilon \dbar_\sMe$. 
    
 \noindent$ \bullet${\bf [Shut down the in-house (Sh)]} If either the market potential of the in-house manufacturer is low or when its reputation is not very good (when $\alpha_\sMi$ is more, its  customers are highly sensitive to price $p$), or when these factors of the out-house are significantly better, then $\V$ has an  option to completely shut its in-house production unit $\Mi$.  Such a choice can reduce the competition for out-house $\Me$ which in turn can become beneficial for $\V$---it may have an option to sell large amount of raw material (as market $D_\sMe$ attracted by $\Me$ can be large)  at good/optimal prices and without expending on operating costs of in-house. 

However it may not be beneficial to allow the out-house to operate in a monopolistic manner; like-wise it may not be beneficial to completely eliminate out-house $\Me$ unless the two production units are completely substitutable (in an ideal world with $\varepsilon=1$).  In such cases, there are other choices for $\Vi$ which we describe next and   are the focus of this paper.

 \noindent$ \bullet${\bf [Co-existence (Co)]}   
    In this scenario\textit{, both    $\Vi$   and out-house  $\Me$ operate; rather $\Vi$ allows both to operate}.  By virtue of this, it can charge sufficiently large (optimal) price $q$ for raw materials, which (probably) leaves few choices for $\Me$---the latter then has to quote larger prices $\pe$  to survive in the downstream market. This   facilitates  $\V$ to benefit from both the worlds, because of the `unhappy' loyal customers ($\varepsilon \alpha_\sMi p$) of $\Me$ that seek product from $\Mi$ as well as from the  high profits derived by selling the raw material to $\Me$ at large~$q$. Basically it chooses optimal $( p, q)$ that provides the best combined utility as a Stakelberg leader, while competing with the out-house manufacturer $\Me$ in the downstream market. 
 There are several sub-possibilities for $\Vi$ here: 
    \begin{itemize}
        \item {\bf [Operate \underline{b}oth \underline{p}rofitably(Bp)]} The coalition $\V$ quotes the price pair $(p,q)$ such that both production units derive non-zero profits. 
        
         \item {\bf [\underline{I}n-house operates at \underline{l}osses (I$\ell$)]} Alternatively $\V$ can operate it's in-house production unit at losses (by quoting large $p$), if that could fetch it a   larger revenue by just supplying to $\Me$; basically it might be beneficial not to allow the out-house to operate in  monopolistic manner in the downstream  market, by expending towards operating its in-house.  
         \item {\bf [\underline{O}ut-house forced to operate at \underline{p}ar (Op)]} In this case, the coalition $\V$ quotes the price $q$ to   $\Me$ such that this manufacturer operates but gets \textit{zero revenue}---this means that the coalition $\V$ quotes $q$ large, but sufficient to keep the out-house manufacturer operate at par.
   \end{itemize}

   \ignore{ In this case, it would chose a $(p,q)$ such that $\alpha_\sMi p > \dbar_\sMi + \alpha_\sMe \pe^*(p,q)  $ and hence such that $D_\sMi = 0$ and then optimizes the following:
    \begin{eqnarray*}
        \sup_{p, q, \ s.t., \ D_\sMi = 0, \ q \le \theta(p) }  \left(\dbar_\sMe + \varepsilon \alpha_\sMi p - \alpha_\sMe \pe^* (p,q) \right)(q- C_\sS) - O_\sMi - O_\sS \\
                \sup_{p, q, \ s.t., \ D_\sMi = 0, \ q \le \theta(p) }  \left(\dbar_\sMe + \varepsilon \alpha_\sMi p - \alpha_\sMe \left (e_1' + \frac{q}{2} +   \frac{p\varepsilon \alpha_\sMi}{2 \alpha_\sMe} \right ) \right)(q- C_\sS) - O_\sMi - O_\sS \\
                                \sup_{p, q, \ s.t., D_\sMi = 0, \ q \le \theta(p) }  \left(\dbar_\sMe + \frac{ \varepsilon \alpha_\sMi p}{2} - \alpha_\sMe \left (e_1' + \frac{q}{2}     \right ) \right)(q- C_\sS) - O_\sMi - O_\sS
    \end{eqnarray*}
    Thus the optimizers are $p^* = p_{max}$, the maximum possible value and $q^*$ satisfies:
    $$
    q^* = 
  \frac{  \dbar_\sMe + \frac{ \varepsilon \alpha_\sMi p_{max}}{2} - \alpha_\sMe  e_1'  + \frac{\alpha_\sMe C_\sS}{2} } {\alpha_\sMe}
    $$
    Thus the optimal utility  in this regime is given by:
    \begin{eqnarray*}
        \frac{ \left (\dbar_\sMe + \frac{ \varepsilon \alpha_\sMi p_{max}}{2} - \alpha_\sMe  e_1'       - \frac{\alpha_\sMe}{2} C_\sS  \right )^2 }{2 \alpha_\sMe }
    \end{eqnarray*}
    }

 \vspace{-1mm}
\medskip
\noindent
The objective is to identify $(p^*,q^*)$ that maximizes the coalition’s payoff across all the above regimes—--\textit{co-existence (in one of the three modes)}, \textit{shutdown of the in-house unit}, and \textit{elimination of the downstream competition}. As shown in Theorem~\ref{Thm_all_in_one}  (provided in later sections), the optimal choice is  driven by essentialness,   relative market strengths and other parameters. The results of the current paper significantly extend those provided in  conference paper \cite{wadhwapartition}, where the dominance of co-existence is established only  under restricted conditions like under essentialness assumption (near $\varepsilon \approx 1$). We now proceed towards detailed analysis and begin with the co-existence regime.

\vspace{-3mm}

\section{ Co-existence---VC allows both units to operate } \label{sec_Co}

The VC, $\Vi = \{S, \Mi\}$,  sets the two prices and the out-house manufacturer $\Me$ responds. 
Thus, it is a Stackelberg game  with    $\Vi$   as  the leader and $\Me$ as the follower. For any given $(p,q)$ the joint-price policy of $\Vi$, the optimal utility of out-house  $\Me$  is given by:

\vspace{-4mm}
{\small
\begin{eqnarray}
U^*_\sMe(p,q) = \big[ \left( \dbar\sMe - \alpha_\sMe \pe^{*}  + \varepsilon \alpha_\sMi p \right)^{+} \left( \pe^{*}  - C_\sMe - q \right)   
- O_\sMe \big ] \indc{\pe^{*} \ne n_o},
\label{Eqn_Opt_for_Mj}  
\end{eqnarray}}%
\ignore{
\begin{eqnarray}
U^*_\sMe(p,q) = \big (  \left( \dbar_\sMe \hspace{-1mm}- \alpha_\sMe \pe^{*} + \varepsilon\alpha_\sMi p \right)^{+} \hspace{-1mm}\left   (\pe^{*}  - C_\sMe - q\right)\nonumber\\ &&- O_\sMe \bigg) \indc{\pe^* \ne n_o} \hspace{3mm},   
\label{Eqn_Opt_for_Mj}
\end{eqnarray}}%
where $\pe^{*}=\pe^*(p,q)$, the optimizer of   $\Me$, is given by~\eqref{Eqn_opt_policy_Mj}.
In this section, we are interested in obtaining the optimal utility of $\Vi$ under co-existence, hence consider those $(p,q)$ for which $\pe^* \ne n_o$. Thus
the feasible region for co-existence  using \eqref{Eqn_opt_policy_Mj}-\eqref{Eqn_feasible_Regioin_Mj}  is given by:
\begin{eqnarray}\label{eqn_f_co}
{\cal F}_{co} &:= & \{ (p, q) \in (0, \infty)^2 :   q \le \theta(p)   \}.
\end{eqnarray}

For  $(p,q) \in {\cal F}_{co}$ or in the co-existence regime,  the utility of $\V$ is given by:

\vspace{-4mm}
{\small\begin{eqnarray}
    U_\sV (p, q) &=& \left ( \dbar_\sMi - \alpha_\sMi p  + \varepsilon \alpha_\sMe \pe^*(p, q) \right )^+ (p - C_\sMi - C_\sS) \nonumber \\
   &+& \left ( \dbar_\sMe + \alpha_\sMi \varepsilon p  -  \alpha_\sMe \pe^*(p, q) \right )^+ (q- C_\sS)  - O_\sMi - O_\sS , \hspace{5 mm} \label{eqn_util_co-exist_given_pq}
\end{eqnarray}}%
and the aim in this section is to optimize the above  over $(p,q) \in {\cal F}_{co}$ of~\eqref{eqn_f_co}. Towards this, first
consider the following `unconstrained' optimization problem, which resembles \eqref{eqn_util_co-exist_given_pq} but for $(\cdot)^+$ operators, and   when   $\pe^*(p,q)  <  {\pe}_{_{mx}}$:

\vspace{-3mm}
{\small\begin{eqnarray}
    && \sup_{p,q} \ \ {\cal U} (p,q) \quad \mbox{where, }\nonumber  \\
    {\cal U} (p, q) &=& \left( \dbar_\sMi - \alpha_\sMi p + \varepsilon \alpha_\sMe \left( \frac{\dbar_\sMe + \varepsilon \alpha_\sMi p}{2 \alpha_\sMe} + \frac{C_\sMe + q}{2} \right) \right) (p - C_\sMi - C_\sS)\nonumber \\
    &&\hspace{-23mm}+ \left( \dbar_\sMe + \alpha_\sMi \varepsilon p - \alpha_\sMe \left( \frac{\dbar_\sMe + \varepsilon \alpha_\sMi p}{2 \alpha_\sMe} + \frac{C_\sMe + q}{2} \right) \right) (q - C_\sS) - O_\sMi - O_\sS. \hspace{2mm}\label{eqn_util_co-exist_uc}
  \end{eqnarray}}%
\spcialComment{
The proof for the existence of the  optimizer of this unconstrained optimization problem is given in  \ref{sec_Appendix_AA} --- the proof is constructed using `sections-wise' (obtained by keeping $p$ or $q$ fixed) convexity and concavity arguments.
Let $(p^*_{co}, q^*_{co} )$ represent its optimizer (derived  in the proof of Theorem \ref{thm_Fco_positive} in steps g(1)--g(3), which is provided in \ref{sec_Appendix_AA}), which is given by:

{\color{blue}
The existence of an optimizer for this unconstrained optimization problem
 is established in ~\ref{sec_Appendix_AA}, for a special case for which it also coincides with the global optimizer of $U_\sV$. In the complementary case,  we further identify  some boundary regions of a subset of ${\cal F}_{co}$,   which  contain the global optimizer  of $U_\sV$.

 is in the interior of certain  region ${\cal F}_{co}^+$  (defined later in \eqref{Eqn_Fco_plus}). The proof is constructed using \emph{section-wise}
analysis: for fixed $p$ (respectively, fixed $q$), the objective function
is concave in $q$ (respectively, in $p$), which guarantees the existence
of a maximizer along each section. Combining these section-wise optimizers
yields a stationary point of the joint problem.
Let $(p^{*}_{co}, q^{*}_{co})$ denote this  stationary point (candidate optimal pair) obtained in
steps~g(1)--g(3) of the proof of Theorem~\ref{thm_Fco_positive} in
~\ref{sec_Appendix_AA} which is given by:}
}%
It is easier to analyze the above unconstrained problem.
Now define the following pair of prices,

\vspace{-3mm}
{\small\begin{eqnarray}
    p^{*}_{co} &=& \hspace{-2mm} -\frac{2w_3w_4 - w_2w_5}{4w_1w_3 - w_2^2}, \quad
    q^{*}_{co} = -\frac{w_2 p^{*}_{co} + w_5}{2w_3},   \mbox{ with,} \label{Eqn_pco_qco}\\
    w_1 &=& \hspace{-2mm}-\alpha_\sMi \left(1-\frac{\varepsilon^2}{2} \right), 
    w_2 = \frac{\varepsilon\left(\alpha_\sMi + \alpha_\sMe\right)}{2}, 
    w_3 = -\frac{\alpha_\sMe}{2}, \nonumber\\
   w_4 &=&\hspace{-2mm} \frac{\scriptstyle 2\dbar_\sMi + \varepsilon\dbar_\sMe + \varepsilon\alpha_\sMe C_\sMe - \varepsilon\alpha_\sMi C_\sS + 2\alpha_\sMi\left(1-\frac{\varepsilon^2}{2}\right)(C_\sMi + C_\sS)}{2}, \nonumber\\
    w_5 &=& -\frac{\varepsilon\alpha_\sMe \left(C_\sMi + C_\sS\right)}{2} + \frac{\left(\dbar_\sMe - \alpha_\sMe C_\sMe + \alpha_\sMe C_\sS\right)}{2},  \hspace{4mm}
    \label{Eqn_ws}
\end{eqnarray}}%
 which become the optimizers of \eqref{eqn_util_co-exist_uc}, when there exists one (proved in  \ref{sec_Appendix_AA})---the above  pair $(p_{co}^*, q_{co}^*)$ 
   will also become the optimizer for the original co-existence objective function $U_\sV$  given in \eqref{eqn_util_co-exist_given_pq}, under the same existence condition---this pair and the equivalence is obtained in
steps~(g.1)-(g.3) of the proof of Theorem~\ref{thm_Fco_positive} provided in
~\ref{sec_Appendix_AA}. For now, we continue with discussing the other details related to the `optimal co-existence policy' or the optimizer of \eqref{eqn_util_co-exist_given_pq} in~${\cal F}_{co}$.

It is clear that, \eqref{eqn_util_co-exist_given_pq} is different from the `unconstrained' function  \eqref{eqn_util_co-exist_uc} in some sub-regimes of the co-existence regime ${\cal F}_{co}$. So in the quest towards the optimal co-existence policy, one also needs to find the optimizer(s) in the sub-regimes where the two differ.  In all, we will   partition   ${\cal F}_{co}$ into many sub-regimes, such  that the objective functions  \eqref{eqn_util_co-exist_given_pq} and   \eqref{eqn_util_co-exist_uc} match
 in the first sub-regime, while they differ in the remaining   sub-regimes (see Figure \ref{fig:feasible region}).   We now consider them one after the other.  Interestingly three of these sub-regimes align with our initial discussion on the choices of $\Vi$, however, an extra boundary line $\{p=\pmax\} \cap {\cal F}_{co}$  of ${\cal F}_{co}$ also becomes important and requires separate attention. 
 
 We begin with operate both profitably (where functions \eqref{eqn_util_co-exist_given_pq} and   \eqref{eqn_util_co-exist_uc} match) in the immediate next, while the sub-regime   with in-house at loss is analyzed in subsection      \ref{sec_in-house_at_loss}; operate at maximum price  (i.e.$\{p = \pmax\}$) is provided in subsection \ref{sec_max_price}, and the
 sub-regimes  $\{p_e^* = {\pe}_{_{mx}}\} $ and operate out-house at par (i.e., $\{ q= \theta(p)\}$) are considered together in subsection \ref{sec_par}.


\begin{figure}[H]
\centering
\vspace{-1mm}
\includegraphics[scale=0.2]{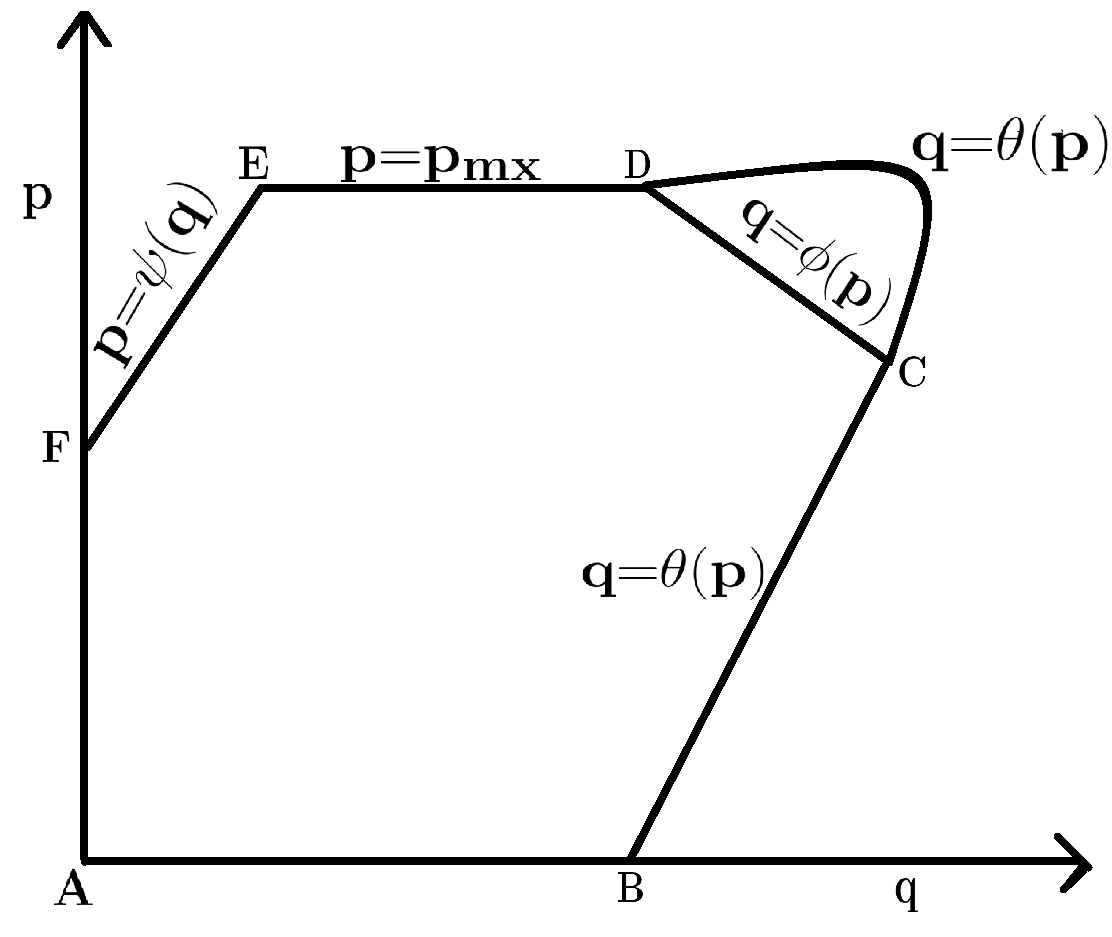}
\caption{Representative feasible region, ${\cal F}^+_{co}$.}
\label{fig:feasible region}
\vspace{-4mm}
\end{figure}

\subsection{ Operate \underline{B}oth \underline{p}rofitably   (Bp regime)} \label{subsec_bp}
The\textit{ BP regime is the set/region of prices $(p,q)$ where both the manufacturers derive strictly positive utilities}.
It is interesting to identify the conditions under which the coalition $\V$ finds it beneficial to operate in Bp regime  and we answer this partially in this subsection. 
   As we will see, such a regime consists of two sub-regimes: a) the interior and some boundaries of a certain region, identified in this subsection; and b) a certain part of the boundary with price $p = \pmax$, discussed in subsection \ref{sec_max_price}.  
 
 We  now  focus on  the first sub-regime which includes the complete interior of the Bp regime;  this sub-regime is identified using the following steps (as we will see, such a choice also ensures  \eqref{eqn_util_co-exist_given_pq} 
equals \eqref{eqn_util_co-exist_uc}):

$\bullet$  includes the pair of prices $(p,q)$, for which the optimal price $\pe^*(p,q)$ of the out-house manufacturer is strictly less than ${\pe}_{_{mx}} = \nicefrac{(\dbar_\sMe + \varepsilon \dbar_\sMi)}{\alpha_\sMe}$; such a regime from \eqref{Eqn_opt_policy_Mj}-\eqref{Eqn_feasible_Regioin_Mj} is given by  $\{ (p,q) :  q < \phi (p) \}$ where $\phi (\cdot) $ is defined below, 
  \begin{eqnarray}\label{Eqn_phi_st_p}
   \phi (p)  := \frac{\dbar_\sMe + 2\varepsilon\dbar_\sMi - \varepsilon\alpha_\sMi p - \alpha_\sMe C_\sMe}{\alpha_\sMe}; 
    \end{eqnarray}
    the boundary of  such a regime is the  straight line,  
    $$
    \mathbb{L}_1 :=  \{ q =  \phi (p)\}; 
    $$ 
this condition ensures $\pe^*$ in \eqref{eqn_util_co-exist_given_pq}  matches with its counterpart in \eqref{eqn_util_co-exist_uc};

$\bullet$ includes the pair of prices $(p,q)$, for which  the $\Vi$ coalition derives strict positive utility from in-house production unit also;  this is the sub-regime where   $\dbar_\sMi + \varepsilon\alpha_\sMe \pe^{*}(p,q) - \alpha_\sMi p >  0$ (see \eqref{eqn_demand_vc_coal} and \eqref{eqn_util_co-exist_given_pq}); such a regime, further  within  $\{  q < \phi (p) \}$, is given by $\{    q < \phi (p) \mbox{ and } p < \psi (q) \}$, with $\psi (\cdot)$ defined below: 
  \begin{eqnarray}
  \label{Eqn_psi_q}
  \psi (q)  :=  \frac{ \left ( 2 \dbar_\sMi + \varepsilon  \dbar_\sMe + \varepsilon \alpha_\sMe  (C_\sMe+q ) \right )}{(2-\varepsilon^2)\alpha_\sMi} ;   
  \end{eqnarray}
observe here that  the   straight  line,     
$$
\mathbb{L}_2 := \{p = \psi (q)  \},
$$bounds the regime  of interest only when  
    it is also bounded by $\mathbb{L}_1$ and these constraints ensure both $(\cdot)^+$ terms in \eqref{eqn_util_co-exist_given_pq} are positive and hence match with the correspnding terms in \eqref{eqn_util_co-exist_uc};

$\bullet$  the pair of prices $(p,q)$  which ensure  co-existence, along with out-house operating, belong to    $\{q \le \theta(p)\}$ with $\theta(\cdot)$ as in \eqref{Eqn_feasible_Regioin_Mj}; when $p \ge p_{sw}$ in \eqref{Eqn_feasible_Regioin_Mj} by simple computations  one can show that\footnote{By directly substituting the terms, for any $p > p_{sw}$, we have,  $\theta(p)-\phi(p) > \theta(p_{sw}) -\phi(p_{sw}) = 0 $.}  $\phi(p ) \le \theta(p)$ and so this constraint is already satisfied by bounding with $\mathbb{L}_1$; so it is sufficient to  ensure bounding by $\theta(p)$ for   $p < p_{sw}$, which is provided in the first row of~\eqref{Eqn_feasible_Regioin_Mj}; hence, in all, it is sufficient to   bound by the following additional straight line    
within the regime bounded by the lines $\mathbb{L}_1$ and $\mathbb{L}_2$,

       {\small{
    $$
    \mathbb{L}_3 = \left \{q =    
      \frac{\dbar_\sMe + \varepsilon\alpha_\sMi p -\alpha_\sMe C_\sMe - 2\sqrt{\alpha_\sMe O_\sMe}}{\alpha_\sMe} \right \};
    $$}}

$\bullet$  and finally bounded by  the horizontal line of maximum price, 
$$
\mathbb{L}_4 =\{ p = \pmax \}.$$

 Such a sub-regime (actually its closure),  represented by ${\cal F}^+_{co}$,  is the regime in the positive quadrant,  bounded by all the lines $\mathbb{L}_1, \mathbb{L}_2, \mathbb{L}_3$ and $\mathbb{L}_4$ (see polygon ABCDEF in  Figure \ref{fig:feasible region} for one representative scenario). To summarize:
\begin{eqnarray} 
    {\cal F}^+_{co} =  \big \{ (p, q) \in  [0,\infty)^2  :    p \le \min \{\pmax, \psi(q)  \} \mbox{ and }  
    q \le \min \left \{\theta (p), \phi (p) \right  \} \big \}. \hspace{2mm}  \label{Eqn_Fco_plus}
\end{eqnarray}
Now the BP regime, where both the manufacturers obtain strictly positive profits, denoted by 
${\cal F}_{_{Bp}}$, is a subset of the above region. More precisely, ${\cal F}_{_{Bp}} = \{ (p, q) \in {\cal F}^+_{co}: q < \theta(p), p < \psi(q) \} $---because with $p = \psi(q)$ or  $q = \theta(p)$  we respectively have $U_{\sM}(p,q) = 0$ in \eqref{eqn_Util_M} or $U^*_{\sMe}(p,q) = 0$ in \eqref{eqn_opt_util_out-house}. 
Thus the   ${\cal F}_{Bp}$ region spans  the interior of ${\cal F}^+_{co}$ and  the boundary lines:
\begin{eqnarray}\label{eqn_f_bp}
 {\cal F}_{Bp} = Interior ({\cal F}^+_{co}) \cup \left (
{\cal F}^+_{co}\cap \left (\mathbb{L}_4 \cup \mathbb{L}_1\right ) \right ).   
\end{eqnarray}  
Hence ${\cal F}_{Bp}$ is not a closed set and may not have an optimizer and therefore it is more convenient to analyze the closed region ${\cal F}^+_{co}$  \eqref{Eqn_Fco_plus}; we proceed with the same for now, and later discuss the possible optimality of   Bp regime.    %

Further ${\cal F}^+_{co}$ is also a sub-region of  the co-existence regime---clearly ${\cal F}^+_{co} \subseteq {\cal F}_{co} $---for example, when non-empty, the I$\ell$ regime $\{p > \psi(q)\} \subset {\cal F}_{co}\setminus {\cal F}^+_{co}$. Hence the optimal utility in co-existence regime is given by:
\begin{eqnarray}\label{eqn_opt_v_coexxx}
     U_{\sV, co}^* = \max \left \{  \max_{(p,q) \in {\cal F}_{co}^+ } U_\sV (p,q),   \max_{(p,q) \in {\cal F}_{co} \setminus {\cal F}_{co}^+ } U_\sV (p,q)     \right \}.
\end{eqnarray}
Thus with the dual purpose of analyzing the co-existence and the BP regimes, we 
  begin with analyzing the first term of \eqref{eqn_opt_v_coexxx}: 
\begin{theorem} 
\label{thm_Fco_positive}
Assume {{\bf A.1}}-{{\bf A.2}}. 
(i)  If $(p^*_{co}, q^*_{co})$  of \eqref{Eqn_pco_qco} is in the interior of ${\cal F}_{co}^+$ then, 
\begin{eqnarray}
\label{eqn_vc_opt_f_co_+}
\max_{(p,q) \in {\cal F}_{co}^+ }U_\sV(p,q) = U_\sV(p^*_{co}, q^*_{co}).
  \end{eqnarray} 
(ii) If  $(p^*_{co}, q^*_{co})$ is not in the interior, then the   optimal  utility across ${\cal F}_{co}^+$ is at one of the non-empty boundaries, excluding the $\{q=0\}$ and $\{p=0\}$ lines:
\begin{eqnarray}
\max_{(p,q) \in {\cal F}_{co}^+ } U_\sV (p,q) =  \max_{l \in \{1, 2, 3, 4\}}  \left \{ \max_{ (p,q) \in {\cal F}^+_{co} \cap \mathbb{L}_l,   ~p > 0, q > 0    }   U_\sV (p,q) \right \}.  
  \end{eqnarray} 
In the above, by convention, the maximum of an empty set is set to zero. 
\end{theorem}

{\bf Proof} is provided in \ref{Proof_thm_1}. \eop

\medskip
Thus by part (i), if   $(p^*_{co}, q^*_{co})$  of \eqref{Eqn_pco_qco} is in the interior of  ${\cal F}_{co}^+$, 
this pair in the Bp regime has the potential to become a global optimizer of the co-existence region ${\cal F}_{co}$; and then that of the overall problem.  From \eqref{eqn_f_bp}, one may also have a Bp optimizer on lines ${\mathbb L}_1$ (i.e., with $q = \phi(p)$)  or on ${\mathbb L}_4$ (i.e., with  $p = \pmax$); we will show in later subsections that the former is not possible but one can have an optimizer in BP regime with $p = \pmax$ in subsection~\ref{sec_max_price}. 
Eventually in section~\ref{sec_overall_comp},  we derive some conditions
for global optimality of a pair where both manufacturers operate profitably (i.e., optimality of Bp regime), by  building upon the results of Theorem \ref{thm_Fco_positive} and that of subsection~\ref{sec_max_price}.  We also provide more insights on this aspect along with others in section \ref{sec_num} using several numerical examples.

Next, we analyze the 
remaining co-existence regions at which at least one of the manufacturers derive zero utility (but still operates at marginal costs)---these comprise of ${\cal F}_{co} \setminus {\cal F}_{co}^+ $ and the boundary lines of $ {\cal F}_{co}^+ $---as already mentioned, there is however a small  exception,   one may find an optimizer in Bp regime  while analyzing   $\mathbb{L}_4 =\{ p = \pmax \}$.

\vspace{2mm}
 \subsection  {  \underline{I}n-house   operates at \underline{l}oss (I$\ell$ regime)}
\label{sec_in-house_at_loss}

 This sub-region corresponds to the case in which the coalition $\V$ allows its in-house manufacturer $\Mi$ to operate while incurring losses. In this scenario, the coalition $\V$ strategically quotes a very high price $p$, leading to zero demand for its in-house production. Such a strategy can still be advantageous to $\V$, as it may enhance the overall market potential captured via the out-house manufacturer $\Me$ (I$\ell$ is the best among  all possible regimes  in Figure \ref{fig_overall} of section \ref{sec_num},   depicted by   green regions).

We denote this sub-region by ${\cal F}_{_{I\ell}} $, and define it  using \eqref{Eqn_psi_q} (see also \eqref{eqn_f_co}):
\begin{eqnarray*}
{\cal F}_{_{I\ell}}   
&=&  \left \{ (p,q) \in {\cal F}_{co} :  p \le  \pmax, \ p > \psi(q), \ q \le  \theta(p) \right \} \\
&=&  \left \{ (p,q) \in {\cal F}_{co} :  p \le  \pmax, \ q  \le   \min \{  \psi^{-1}(p), \    \theta(p) \} \right \}
\end{eqnarray*}

From~\eqref{Eqn_psi_q}, $\psi(\cdot)$ is an increasing function in $q$, immediately implying,  \textit{I$\ell$ regime is non-empty only if $\psi(0) < \pmax$}.  
Within this sub-region, the co-existence utility $U_\sV$ in~\eqref{eqn_util_co-exist_given_pq},     simplifies to  the following function (see~\eqref{Eqn_opt_policy_Mj}):
\begin{eqnarray}
U_{_{I\ell}} (p,q) 
:=  
\left(\dbar_\sMe - \alpha_\sMe \pe^{*}(p,q) + \varepsilon\alpha_\sMi p \right)\left(q- C_\sS \right) 
- O_\sMi - O_\sS, 
\  \forall (p,q) \in {\cal F}_{_{I\ell}} . \hspace{1mm}
\label{Eqn_Il_util}
\end{eqnarray}
This is basically the regime, where  $\V$ keeps the presence of its in-house alive with negligible production and profits, whose presence molds the demand of the out-house in a much more profitable manner  for the supplier component of~$\V$  (and this happens, for example, when I$\ell$ is the optimal regime).  

Let the optimal utility under this regime be denoted by $U_{_{I\ell}}^{*}$.  
Using~\eqref{Eqn_opt_policy_Mj} and \eqref{Eqn_Il_util}, it immediately follows that the optimizer of $U_{_{I\ell}} $ (which is strictly increasing in $p$ for any fixed~$q$) is given by $(\pmax,q^*_{_{I\ell}})$, where $q^*_{_{I\ell}}$ is the solution to the following optimization problem (recall  $\psi(\cdot)$ is increasing),
\begin{eqnarray}
U_{_{I\ell}}^{*} 
&=& 
\max_{q : (\pmax,q)\in {\cal F}_{_{I\ell}}  } U_{_{I\ell}}(\pmax, q) 
= \max_{q \le  r_{_{I\ell}}} U_{_{I\ell}}(\pmax, q),\nonumber \\
\text{where} \quad 
 r_{_{I\ell}} &:=& \min \left \{  \max\{0, \psi^{-1} (\pmax)\}, \ \theta(\pmax) \right \}.\label{eqn_pi_il}
\end{eqnarray}
Thus the optimal in I$\ell$-regime (if non-empty) is along ${\mathbb L}_4 = \{p = \pmax\}$ line.  As we will see in the next sub-section the Mp-Bp regime along the same line starts after I$\ell$-regime, when both the regimes are non-empty  (clearly, we have 
$ r_{_{I\ell}} \le l_{_{Mp}} $, 
the left boundary of Mp regime given in \eqref{Eqn_psi_inv_pmax}).

Further, using~\eqref{Eqn_phi_st_p}, for any $q \le  r_{_{_{I\ell}}}$ (with the convention that $[a,b]=\emptyset$ when $a>b$), we have:

\vspace{-3mm}
{\small\begin{eqnarray*}
U_{_{_{I\ell}}}(\pmax, q) 
&=& \hspace{-2mm} 
\Bigg( 
\frac{\dbar_\sMe (1 + \varepsilon^2) + \varepsilon \dbar_\sMi - \alpha_\sMe (q + C_\sMe)}{2} \I_1(q)     
+ \varepsilon^2 \dbar_\sMe \I_2(q)
\Bigg) (q - C_\sS)  \\
&-&  O_\sMi - O_\sS, 
\quad \text{where} \\
\I_1(q) &:=& \indc{q \le \phi(\pmax)}, 
\qquad 
\I_2(q) := \indc{q \in [\phi(\pmax),  r_{_{_{I\ell}}}]}.
\end{eqnarray*}

The first term can be optimized by ignoring the indicator $\I_1(q)$ to obtain a candidate optimizer ${\tilde q}^*$ (given below).  
The overall optimizer for this sub-case then becomes (first term in $U_{_{_{I\ell}}}$ is concave, while  the second is linear in $q$):
\begin{eqnarray}
q^*_{_{I\ell}}
&=&  {\tilde q}^* \indc{  {\tilde q}^* \le \min\{ r_{_{_{I\ell}}}, \phi(\pmax)\}} 
+ 
r_{_{_{I\ell}}} \indc{  {\tilde q}^* >  \min\{ r_{_{_{I\ell}}}, \phi(\pmax)\}},
\mbox{  
where } \nonumber  \\
\label{eqn_tilde_q_star}
{\tilde q}^* 
&=& \left(
\frac{\dbar_\sMe(1+\varepsilon^2) + \varepsilon\dbar_\sMi - \alpha_\sMe C_\sMe}{2\alpha_\sMe} 
+ \frac{C_\sS}{2}
\right).
\end{eqnarray}
By direct substitution (when $\varepsilon > 0$, the inequality '$a$' is strict, in the below), 

\vspace{-3mm}
{\small\begin{eqnarray}
  \label{Eqn_psi_inv_less_phi_at_pmax}  
  \psi^{-1} (\pmax) 
 = \frac{ (1-\varepsilon^2) \dbar_\sMe - \varepsilon \dbar_\sM} {\alpha_\sMe} - C_\sMe 
 \stackrel{a}{\le} \phi (\pmax) 
 =  \frac{ (1-\varepsilon^2) \dbar_\sMe + \varepsilon \dbar_\sM} {\alpha_\sMe} - C_\sMe,
\end{eqnarray}}%
and since 
 $r_{_{_{I\ell}}} \le \max\{0, \psi^{-1} (\pmax)\}$
  we have 
$
\min\{ r_{_{_{I\ell}}}, \phi(\pmax)\} =  r_{_{_{I\ell}}},
$
 and so,
\begin{eqnarray}\label{eqn_opt_q_loss}
q^*_{_{I\ell}}
= {\tilde q}^* \indc{  {\tilde q}^* \le  r_{_{_{I\ell}}} } 
+  r_{_{_{I\ell}}} \indc{  {\tilde q}^* >  r_{_{_{I\ell}}} }.
\end{eqnarray}
Hence the optimal pair in this sub-regime (when non-empty, i.e., when $\psi(0) < \pmax$)  is $(\pmax,q^*_{_{I\ell}})$  and the corresponding (sub) optimal utility   equals:

\vspace{-3mm}
{\small
\scalebox{0.82}{
\begin{minipage}{\linewidth}
\begin{eqnarray}
\hspace{-8mm}
U_{_{I\ell}}^* 
\label{Eqn_opt_loss_util} 
\label{Eqn_operate_at_loss_inhouse}
=     
\frac{ \left(\dbar_\sMe(1+\varepsilon^2)+ \varepsilon\dbar_\sMi - \alpha_\sMe(C_\sMe + C_\sS)  \right)^2}{8\alpha_\sMe} \indc{  {\tilde q}^* \le   r_{_{_{I\ell}}} }    
   + \varepsilon^2\dbar_\sMe( r_{_{_{I\ell}}} - C_\sS)\indc{  {\tilde q}^* >    r_{_{_{I\ell}}}  }
 -  O_\sMi - O_\sS. \hspace{1mm}
\label{Eqn_Ustar_Il}
\end{eqnarray}
\end{minipage}
}
}

The above indicates the possibility of an interesting optimizer for $\V$---when $r_{_{_{I\ell}}}  = \theta(\pmax) $ and  ${\tilde q}^* \ge  r_{_{_{I\ell}}}$ in \eqref{eqn_opt_q_loss}, then $\V$ might find the optimizer in the combined I$\ell$-Op regime, where the in-house incurs losses and the out-house operates at par simultaneously.  However, 
 we could not observe the optimality of such a pair in the exhaustive numerical examples of  section \ref{sec_num}; we did not observe such an optimality even in the other exhaustive set of numerical examples not included in the paper.

 }


  \subsection{ Operate at \underline{M}aximum \underline{p}rice (Mp regime)}\label{sec_max_price}
  We now consider the sub-regime located along the boundary $\mathbb{L}_4$ of ${\cal F}_{co}^+$ and find the corresponding sub-optimizer.  As mentioned in subsection \ref{subsec_bp}, this optimizer can also be in the BP regime, where both the \manu s derive strictly positive utility. 
Define the following using \eqref{Eqn_feasible_Regioin_Mj},\eqref{Eqn_phi_st_p}-\eqref{Eqn_psi_q}, to reflect the boundary points corresponding to $p=\pmax$:

 \vspace{-3mm}
{\small \begin{eqnarray}
l_{_{Mp}} 
&:=& 
\max\left \{ \psi^{-1} (\pmax), 0 \right \} 
=  
 \max\left \{ 
\frac{ -\varepsilon \dbar_\sMi + (1-\varepsilon^2) \dbar_\sMe - \alpha_\sMe C_\sMe }{ \alpha_\sMe}, \ 0 
\right \},  
\hspace{5mm}\label{Eqn_psi_inv_pmax} \\ 
r_{_{Mp}} 
&:=& 
\scalebox{1}{$\bar{q}(\pmax)$} 
\mbox {, where, }  \bar q (p) := \min \{\theta(p), \phi(p) \}. 
%
\label{Eqn_r_mx}
\end{eqnarray}}
It is  clear that Mp regime corresponds to $\left \{ (p,q) : p = \pmax,  q \in  [\, l_{_{Mp}}, r_{_{Mp}}  \,] \right \}$ and  
  is non-empty only when $l_{_{Mp}} < r_{_{Mp}} $. Using \eqref{Eqn_feasible_Regioin_Mj}-\eqref{Eqn_psw},  $r_{_{Mp}}$ equals:
\begin{equation}
\label{eqn_r_max_simplified}
r_{_{Mp}}  
:= 
\scalebox{1}{$\bar{q}(\pmax)$} 
= 
\scalebox{1}{$\left \{ 
\begin{array}{lll}
\frac{\dbar_\sMe (1+\varepsilon^2) + \varepsilon \dbar_\sMi -\alpha_\sMe C_\sMe - 2\sqrt{\alpha_\sMe O_\sMe}}{\alpha_\sMe},  
& \text{if } \varepsilon^2 \dbar_\sMe  < \sqrt{\alpha_\sMe O_\sMe}, \\[8pt]
\frac{\dbar_\sMe (1-\varepsilon^2)+ \varepsilon\dbar_\sMi - \alpha_\sMe C_\sMe}{\alpha_\sMe}, 
& \text{otherwise.}
\end{array}
\right .$}   
\end{equation}
When $\pmax \ge  p_{sw}$ (i.e., when $\varepsilon^2 \dbar_\sMe  \ge  \sqrt{\alpha_\sMe O_\sMe}$ in the above), it 
is immediate that $l_{_{Mp}} < r_{_{Mp}} $, however the same is not always guaranteed for $\pmax < p_{sw}$.
  For non-empty Mp regime, the corresponding sub-optimizer (for $p = \pmax$) is obtained as in the proof of Theorem~\ref{thm_Fco_positive} and is given by (function $h$ is defined in \eqref{eqn_q_star_p} of \ref{sec_Appendix_AA}):
\begin{eqnarray}
U^*_{_{Mp}} 
:= 
U_\sV (\pmax, q^*(\pmax)), 
 \text{ where }  
q^*(\pmax) 
= \max \{\, l_{_{Mp}}, \min\{ r_{_{Mp}} , h(\pmax) \} \}. 
 \hspace{2 mm}\label{Eqn_opt_max_util}
\end{eqnarray}
 Clearly,  when $q^{*}(\pmax) = h(\pmax)$ or when $q^{*}(\pmax) = r_{_{Mp}} = \phi(\pmax)$, this sub-optimizer is in Bp regime. If the overall optimal pair for $\V$ is at such points, we say the optimal regime is \textit{Mp-Bp} regime. Likewise, the sub-optimizer is in the OP regime if $q^{*}(\pmax) = r_{_{Mp}} = \theta(\pmax)$ (recall in Op regime, the out-house derives zero utility and hence is not in Bp regime).

We are now left with two additional sub-regimes:  
one in which the out-house’s optimal price equals ${\pe}_{_{mx}}$ in~\eqref{Eqn_opt_policy_Mj},  
and another in which the out-house $\Me$ is forced to operate at break-even.  
In the latter case, the optimal utility of $\Me$ is exactly zero.  We first discuss the former case. 
\subsection*{Out-house operates at maximum price}

From \eqref{Eqn_opt_policy_Mj} and \eqref{eqn_util_co-exist_given_pq},  when the optimal price of the out-house saturates at ${\pe}_{_{mx}}$, then   the utility function $U_\sV$ of  coalition $\V$ modifies to the following:
    \begin{eqnarray*}
      U_{_{St}}(p,q)  
   \hspace{-2mm}   &:= & \hspace{-2mm}\left(\dbar_\sMi(1+\varepsilon^2) + \varepsilon\dbar_\sMe - \alpha_\sMi p \right)\left(p- C_\sMi - C_\sS\right) + \varepsilon\left(\alpha_\sMi p -\dbar_\sMi \right)\left(q- C_\sS \right) \\ && \  \  -  O_\sMi - O_\sS.
    \end{eqnarray*}
    The set of $(p,q) \in {\cal F}_{co}$ where such a saturation occurs 
    is given by 
    (see \eqref{Eqn_feasible_Regioin_Mj}~\eqref{Eqn_phi_st_p}):

    \vspace{-4mm}
    {\small\begin{eqnarray}
        {\cal F}_{_{St}} &=& \left \{ (p,q) \in {\cal F}_{co} :  \pe^* (p,q) =\frac{\dbar_\sMe + \varepsilon 
    \dbar_\sMi}{\alpha_\sMe}  \right \} \nonumber \\ &=&     \left \{ (p,q) \in {\cal F}_{co} :  q > \phi(p)  \mbox{ and }  q \le \theta(p) \right \} \label{Eqn_Fco_Saturate} .
    \end{eqnarray}}
Comparing section wise,  once again across $q$, one can easily verify that 
  \begin{eqnarray}
   U_{_{St}}(p, q) \le U_{_{St}} (p, \theta(p) ) \mbox{ for all } p \mbox{ such that  }  (p, \theta(p) ) \in {\cal F}_{_{St}}.  \label{Eqn_Fst}
  \end{eqnarray}
   Further,  it is not difficult to see that if there exists a $p$ such that $(p,q) \in {\cal F}_{_{St}}$, then $(p,\theta(p)) \in {\cal F}_{_{St}}$. 
  Thus the optimal co-existence utility in ${\cal F}_{_{St}}$ is given by the optimal across all points in which the out-house operates at par and at saturation, i.e., in $  {\cal F}_{_{St}} \cap \{ (p, \theta(p) ) \}.$  As a result, towards finding the global optimization point, it is sufficient to consider the optimal across Op regime (and compare with others).  This is discussed in the immediate next. \textit{Before proceeding, we would also like to note here that an optimal pair in BP regime does not exist with out-house operating at maximum price.   This also implies Mp-Bp regime is optimal (if at all) only when $q^{*}_{\pmax} = h(\pmax)$ in \eqref{Eqn_opt_max_util}.} 

\subsection  {\underline{O}ut-house operates at \underline{p}ar  (Op regime)}\label{sec_par} 
This is the regime, where the out-house  operates, however is forced to do so at par. 
As just discussed in \eqref{Eqn_Fst},  
towards finding the optimal among  the  Op and the out-house-price-saturation  regions, it is sufficient to consider optimal across the  Op regime---the relevant optimization problem is:
  \begin{eqnarray*}   
 U_{_{Op}}^{*} := \max_{ (p, q) \in {\cal F}_{co} : q = \theta(p) }  U_{\sV} (p,q).
  \end{eqnarray*}
Towards solving the above, we need to proceed separately depending upon  the sign of  $ (p_{sw}-p)  $    (see \ref{Eqn_feasible_Regioin_Mj}). The following optimization problem is relevant for $p \le p_{sw}$ 
\begin{eqnarray}
    \max_{ p \le \min\{p_{sw}, \pmax\} }
    &\Bigg( \left(\dbar_\sMi + \varepsilon\dbar_\sMe - \alpha_\sMi(1-\varepsilon^2)p - \varepsilon\sqrt{\alpha_\sMe O_\sMe}\right)\left(p - C_\sMi - C_\sS\right) \nonumber\\
    &\hspace{-15mm}+ 
    \frac{\sqrt{\alpha_\sMe O_\sMe}\left(\dbar_\sMe + \varepsilon\alpha_\sMi p - \alpha_\sMe C_\sMe - \alpha_\sMe C_\sS - 2\sqrt{\alpha_\sMe O_\sMe}\right)}{\alpha_\sMe}  - O_\sMi - O_\sS
    \Bigg) \label{eqn_u1}.
\end{eqnarray}

The optimizer of the above by strict concavity is  at $p^{1, *}  $ given below:

\vspace{-3mm}
\begin{eqnarray}\label{eqn_u2}
  p^{1, *}:=  \min \Bigg\{ p_{sw}, \  \pmax, & \frac{(C_\sMi + C_\sS)}{2} 
    + \frac{\left( \dbar_\sMi + \varepsilon\dbar_\sMe - \varepsilon\sqrt{\alpha_\sMe O_\sMe} + \frac{\varepsilon\alpha_\sMi \sqrt{\alpha_\sMe O_\sMe}}{\alpha_\sMe} \right)}{2\alpha_\sMi(1-\varepsilon^2)} \Bigg\}.
\end{eqnarray}

The second optimization for $p > p_{sw}$ is given by the following and is applicable only when $p_{sw} \le  \pmax$:

 \vspace{-4mm}
{\small
\begin{eqnarray}
\begin{aligned}
    \max_{ p_{sw} \le p \le \pmax }
    & \Bigg( - O_\sMi - O_\sS + \left( \dbar_\sMi(1+\varepsilon^2) + \varepsilon\dbar_\sMe - \alpha_\sMi p \right) \left(p - C_\sMi - C_\sS\right) \\
    & \quad \hspace{-21mm} + \left( \varepsilon\alpha_\sMi p - \varepsilon\dbar_\sMi \right) \left( \frac{\left( \dbar_\sMe + \varepsilon\dbar_\sMi - \alpha_\sMe C_\sMe - \alpha_\sMe C_\sS \right)\left( \varepsilon\alpha_\sMi p - \varepsilon\dbar_\sMi \right) - \alpha_\sMe O_\sMe}{\alpha_\sMe (\varepsilon\alpha_\sMi p - \varepsilon\dbar_\sMi)} \right)  
   \Bigg).
\end{aligned}
\label{Eqn_second_opt_Par}
\end{eqnarray}}%
    The optimizer of the above by strict concavity is at   $p^{2, *}$, given below: 

\vspace{-3mm}
{\small\begin{eqnarray}
 p^{2, *} &:=& 
\max\!\left\{
p_{sw},\ 
\min\!\left\{
\pmax,\ 
p_{in}^{2, *}
\right\}
\right\}
\label{Eqn_p2_star}, \mbox{ where, } \\
p_{in}^{2, *} &:= &\frac{
\dbar_{\sMi}(1+\varepsilon^{2})
+ \varepsilon\dbar_{\sMe}
+ \alpha_{\sMi}(C_{\sMi}+C_{\sS})
+ \frac{\varepsilon\alpha_{\sMi}}{\alpha_{\sMe}}
(\dbar_{\sMe}+\varepsilon\dbar_{\sMi}-\alpha_{\sMe} C_{\sMe}-\alpha_{\sMe} C_{\sS})
}{2\alpha_{\sMi}}. \nonumber
\end{eqnarray}}

In all, the optimal value in the Op sub-regime  is given by:

\vspace{-4mm}
{\small\begin{eqnarray}\label{eqn_u_star_pr}
U_{_{Op}}^{*} = \max\Bigg \{ U_\sV (p^{1,*}\theta(p^{1,*}),  \  \  U_\sV (p^{2,*},  \theta(p^{2,*}) \indc{p_{sw} \le \pmax }   \Bigg \}.\hspace{1mm}\label{Eqn_opt_util_par} 
\end{eqnarray}}%
The  sub-optimal pair
is either  $(p^{1,*}\theta(p^{1,*})) $ or $(p^{2,*}, \theta(p^{2,*}) )  $, depending upon the bigger of the two in the above.

    We finally have the following result using  Theorem \ref{thm_Fco_positive} and the sub-optimal utilities in  \eqref{Eqn_opt_loss_util}, \eqref{Eqn_opt_max_util}, \eqref{Eqn_opt_util_par}:
\begin{theorem}
\label{Thm_all_in_one}\textbf{[Optimal  under Co-existence]}
    Assume {{\bf A.1-2}}.  Then, 
    \begin{eqnarray*}
        U^*_{co} 
      =  \max \Bigg \{  U(p_{co}^*, q_{co}^*) \indc{ (p_{co}^*, q_{co}^*) \in {\cal F}_{co}^+ }, \   U_{_{Op}}^{*},  \   U_{_{I\ell}}^{*} \indc{\psi(0) < \pmax},   \ U^*_{_{Mp}} \indc{ l_{_{Mp}}  <  r_{_{Mp}}  }   \Bigg \}. 
    \end{eqnarray*}
    \eop
\end{theorem}
{\bf Remarks:} Thus, the optimal operating pair for $\V$ within the co-existence regime can occur in one of the following four (actually three) sub-regimes:

\begin{enumerate}[(a)]
    \item \textbf{Operate both profitably (Bp):} The optimal pair lies in the interior of ${\cal F}_{co}^+$, where both the manufacturers derive strictly positive profits. This sub-regime also corresponds to the case where $\V$ and the out-house $\Me$ operate simultaneously and profitably.

    \item \textbf{Operate at par (Op):} The out-house manufacturer is forced to operate at zero profit, with the coalition $\V$ quoting its optimal price at the boundary $p^* = \theta(q^*)$. This reflects a strategic equilibrium where $\V$ maximizes its own payoff while constraining the competitor’s profit margin.

    \item \textbf{Operate at losses (I$\ell$):} The in-house unit $\Mi$ operates at a loss, a regime that is non-empty only when $\psi(0) < \pmax$. Here, $\V$ intentionally accepts some losses for its in-house unit to achieve bigger strategic advantage through the out-house market channel.

    \item \textbf{Operate at max (Mp):} The optimal price quoted by the in-house manufacturer $\Mi$ equals $\pmax$, corresponding to the boundary regime along $\{p = \pmax\}$. This sub-regime is non-empty only when $l_{_{Mp}} < r_{_{Mp}} $ (see~\eqref{Eqn_r_mx}).  This (Mp) optimizer is actually in Op regime if  $q^{*}(\pmax) = \theta(\pmax) $ in \eqref{Eqn_opt_max_util}, else it is in Bp regime. 
\end{enumerate}

Taken together, these three co-existence sub-regimes capture the coalition’s strategic flexibility in balancing cooperation and competition. From a managerial perspective, they demonstrate that profitability for a vertically integrated supplier need not always depend on the in-house unit’s direct success. In certain market conditions, deliberately constraining or influencing the out-house manufacturer through strategic pricing, for its in-house, can yield superior overall outcome.

\subsection{Comparison: Sub-regimes of Co-existence}\label{subsec_Comp}

We now compare the different sub-regimes within the co-existence regime in order to identify which operating modes are optimal under the given market conditions that satisfy the feasibility assumptions {\bf{A.1-2}}.

We begin with a case where $\varepsilon $ is near $ 1$ and with I$\ell$ regime. 
When $\psi(0) \ge   p_{mx}$,  the I$\ell$-regime is empty. From \eqref{Eqn_psi_q}  
the difference, $\psi(0) - p_{mx}
= \nicefrac{ (\varepsilon^2\dbar_\sMi - \varepsilon (1-\varepsilon^2)\dbar_\sMe + \varepsilon \alpha_\sMe C_\sMe) }
{ (\alpha_\sMi (2-\varepsilon^2))}$ 
and   thus when, 

\vspace{-2mm}
{\small\begin{eqnarray}
   \varepsilon \dbar_\sMi > (1-\varepsilon^2)\dbar_\sMe - \alpha_\sMe C_\sMe, \label{eqn_il_not_opt}
\end{eqnarray}}%
operating the in-house production unit at loss  is never optimal. In other words, \textit{when the product is sufficiently essential,  I$\ell$ is never an optimal regime.}  
\ignore{
We next establish the conditions 
under which the coalition  never prefers Bp regime: 
\begin{lemma}\label{lem_comp}
\textit{If}
$
(8-6\varepsilon^2)\alpha_\sMi \alpha_\sMe
- \varepsilon^2(\alpha_\sMi^2 + \alpha_\sMe^2) < 0,
$
\textit{then the $\Vi$ coalition finds it optimal to operate either in I$\ell$, Op, or Mp regimes.}
\end{lemma}

{\color{red}
Lemma \ref{lem_comp} implies that, under the above condition, interior co-existence equilibria in which both manufacturers earn strictly positive profits cannot be optimal except for the Mp-Bp regime. Equivalently, profitable co-existence without price saturation is possible only when the condition of Lemma \ref{lem_comp} fails to hold, i.e., when product essentialness is sufficiently small.}

Furthermore, for any fixed set of parameters other than $\varepsilon$, there exists a threshold $\bar{\varepsilon}<1$ such that for all $\varepsilon\ge\bar{\varepsilon}$, Bp  ceases to be optimal. To see this, define
\[
k(\varepsilon):=(8-6\varepsilon^2)\alpha_\sMi \alpha_\sMe
- \varepsilon^2(\alpha_\sMi^2 + \alpha_\sMe^2).
\]
When $\alpha_\sMi=\alpha_\sMe=\alpha$, this expression simplifies to
\[
k(\varepsilon)=8\alpha^2(1-\varepsilon^2),
\]
which is strictly decreasing in $\varepsilon$ and satisfies
\[
\lim_{\varepsilon\to1} k(\varepsilon)=0.
\]
Thus, although $k(\varepsilon)$ decreases monotonically, it converges to zero from above as $\varepsilon\to1$  in the symmetric case. }

For subsequent theoretical  analysis, in the regime $\varepsilon \to 1$, we require   $\alpha_\sMi\ne\alpha_\sMe$ for a technical reason   (for example results in  Lemma \ref{lem_compr} and Theorem~\ref{thm_all_r})). One can separately analyze the case with equal $\alpha$'s in a similar manner, however we skip them for keeping the paper short and because the case with $\alpha_\sMe$ close to $\alpha_\sMi$  would provide a good picture about the equal case, by regular continuity arguments (one does require some technical proof for saying this formally).  
We begin with the following lemma which first shows that  the optimal   among the  co-existence regimes near $\varepsilon \to 1$ is either  Mp-Bp or Op, and also finds the best among the two.

\begin{lemma}\label{lem_compr}
\textit{Assume $\alpha_\sMi\neq\alpha_\sMe$.}
 \textit{For any fixed set of parameters excluding $\varepsilon$, there exists $\bar{\varepsilon}<1$ such that for all $\varepsilon\ge\bar{\varepsilon}$, it is optimal for the coalition $\V$ to operate in Op or Mp-Bp regime. Further,}
\begin{itemize}
\item[(i)] 
\textit{ it is optimal to operate in Op regime, with  $(\pmax,\theta(\pmax))$ as the optimal pair if the co-existence score defined below is non-negative,}
\begin{eqnarray}
\Sco  :=   \frac{ \alpha_\sMe - \alpha_\sMi}{\alpha_\sMi}\dbar_\sMi
+ \frac{2\alpha_\sMi + \alpha_\sMe}{\alpha_\sMi}\dbar_\sMe
+   \alpha_\sMe(C_\sMe - C_\sMi)    \ge  0 ;
\label{Eqn_cond_forat_max} 
\end{eqnarray}
\item[(ii)]  \textit{
it is optimal to operate in Mp-Bp regime with $(\pmax,h(\pmax))$ as the optimal pair  
if $\Sco$ score given in  \eqref{Eqn_cond_forat_max} is negative.}
\end{itemize}
\end{lemma}
The proof is in \ref{Proof_Lemma_2}. \eop

Lemma~\ref{lem_compr} characterizes the asymptotic structure of the optimal co-existence regime as the product essentialness parameter approaches unity, i.e., as $\varepsilon \to 1$. In this regime, product becomes highly essential and the customer loyalty weakens. As a result, the coalition $\V$ no longer finds it beneficial to operate it's in-house unit at loss (I$\ell$ regime)---basically there is no requirement to explicitly force the customers towards the out-house manufacturer as they automatically switch gears because of essentialness (and non-loyalty). Instead, the optimal strategy necessarily reduces to one of two boundary regimes: the \textit{operate-at-par} or Op regime  or the \textit{maximum-price profitable co-existence} or Mp-Bp regime (here opponent also derives positive profit and $\V$ operates at maximum price, latter is  facilitated  again by essentialness).

More specifically, condition~\eqref{Eqn_cond_forat_max} serves as the threshold criterion governing this choice.  The score $\Sco$ of \eqref{Eqn_cond_forat_max} can be seen as a consolidated relative strength indicator, near essentialness regime, which is constructed by a special combination of  the three levers or the characteristics of the two production houses:  market potentials ($\dbar_\sMi, \dbar_\sMe$), price sensitivities ( $\alpha_\sMi, \alpha_\sMe$) and the production costs ($C_\sMi, C_\sMe$). When $\Sco < 0$, 
the out-house is a significant player and can't be bullied by $\V$---at the optimal choice of $\V$, the out-house also derives positive profit; here the optimal pair for $\V$ is  $(p_{mx},h(p_{mx}))$. 

If  $C_\sMe > C_\sMi$ and $\alpha_\sMe > \alpha_\sMi$, then the out-house is clearly weaker---here $\Sco > 0$ and $\V$ manages to force the out-house to operate at par---its optimal pair is $(p_{mx},\theta(p_{mx}))$. If one or both of the above inequalities are not true, i.e.,    when the out-house is supremum either in customer reputation (with $\alpha_\sMe <  \alpha_\sMi$) or in production cost  (with $C_\sMe < C_\sMi$), there is a possibility of  it operating with profit in the essentialness regime   (as $ \Sco $ can become negative under such conditions).   

Also observe the market potentials $\dbar_\sM$ and $\dbar_\sMe$ play a role in determining the zero/non-zero profits of out-house, only when the out-house is strictly superior in one of the two quantities, reputation factor or  the production costs.  More strikingly, with good reputation, for the sake of better explanation consider $\alpha_\sMe \approx 0$, we observe that a higher market potential of out-house can force it to operate at par: note with $\alpha_\sMe \approx 0$, $\Sco \approx - \dbar_\sM + 2 \dbar_\sMe$. In other words, having a higher market potential in the essentialness regime along with good reputation, can only become hazardous for the out-house manufacturer. We provide an elaborate numerical case study to clearly illustrate this surprising dependency in subsection \ref{subsec_num_2}, see Figures \ref{fig:low_ess}-- \ref{fig:high_ess}.

Thus, Lemma~\ref{lem_compr} is instrumental in characterizing the high-essentialness asymptotics of the co-existence regime.  It also provides the foundation for the subsequent analysis of global optimality that includes regimes other than co-existence, or  more precisely the single-existence   regimes, where one of the production units is shut or is forced to shut completely; these regimes  are  discussed in the next  section. The complementary low-essentialness scenario (results when $\varepsilon \to 0$), where the customers remain highly loyal to their respective manufacturers, is directly  analyzed along with single-existence regimes  after analyzing the latter.

\section{Single-existence regimes}\label{sec_Se}
In the single-existence regimes,  the coalition $\V$ allows  only one production unit to operate.  Accordingly, we have two regimes: i) Shut down its in-house unit (Sh regime); and  ii) Eliminate the downstream competition (E$\ell$ regime). We begin with the  study of the Sh regime. 

\subsection{Shut down the in-house production unit (Sh regime) }\label{subsec_sh}

In contrast to the co-existence regimes analyzed in Section~\ref{sec_Co}, the coalition $\V$ may  choose to \textit{completely shut down its in-house production unit (Sh)} and operate solely as an upstream supplier, if the choice is optimal.
By withdrawing from the downstream competition, the coalition eliminates internal cannibalization and effectively consolidates the downstream market in favor of the out-house $\Me$. Although this results in the loss of direct retail revenue, it may enhance the upstream profitability---a stronger (and monopoly) downstream presence of $\Me$ may expand the total demand captured, enabling the supplier to extract probably a much higher wholesale revenue.

We model the 
  demand attracted by the out-house manufacturer, in this $\Mi$-absent-regime, by first considering that 
  the market potential  of out-house  increases to $\dbar_\sMe + \varepsilon \dbar_\sMi$ (keeping in view of substituitability and essentialness factors). We then model the demand after customer response as below:
\begin{equation}\label{demand_alonee}
D_\sMe(n_o,\pe) = \dbar_\sMe + \varepsilon \dbar_\sMi - \alpha_\sMe \pe + r\varepsilon \alpha_\sMe \pe,
\end{equation}
using an additional parameter $r \in [0,1]$, which we refer to as the 
\textit{secondary fallback rate}, to capture the \textit{double fold-back} behavior 
of customers in the absence of~$M$. This modeling is considered  because of the following reasons:  a) among the $\alpha_\sMe \pe$ fraction of customers dissatisfied with $\Me$, the sub-fraction $\varepsilon \alpha_\sMe \pe$ would have resorted to $\Mi$, had $\Mi$ been operational; 
b) however, since $\Mi$ is not operational, they attempt to fold back to $\Me$ due to lack of options; and c) we model this fraction as $r\varepsilon \alpha_\sMe \pe$, using the additional parameter $r$. 
In effect, due to the absence of competition, the effective price sensitivity  reduces to  $\alpha_\sMe(1-r \varepsilon)$ and the market potential   increases to $\dbar_\sMe  + \varepsilon \dbar_\sMi$.
And the parameter $r$ quantifies the intensity of the fallback effect:
\begin{itemize}
    \item[i)] When $r \approx 0$, almost none of the customers who would have switched to $M$ 
    return to $\Me$, reflecting strong dissatisfaction or low substituitability.
    
    \item[ii)] When $r \approx 1$, nearly all such customers return to $\Me$, corresponding 
    to highly substitutable products and strong fallback behavior when alternatives 
    are unavailable.
\end{itemize}

We study a general  problem  by analyzing the system for different values of~$r$ and other system parameters.  We also have a special numerical case-study where we set  $ r =\varepsilon$ in subsection \ref{subsec_num_3}---such a study is important as both the factors represent a kind of fallback or substituting nature of the customers.

The utilities of the coalition $\V$ and the out-house manufacturer $\Me$ under this regime are given by (see \eqref{Eqn_Umj}-\eqref{eqn_Util_M}):
\begin{align}
U_{\sV} (n_o, q; \pe) &= D_\sMe(n_o,\pe) (q - C_\sS) - O_\sS, \label{eqn_util_v_sh}\\
U_\sMe (\pe; n_o, q) &= D_\sMe(n_o,\pe) (p_e - q - C_\sMe) - O_\sMe \label{eqn_util_Me_sh}.
\end{align}

Formally, the shutdown regime is characterized by the set of wholesale prices $q$ that ensure the out-house operates in the absence of in-house and hence is given by:
\begin{equation}
\F_{_{Sh}} := \{ (p, q) : p = n_o,\,  q \ne n_o,\, q \le  \theta(n_o)\},
\end{equation}
where $ \theta(n_o)$ is defined in a similar manner as in \eqref{Eqn_feasible_Regioin_Mj}---this represents the maximum wholesale price $q$ beyond which $\Me$ prefers not to operate even in the monopoly regime with $r$-fold-back---we obtain it by solving $U_\sMe^{*}(n_o,q)  = 0$, with   $D_\sMe$ as  in \eqref{demand_alonee} (see also   \eqref{Eqn_opt_policy_Mj}):
\begin{eqnarray}\label{eqn_theta_no}
    \theta(n_o) = \frac{\dbar_\sMe +\varepsilon\dbar_\sM - \alpha_\sMe(1- r\varepsilon)C_\sMe - 2\sqrt{\alpha_\sMe (1- r\varepsilon)O_\sMe}}{\alpha_\sMe(1-r\varepsilon)}.
\end{eqnarray}
 For finding the optimal utilities of the coalition $\V$ and the out-house $\Me$, we need to solve the Stackelberg game as we did in section \ref{sec_Co}, but here the utilities are given by \eqref{eqn_util_v_sh}-\eqref{eqn_util_Me_sh}. 
 For Sh regime, the Stackelberg game is simplified to optimization problem at both the levels---with  $p=n_o$, coalition $\V$ just has to optimize it's upstream wholesale price $q$ while $M_e$ has to optimize it's downstream retail price $p_e$ (as a monopoly). Thus the optimal or equilibrium  strategies  are obtained by solving the Stackelberg game as in the following:
 \begin{eqnarray*}
  \pe^{*}(q)  := \arg\max_{\pe}U_\sMe(\pe;n_o,q),  \ \mbox{ and } \ q^{*} := \arg\max_{q \le \theta(n_o)} U_\sV(n_o,q; \pe^{*}(q)). 
  \end{eqnarray*} 
  Further with $\pe^{*} := \pe^{*}(q^{*})$,
   \begin{eqnarray*}
  U_\sMe^{*}  :=   U_\sMe(\pe^{*};n_o,q^{*}) \mbox{ and } U_{_{Sh}}^{*} := U_\sV(n_o,q^{*};\pe^{*}).
 \end{eqnarray*}
 represent the utilities at the equilibrium. We obtain the same in the following:
\begin{lemma}\label{lem_shut}
Under $\bf{A.1}$ , the equilibrium prices and corresponding utilities in the Sh regime are as follows:
\begin{eqnarray*}
\pe^{*}  &=& 
\frac{3 (\dbar_{\sMe} + \varepsilon \dbar_\sMi) + \alpha_\sMe(1-r\varepsilon)(C_{\sS} + C_{\sMe})}{4 \alpha_\sMe(1-r\varepsilon)}, \nonumber \\
q^{*} &=& 
\frac{ \dbar_{\sMe} + \varepsilon \dbar_\sMi + \alpha_{\sMe}(1-r\varepsilon)(C_{\sS} - C_\sMe)}{2\alpha_{\sMe}(1-r\varepsilon)}, \label{eqn_opt_price_Sh}\\
U^{*}_\sMe &=& 
\frac{\left(\dbar_\sMe + \varepsilon \dbar_\sMi - \alpha_\sMe(1-r\varepsilon)(C_\sMe+ C_\sS)\right)^2}{16\alpha_\sMe(1-r\varepsilon)} - O_\sMe, \\
U^{*}_{_{Sh}}  &=&
\frac{\left(\dbar_\sMe + \varepsilon \dbar_\sMi - \alpha_\sMe(1-r\varepsilon)(C_\sMe+ C_\sS)\right)^2}{8\alpha_\sMe(1-r\varepsilon)} - O_\sS. \hspace{4mm} 
 \end{eqnarray*}
\eop 
\end{lemma}
{\bf Proof }
 follows from \cite[Lemma 4]{wadhwapartition} . 

\subsubsection*{Difference between  I$\ell$ and Sh regimes  }
The I$\ell$ and Sh regimes differ fundamentally in both market structure and strategic objective. In the \(I\ell\) regime, the coalition \(\V\) keeps the in-house unit operational but prices it sufficiently high that the unit incurs some losses. Its role is therefore not profit generation, but strategic discipline: by remaining active, the in-house unit constrains the downstream market power of the out-house manufacturer \(\Me\) and preserves the coalition’s bargaining position in wholesale pricing.

In contrast, under the Sh regime the coalition completely shuts down the in-house production unit \((p=n_o)\), allowing \(\Me\) to become the sole downstream producer. The \(Sh\) regime thus represents a qualitative departure from~I$\ell$---rather than maintaining downstream presence to influence market outcomes, the coalition deliberately relinquishes downstream control and relies exclusively on upstream value extraction.

{\color{blue}
Additionally, the two regimes differ fundamentally in the mechanism through which the out-house manufacturer captures demand. In the Sh regime, the complete shutdown of the in-house unit leaves the out-house manufacturer as the sole downstream seller. Consequently, its effective market potential increases from $\dbar_\sMe$ to $\dbar_\sMe+\varepsilon\dbar_\sMi$, while the effective price sensitivity changes from $1$ to $(1-r\varepsilon)$. In contrast, under the I$\ell$ regime, the in-house manufacturer remains operational, albeit at a loss. Consequently, dissatisfied customers continue to have an alternative seller and can switch directly to the out-house manufacturer instead of facing a monopolistic market. Thus, the additional demand captured by the out-house arises through customer substitution rather than through monopoly expansion. Since the coalition $\V$ may prefer to retail exclusively through the out-house manufacturer, it compares these two fundamentally different operating modes and selects the one yielding the higher profit. For ease of comparison, the equilibrium demands of the out-house manufacturer in the two regimes are summarized below (see \eqref{Eqn_Il_util}, \eqref{demand_alonee}):
\begin{eqnarray*}
\mbox{Sh regime: } \quad
D_\sMe^{*}
&=&
\frac{\left(\dbar_\sMe+\varepsilon\dbar_\sMi-\alpha_\sMe(1-r\varepsilon)(C_\sS+C_\sMe)\right)}{4},
\\[2mm]
\mbox{I$\ell$ regime: } \quad
D_\sMe^{*}
&=&
\frac{\left(\dbar_\sMe(1+\varepsilon^2)+\varepsilon\dbar_\sMi-\alpha_\sMe(C_\sMe+C_\sS)\right)}{4}
\indc{{\tilde q}^*\le r_{_{_{I\ell}}}}
\\
&&\qquad
+\;
\varepsilon^2\dbar_\sMe\,
\indc{{\tilde q}^*>r_{_{_{I\ell}}}}.
\end{eqnarray*}
In all, the Sh and the I$\ell$ regimes typically differ in customer response towards the out-house manufacturer.
}

This distinction raises a broader strategic question: under what market conditions should the coalition derive value through upstream wholesale extraction, downstream retail participation, or a combination of both? Section~\ref{sec_Co} showed that the I$\ell$  regime is not optimal when the product essentialness is high (see \eqref{eqn_il_not_opt}), irrespective of the relative powers  ($\alpha_\sMi,\alpha_\sMe$, $\dbar_\sMi, \dbar_\sMe$ and $C_\sMi, C_\sMe$) of the two production units. An immediate question in this regard is, will  $\V$ find Sh regime to be optimal in such scenarios?
The problem becomes more nuanced when product essentialness is moderate or low. 
Our aim in this paper is to derive answers to such questions in complete generality (however still under the minimal assumptions of this paper).  For now, we would like to mention that one can have scenarios under which  I$\ell$ is optimal or  Sh regime is optimal  among all possible choices for $\V$
(see  Figure \ref{fig:inferior}  of section \ref{sec_num} where we depict such possibilities, after considering the overall comparison).

\subsection{ Eliminate downstream competition (E$\ell$ regime)}
Similar to the exclusionary strategy followed by the coalition $\V$ to shut down it's in-house unit which led to the Sh regime analyzed in the previous subsection, the coalition $\V$ may choose to \textit{completely eliminate the out-house manufacturer (E$\ell$)} and operate solely through its in-house production unit, if such a choice is optimal. By strategically quoting a sufficiently high wholesale price, the coalition renders the operation of $\Me$ economically infeasible, thereby forcing its exit from the downstream market (recall we are considering a stylized model with one exclusive supplier, which can be close to real-world applications with monopoly in supply segment---probably where the competitor suppliers are significantly inferior).

This strategic choice is not motivated solely by short-term profit
considerations. By driving $\Me$ out of the market, the supplier reshapes the
competitive environment and attains full vertical control: the in-house unit
$\Mi$ becomes the exclusive downstream producer, effectively creating a
monopolistic configuration at both levels. Such a structure enables the supplier to govern
the entire value chain and this can become profitable  under certain conditions on the system parameters or the relative powers of the two production units.

The  payoff function of the coalition $\V$ in this regime is
\begin{equation}
U_{\sV}(p,q;n_o)
=
D_{\sMi}(p,n_o)\,(p - C_{\sMi} - C_{\sS})
-
O_{\sMi}
-
O_{\sS},
\end{equation}
where the demand $D_\sMi(p,n_o)$ faced by the in-house manufacturer $\Mi$ is given by the
$\Me$-absent demand model which follows a similar functional form as the  previously introduced $\Mi$-absent demand model in~\eqref{demand_alonee}:
\begin{equation}\label{demand_inhouse_alonee}
D_\sMi ( p, n_o) = \dbar_\sMi + \varepsilon \dbar_\sMe - \alpha_\sMi p + r\varepsilon \alpha_\sMi p,
\end{equation}
We again increase the potential to $\dbar_\sMi + \varepsilon \dbar_\sMe$ and decrease the price sensitivity to $\alpha_\sMi  + r\varepsilon \alpha_\sMi$, due to  the absence of  $\Me$, as in the previous subsection---
basically the term $\alpha_\sMi p$ captures the demand loss caused by the retail price $p$, while the final term $r\varepsilon \alpha_\sMi p$ models the fallback of  some customers who still consume the product of $M$  due to lack of options, despite being dissatisfied with the quoted price $p$.

Under the elimination regime, the supplier’s pricing decision therefore
centralizes all downstream activity within its own production unit.
Formally, the elimination regime is represented by the feasible action set
\begin{equation}
\F_{_{E\ell}}
:= 
\left\{
(p, q) :
p\neq n_o,\;
q \neq n_o,\;
q > \bar{\theta}
\right\},
\end{equation}
where $\bar{\theta}$ denotes the critical wholesale-price threshold beyond which the out-house manufacturer $\Me$ finds downstream operation economically infeasible and therefore optimally exits the market.
In this regime, the Stackelberg game like in section \ref{sec_Co} (see also the game in subsection \ref{subsec_sh}) collapses to the optimization problem (because of absence of $\Me$) where $\V$ maximizes the downstream retail price $p$. Thus we have the following:
\begin{eqnarray*}
    p^{*} := \arg \max_{p} U_\sV(p;n_o),
 \mbox{ and the optimal utility is, }
 U_{El}^{*} := U_\sV(p^{*};n_o)
 \end{eqnarray*}We prove that the above optimizer $p^*$ exists and obtain $U_{El}^*$  in the following:
\begin{lemma}\label{lem_elim}
The optimal price and utility $(p^*, U_{El}^*)$  of El regime is given by:
\begin{eqnarray}
    p^{*} &=& \frac{\dbar_\sMi + \varepsilon\dbar_\sMe + \alpha_\sMi(1-r\varepsilon)(C_\sMi + C_\sS)}{2\alpha_\sMi(1-r\varepsilon)}\label{eqn_opt_price_el}\\
    U_{E\ell}^{*} &=& \frac{\left(\dbar_\sMi+ \varepsilon\dbar_\sMe 
 -\alpha_\sMi(1-r\varepsilon)(C_\sMi+ C_\sS)\right)^2}{4\alpha_\sMi(1-r\varepsilon)}- O_\sMi - O_\sS.
\end{eqnarray}
\eop
\end{lemma}

{\bf Proof }
 follows from \cite[Lemma 4]{wadhwapartition}.

 \subsubsection*{Difference between  E$\ell$ and Op regimes  }
The E$\ell$ and Op regimes again differ fundamentally in both market structure and strategic objective. In the E$\ell$ regime, the coalition $\V$ completely eliminates the out-house manufacturer $\Me$  and  operates in a monopolistic manner in both the segments---it banks on the  effective market demand captured due to the absence of $\Me$,  which includes the transferable part of  $\Me$, captured with the help of parameters $(\varepsilon, r)$.
In contrast, under the Op regime,  it allows $\Me$ to remain operational while strategically maximizing its own revenue---this is achieved by setting a sufficiently high wholesale  price  that renders $\Me$ operate   at break-even point (or derive zero profit).  The Op regime therefore represents a controlled co-existence structure in which downstream competition is preserved operationally, but all economic surplus generated by the out-house manufacturer is extracted upstream by the coalition.

 As observed in Lemma~\ref{lem_compr} and also illustrated by numerical experiments in subsection \ref{subsec_num_2} (see cyan regions in Figure \ref{fig_overall}), even a relatively strong out-house manufacturer can be forced into break-even operation when the coalition strategically leverages its dominant upstream position, more so when the substituitability or the essentialness factors are high. Interestingly, we find many scenarios in which E$\ell$ regime can also become the optimal choice for coalition $\V$ (see yellow regions in Figure \ref{fig_overall}).

\section{Overall optimal choice for VC coalition}
The final question is regarding the optimal or the best choice among all the available configurations, which were   discussed in the previous sections. Clearly the optimal value of $\V$ is given by the following, using Theorem \ref{Thm_all_in_one}: 

\vspace{-3mm}
 {\small   \begin{eqnarray*}     
  U^{*}_{\sV} 
     & =&  \max \Bigg \{  U^{*}_{_{Sh}}, U^{*}_{_{E\ell}},  U^{*}_{co}  \Bigg \}  \\
      && \hspace{-15mm} = \  \max \Bigg \{ U^{*}_{_{Sh}}, U^{*}_{_{E\ell}} \\
      &&\hspace{-3mm}\max \Bigg \{  U(p_{co}^*, q_{co}^*) \indc{ (p_{co}^*, q_{co}^*) \in {\cal F}_{co}^+ }, \   U_{_{Op}}^{*},  \   U_{_{I\ell}}^{*} \indc{\psi(0) < \pmax},   \ U^*_{_{Mp}} \indc{ l_{_{Mp}}  <  r_{_{Mp}}  }   \Bigg \}    \Bigg \} ,
    \end{eqnarray*}}%
    and the corresponding optimizer (including the configuration) provides the overall  optimal choice for $\V$. 
The final solution of the above problem characterizes the optimal strategy of the encroaching supplier among the three strategic regimes (and their sub-regimes): \textit{co-existence (with Bp, I$\ell$ and Op as sub-choices)}, \textit{shutdown of the in-house production unit}, and \textit{elimination of downstream competition}. We next compare the equilibrium utilities achieved under each regime to identify the coalition's optimal strategy for the given market conditions and the relative strengths of the agents involved.

\subsection{Overall Comparison Analysis}\label{sec_overall_comp}

In the previous sections we derived the theoretical performance---the closed form expressions for optimizers and the utilities---in various sub-regimes.  
We now characterize the overall optimal choice of coalition $\V$, using these  performance expressions.  
One can  compute the  optimal  utilities of various regimes numerically using the derived expressions  and  easily obtain a numerical  comparative study for any given set of system parameters, the strength indicators of the two production units  $(\alpha_\sMi, \alpha_\sMe), $ $(\dbar_\sMi, \dbar_\sMe),$ $(C_\sMi, C_\sMe, C_\sS)$,  $(O_\sMi, O_\sMe, O_\sS)$ and the system parameters $(\varepsilon, r)$.
One can also derive
theoretical characterization of the optimal choice of $\V$ under some asymptotic regimes. 
We begin with the theoretical study in the immediate next, 
  while a more complete numerical study is considered later.

We study two asymptotic scenarios primarily based on the essentialness factor $\varepsilon$---the scenario characterized by $\varepsilon \to 0$ is referred to as the Low-Essentialness (LE) scenario, while that near $\varepsilon \to 1$ is the High-Essentialness (HE) scenario. We begin with characterization of the optimal choice in the LE scenario (proof is in  \ref{ref_thm_4}).
\begin{theorem}\label{thm_coex}{\bf [Optimal in LE scenario]}
Assume {\bf A.1--A.2}  and fix all system parameters except for $\varepsilon$.
There exists an $\bar{\varepsilon}>0$ such that for all $\varepsilon \le \bar{\varepsilon}$, it is optimal 
for the coalition $\V$ to operate in Bp regime, with both production units earning strictly positive profits. Further  the optimal price pair of $\V$ is given by \eqref{Eqn_pco_qco}, which converge to:
\begin{eqnarray*}
\lim_{\varepsilon \to 0}
 p^{*}_{co} = \frac{\left(\dbar_\sMi + \alpha_\sMi(C_\sMi + C_\sS)  \right)}{2\alpha_\sM} \  \mbox{ and } \  \lim_{\varepsilon \to 0} q^{*}_{co}  = \frac{\left(\dbar_\sMe - \alpha_\sMe(C_\sMe + C_\sS)  \right)} {2\alpha_\sMe}. \hspace{2mm} \mbox{ \eop}
\end{eqnarray*}
\end{theorem}

Thus when the product 
is not sufficiently essential and   when the customer loyalty 
is high (when the consumers are loyal towards the preferred brands and would only purchase when the prices are not too high),  
 coalition $\V$ finds Bp as the optimal strategy---where both the units operate profitably. High customer loyalty towards individual manufacturers ensures that the coalition $\V$ finds it neither beneficial to shut down a unit nor optimal to choke one of them  (i.e., operate in  I\(\ell \) or Op regimes). Interestingly, this is true irrespective of the fallback rate $r$.

We now derive the optimal configuration for  the HE scenario (the proof is in \ref{ref_thm_5}).
\begin{theorem}\label{thm_all_r} {\bf [Optimal  in HE scenario]}
Assume {\bf A.1--A.2} and consider that $\alpha_\sMi \ne \alpha_\sMe$ 
    Fix all system parameters except 
$\varepsilon$ and $r$. 
There exists a threshold $\bar r$ and another threshold $\epsilon_r < 1$ for each $r \ne \bar r$,  such that the following are the optimal choices for  coalition $\V$:
\begin{itemize}
    \item [i)] for  $r < \bar r$ and $\varepsilon \in (\epsilon_r, 1] $, \underline{the optimal regimes are Mp-Bp or Op}  depending upon the validity of ~\eqref{Eqn_cond_forat_max}, and with optimal pair as in Lemma~\ref{lem_compr}. 

    \item[ii)] for  $r > \bar r$ and $\varepsilon \in (\epsilon_r, 1] $, \underline{Sh} is the optimal regime with optimal pair as in Lemma \ref{lem_shut}, if the single existence score $\Sse $, defined below is positive:

\vspace{-3mm}
{\small\begin{eqnarray}
\Sse (r) := 
 \alpha_\sMi \left(\dbar_\sMe +  \dbar_\sMi - \alpha_\sMe(1-r)(C_\sMe+ C_\sS)\right)^2 - 8 \alpha_\sMe \alpha_\sMi (1-r) O_\sMi\hspace{-93mm} \nonumber \\ 
& &  - 2\alpha_\sMe  \left(\dbar_\sMi + \dbar_\sMe - \alpha_\sMi(1-r)(C_\sMi+ C_\sS)\right)^2  
 > 0, \hspace{10mm}\label{eqn_ineq_elim_shut}
\end{eqnarray}}%
 else and if the score $\Sse(r)$ is strictly negative then  \underline{E$\ell$} is the  optimal regime with optimal price \eqref{eqn_opt_price_el}.
\end{itemize}
\end{theorem}

The above result provides theoretical insights into optimal choice  of $\V$ when  $\varepsilon \approx 1$, see also Figure \ref{fig:theorem-regime}. 
One can further evaluate the two scores $\Sco$ and $\Sse$ analytically, to exactly determine the
the optimal configuration, 
for  two interesting case studies, which we consider in the immediate next.

\subsubsection{Superior In-house}

Consider a scenario where the in-house is superior both in terms of price-sensitivity and production technology:
 $\alpha_\sMi < \alpha_\sMe$ and $C_\sMi < C_\sMe$.  Then
$\alpha_\sMi (C_\sMi+ C_\sS) < \alpha_\sMe (C_\sMe+ C_\sS)$ and so
$
\alpha_\sMi \left(\dbar_\sMe+\dbar_\sMi-\alpha_\sMe(1-r)(C_\sMe+C_\sS)\right)^2
 $ is  strictly less than $
\alpha_\sMe\left(\dbar_\sMi+\dbar_\sMe-\alpha_\sMi(1-r)(C_\sMi+C_\sS)\right)^2
$, implying the score $\Sse$ is negative in \eqref{eqn_ineq_elim_shut}. 
Further $\Sco$ of  \eqref{Eqn_cond_forat_max} is positive,  using {\bf A}.1:

\vspace{-4mm}
{\small\begin{eqnarray*}
\dbar_\sMi+\dbar_\sMe-\alpha_\sMi(1-r)(C_\sMi+C_\sS)
-\sqrt{8\alpha_\sM(1-r)O_\sMi}  \hspace{-87mm}
\\
&=&\hspace{-2mm} 
\dbar_\sMi-\alpha_\sMi(1-r)(C_\sMi+C_\sS)
-\hspace{-1mm}\sqrt{4\alpha_\sM(1-r)O_\sMi}
+\dbar_\sMe
-(\hspace{-1mm}\sqrt8-\hspace{-1mm}\sqrt4)\sqrt{\alpha_\sM(1-r)O_\sMi}
\\
&\ge&
0+\dbar_\sMe
-(\sqrt8-\sqrt4)\sqrt{\alpha_\sMe(1-r)O_\sMi}>0.
\end{eqnarray*}}

Hence  \textit{E$\ell$ regime is optimal for
$r\ge\bar r$ and  Op is optimal when
$r<\bar r$.}
In other words, 
  when in-house unit is superior to  out-house solely in terms of price sensitivity and
production cost,  it is optimal for the coalition to keep the out-house manufacturer at
par when the fallback rate is low (not many customers of a  \manu \  exiting the market are interested in buying from the other) and eliminate  the downstream competition once
the fallback rate exceeds $\bar r$---this is true irrespective of the market potentials of the two production
units.

\vspace{2mm}

\subsubsection{Superior Out-house}
We next consider the scenario  where the out-house is superior  with,  $2\alpha_\sMe<\alpha_\sMi$ and $C_\sMe<C_\sMi$. Then
$
\alpha_\sMi(C_\sMi+C_\sS)>
\alpha_\sMe(C_\sMe+C_\sS),
$
and so
\vspace{-3mm}
{\small
\begin{eqnarray}
\alpha_\sMi
\left(\dbar_\sMe+\dbar_\sMi
-\alpha_\sMe(1-r)(C_\sMe+C_\sS)\right)^2
\nonumber\\
>
2\alpha_\sMe
\left(\dbar_\sMi+\dbar_\sMe
-\alpha_\sMi(1-r)(C_\sMi+C_\sS)\right)^2.
\end{eqnarray}
}%
Clearly, from \eqref{eqn_ineq_elim_shut}, $\Sse > 0$,  as
$8\alpha_\sMe\alpha_\sM(1-r)O_\sM>0$. Thus in scenarios with higher fallback rates, the Sh regime is optimal. 
However, the sign of $\Sco$ of \eqref{Eqn_cond_forat_max}  depends upon the relative values of the third characteristic of the manufacturers' market strengths, the market potentials $\dbar_\sMi$ and $\dbar_\sMe$---the score  $\Sco >0$  only if $\dbar_\sMe $ is sufficiently bigger than $\dbar_\sMi$.
In other words, in scenarios with smaller fallback rates, Op regime  is optimal if 
$\dbar_\sMe $ is sufficiently large, else the coalition finds it optimal to let the out-house operate profitably. 

Thus the surprising negative dependency  of the optimal configuration on the market potential of the out-house, among the co-existence scenarios, discussed immediately after Lemma \ref{lem_compr},  continues to  hold even for the overall optimal choice, albeit only when the fallback rates are small. More precisely, the coalition facing a strong out-house (in all three market characteristics)  will find it beneficial to shut-down its in-house unit, only if the fallback rate  is higher;  however will force the mighty out-house to operate at par when the fallback rate is small---and such a choice becomes optimal as the out-house market potential  $\dbar_\sMe$ increases beyond a threshold.

\subsubsection{Main Takeaways} 
Thus the following is the summary of our theoretical results.

\begin{itemize}
    \item[i)] If  we are dealing with a luxury product (with $\varepsilon $ close to zero), then the coalition finds it beneficial to operate profitably with optimal prices given by $p_{co}^*$, $q_{co}^*$ of Theorem \ref{thm_coex}.

     \item[ii)] On the other extreme if the product is essential, the optimal regime is determined by the sign of the two scores $\Sco$ and $\Sse$ defined respectively in \eqref{Eqn_cond_forat_max} and \eqref{eqn_ineq_elim_shut} and the fall back rate $r$ as in Theorem \ref{thm_all_r}. 

     \item[iii)]  For the remaining parameter space, the optimal regime can be obtained
numerically by evaluating and comparing  all the sub-optimal configurations in various sub-regimes, using the results  derived in the previous sections.
This numerical  procedure  is summarized
in Algorithm \ref{alg:figure4}. 

\begin{algorithm} 
\caption{To compute the optimal-choice of coalition~$\V$ }
\label{alg:figure4}
\small
\KwIn{Model parameters
$\dbar_\sM,\dbar_\sMe,
\alpha_\sM,\alpha_\sMe,
C_\sM,C_\sMe,C_\sS,
O_\sM,O_\sMe,O_\sS$, $\varepsilon$, $r$}

{

\Indp

1. Compute $(p^*_{co}, q^*_{co})$ using \eqref{Eqn_pco_qco}.

Check the validity of interior condition of Theorem \ref{thm_Fco_positive}, by verifying if the following inequalities are satisfied, using\eqref{Eqn_feasible_Regioin_Mj}, \eqref{Eqn_phi_st_p}-\eqref{Eqn_Fco_plus}: 
\begin{eqnarray*}
q_{co}^* &<& \min \{ \theta(p_{co}^*), \phi(p_{co}^*) \} , \  q_{co}^* > 0, \\  p_{co}^* &<& \min \{ \pmax, \psi(q_{co}^* ) \}
\mbox{ and }  p_{co}^* > 0, \mbox{ with } \pmax = \frac{\dbar_\sMi + \varepsilon \dbar_\sMe}{\alpha_\sMi}.
\end{eqnarray*}

If the above conditions are satisfied
set,   

$\qquad U^*_{co} = U(p_{co}^*, q_{co}^*)$; 

else set 

$\qquad  U^*_{co}  = - B$, where $B $ is a  positive value to indicate infeasibility. \\

\medskip
2.For the Op regime, compute the optimal utility $U_{_{Op}}^*$ using \eqref{eqn_u2}, \eqref{Eqn_p2_star}-\eqref{Eqn_opt_util_par}.\\

\medskip
3. Compute the optimal utility in Mp regime  
 using \eqref{Eqn_psi_inv_pmax}-\eqref{Eqn_opt_max_util}:\\
 if $l_{_{Mp}}  <  r_{_{Mp}}$, (see Theorem \ref{Thm_all_in_one}), set  
 
 $\qquad
 U_{_{Mp}}^* = U_\sV(\pmax, q^{*}(\pmax))$;
\\
 else set 
 
 $\qquad U_{_{Mp}}^* = - B$,  to indicate infeasibility.\\

\medskip
4. Compute the optimal utility  in $I\ell$ regime using \eqref{eqn_tilde_q_star}-\eqref{Eqn_opt_loss_util}:\\
 if $\psi(0) < \pmax$ (see Theorem \ref{Thm_all_in_one}), set 
 
 $\qquad
 U_{_{I\ell}}^*  = U_\sV(\pmax, q^*_{_{I\ell}})$;
 \\
 else set
 
 $\qquad U_{_{I\ell}}^* = -B$, to indicate infeasibility.\\
 
\medskip
5. Compute the optimal utilities
$U_{_{Sh}}^*$ and $U_{_{E\ell}}^*$ in Sh and E$\ell$ regimes respectively using Lemmas \ref{lem_shut}-\ref{lem_elim}.\\

\medskip

6. Obtain the  optimal value
$
U_{\sV}^*
=
\max
\left\{
U_{co}^* ,
U_{_{Op}}^* ,
U_{_{Mp}}^* ,
U_{_{I\ell}}^* ,
U_{_{Sh}}^* ,
U_{_{E\ell}}^*
\right\}$,   and  find the optimal regime.  

\ If Mp is optimal  

$\qquad$ if $r_{_{Mp}} = \theta(\pmax)$ 

$\qquad$ $\qquad$ then reassign the optimal regime as Op

$\qquad$ else

$\qquad$ $\qquad$ then reassign the optimal regime as Bp. 

 \Indm

}

\end{algorithm}

\end{itemize}

The overall theoretically derived optimal strategy of the coalition is shown in Figure \ref{fig:theorem-regime}.
Now we move to numerical examples  starting with the optimal regime characterization using Algorithm \ref{alg:figure4} to obtain the complete picture .

\begin{figure}[H]
\centering
\begin{tikzpicture}[scale=4]

\draw[->] (0,0) -- (1.05,0) node[right] {$\epsilon$};
\draw[->] (0,0) -- (0,1.05) node[above] {$r$};

\def\epslow{0.35}
\def\epshigh{0.65}
\def\rbar{0.5}

\fill[blue!25] (0,0) rectangle (\epslow,1);
\node at (0.18,0.5) {\textbf{Bp}};

\fill[green!25] (\epshigh,0) rectangle (1,\rbar);
\node at (0.83,0.25) {\textbf{Op/Bp}};

\fill[red!25] (\epshigh,\rbar) rectangle (1,1);
\node at (0.83,0.75) {\textbf{Sh/E$\ell$}};

\draw[dashed] (\epslow,0) -- (\epslow,1);
\draw[dashed] (\epshigh,0) -- (\epshigh,1);
\draw[dashed] (\epshigh,\rbar) -- (1,\rbar);

\node[below] at (\epslow,0) {$\bar{\epsilon}$};
\node[below] at (\epshigh,0) {$\tilde{\epsilon}$};
\node[left] at (0,\rbar) {$\bar r$};

\end{tikzpicture}
\caption{Optimal regimes for the  Coalition $\V$ in the $(\epsilon,r)$ plane.}
\label{fig:theorem-regime}
\end{figure}
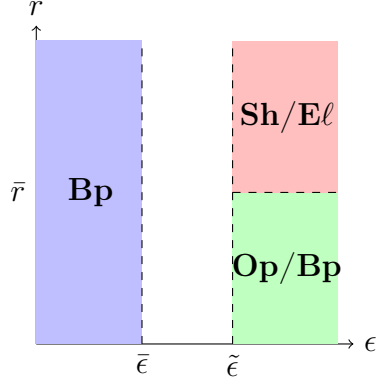

\section{Numerical comparison of various regimes}\label{sec_num}
We derived complete analytical insights for some important  corner cases in the previous sections. We also derived closed-form expressions for optimal utilities and the prices in all the possible configurations.  Thus one can numerically compute the optimal  utilities in each sub-regime and then determine the overall optimal configuration for $\V$, for any given set of parameters as in Algorithm \ref{alg:figure4}. 
We consider the same in the immediate next, with an aim to derive more insights into the  problem. 

\subsection{Optimal Regimes with varying $r$ and $\varepsilon$}\label{subsec_num_1}

To validate the analytical results and extract more practical insights, we conduct numerical experiments across a wide range of market environments by varying the \textit{product essentialness parameter} $(\varepsilon)$ and the \textit{secondary fallback rate} $(r)$. We also consider experiments with varying potentials.  These numerical examples illustrate how customer loyalty, market potential asymmetries, and fallback behavior jointly determine the supplier’s optimal strategy among  Bp, I$\ell$ ,Op,  E$\ell$ and Sh regimes.
The following parameters are kept fixed, while others are fixed/varied based on the experiment:
\[
C_\sMi = C_\sMe = 4, \quad C_\sS = 3, \quad O_\sMi = O_\sMe = O_\sS = 10.
\]

Our first set of results are provided in 
Figure~\ref{fig_overall} and its five sub-figures. These results summarize the coalition’s optimal regime as a function of $\varepsilon$ and $r$ under five representative market structures---recall the in-house  and  out-house can be compared in terms of all three levers, price-sensitivity $\alpha$ factors, production  $C_{\sM}$ costs and  $\dbar$ the market potentials and our studies are focused on $(\alpha, \dbar)$ for equal production costs---the combinations of these parameters define the five market structures as described below:  
\begin{itemize}
    \item In the first sub-figure, we consider a completely symmetric scenario with $\dbar_\sMi = \dbar_\sMe = 100 $ and $\alpha_\sMi = \alpha_\sMe = 0.1$. 

    \item In the second sub-figure, we consider an inferior in-house scenario with $\dbar_\sMi = 50, $ $ \dbar_\sMe = 100, $ $ \alpha_\sMi = 0.1 $ and $ \alpha_\sMe = 0.01$.

    \item In the third sub-figure, we consider one non-comparable scenario  with $\dbar_\sMi = 10, $ 
    $\dbar_\sMe = 100, $ $ \alpha_\sMi = 0.01$ and $ \alpha_\sMe = 0.1$.

    \item In the fourth sub-figure, we consider yet  another non-comparable  scenario  with $\dbar_\sMi = 100, $ $\dbar_\sMe = 10, $  $ \alpha_\sMi = 0.1 $ and $\alpha_\sMe = 0.03$.

     \item In the fifth sub-figure, we consider a superior in-house scenario  with $\dbar_\sMi = 100, $ $\dbar_\sMe = 50, $ $ \alpha_\sMi = 0.03 $ and $ \alpha_\sMe = 0.1$.
\end{itemize}
 In all the figures and sub-figures  we provide color patches in a two dimensional square, where each color  at any point represents the optimal configuration for  $(r, \varepsilon)$ pair representing the point. The following are the observations:
 \begin{itemize}
     \item In all the sub-figures Bp is the optimal regime (dark blue color) near $\varepsilon$ close to zero  and irrespective of $r$. On the other hand, when    $\varepsilon $ is close to one, either Sh/E$\ell$ is optimal for large values of $r$ or Op/Bp is optimal for smaller $r$. These results match exactly with those in Theorems \ref{thm_coex}-\ref{thm_all_r}.

     \item We also plot the  $\Sse$ score of \eqref{eqn_ineq_elim_shut}, vertically after each figure, to represent the score as a function of $r$; we use the same color convention as the main figures. As can be seen, in all the sub-figures the theoretical prediction via $\Sse$ scores matches exactly  with the numerically computed optimal configuration, when $\varepsilon$ is close to one and $r$ is above a threshold (observe the yellow and orange patches towards the right top corners of the five sub-figures match with the colors of the vertical bars.) 
     We further compute the $\Sco$ score of \eqref{Eqn_cond_forat_max} (see captions of each sub-figure) and observe that the predictions exactly match with optimal configurations of right bottom corners (for smaller $r$ and $\varepsilon$ close to 1). 
     In all, our theoretical estimates are good indicators for both the asymptotic regimes. 

     \item In most of the figures (except for sub-figure \ref{fig:inferior} with inferior in-house), the optimal configurations are one among the asymptotic ones and extend in a natural manner---for example in sub-figure \ref{fig:symmetric} for all values of $(r, \varepsilon)$ the optimal configuration is one among Bp or Op or E$\ell$, the ones predicted by theory in the two asymptotic regimes.  Thus our theoretical predictions through the two scores $\Sco$ and $\Sco$ near $\varepsilon$ close to one  can extend to  sufficiently smaller values of $\varepsilon$, while the Bp can be optimal not just near $\varepsilon$ close to 0, but for some sufficiently big values of $\varepsilon$ too. For exact understanding of the optimal configuration, one can use Algorithm \ref{alg:figure4},

     \item When the in-house is inferior, one can see more dependency of the optimal configuration on values of $(r, \varepsilon)$. This is the interesting case, where the coalition finds it optimal to operate its in-house at losses---observe I$\ell$ (represented by green color) is optimal for intermediate values of $\varepsilon$ and $r$ close to zero.  \textit{In other words, when the in-house is inferior in all aspects,   the product is not so essential and the customers are not  willing to switch easily, the coalition finds it beneficial to operate its in-house at losses.}
     
 \end{itemize}

\begin{figure}[h!]
\centering


\begin{subfigure}[t]{0.4\textwidth}
    \centering
    \includegraphics[width=\linewidth,height=3.2cm]{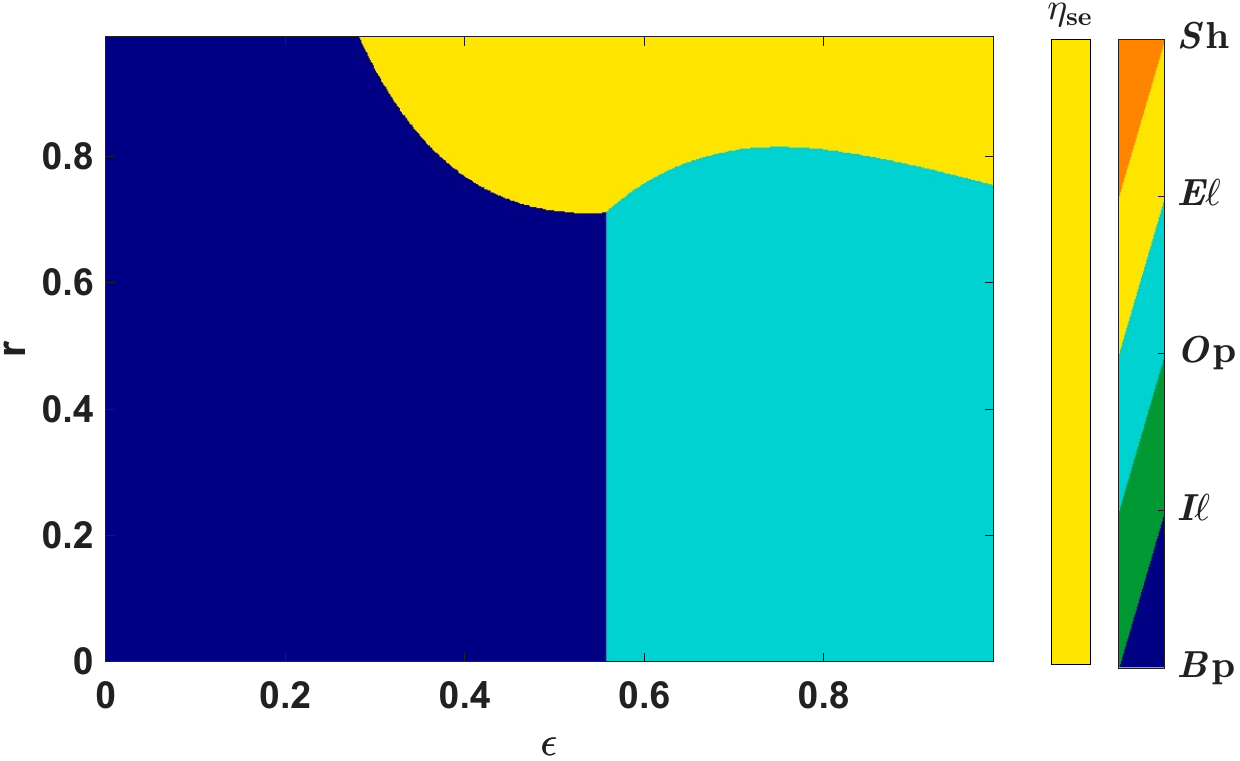}
    \caption{Symmetric ($\Sco = 300,$ so Op)}
    \label{fig:symmetric}
\end{subfigure}
\hfill
\begin{subfigure}[t]{0.4\textwidth}
    \centering
    \includegraphics[width=\linewidth,height=3.2cm]{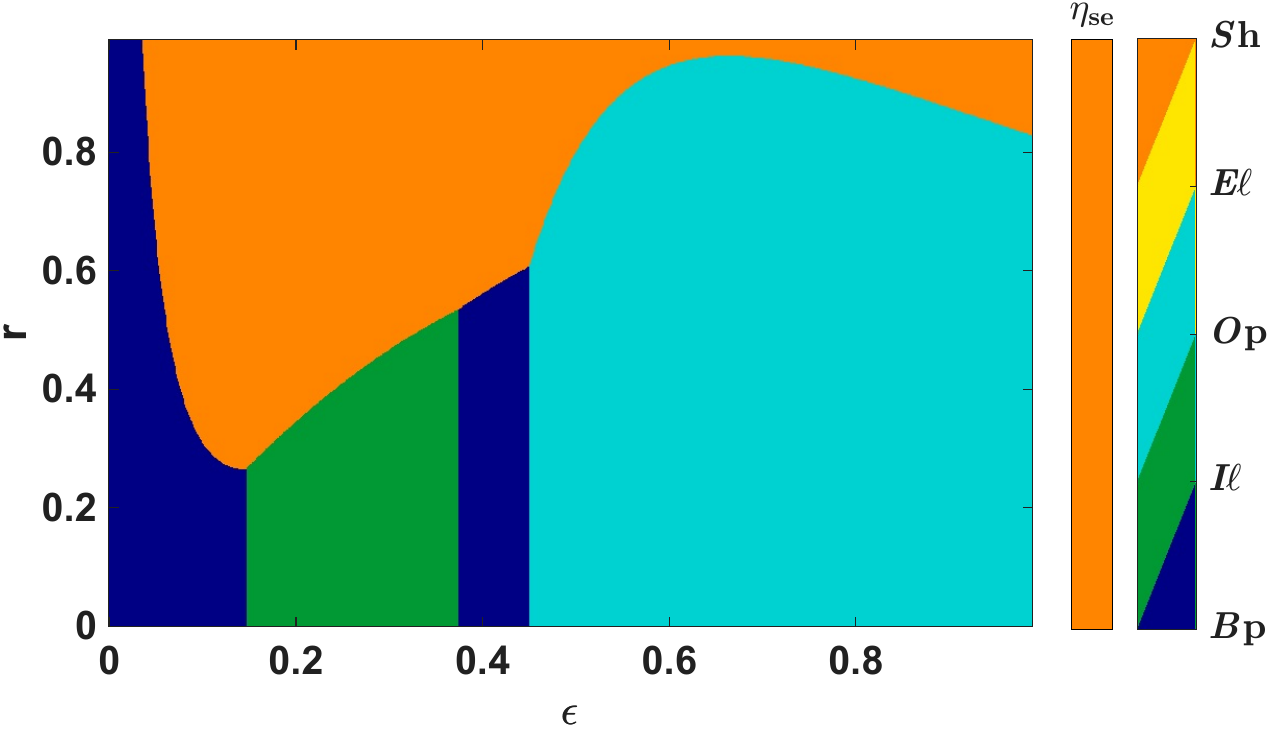}
    \caption{Inferior In-house ($\Sco = 165$)}
    \label{fig:inferior}
\end{subfigure}

\vspace{4mm}


\begin{subfigure}[t]{0.4\textwidth}
    \centering
    \includegraphics[width=\linewidth,height=3.2cm]{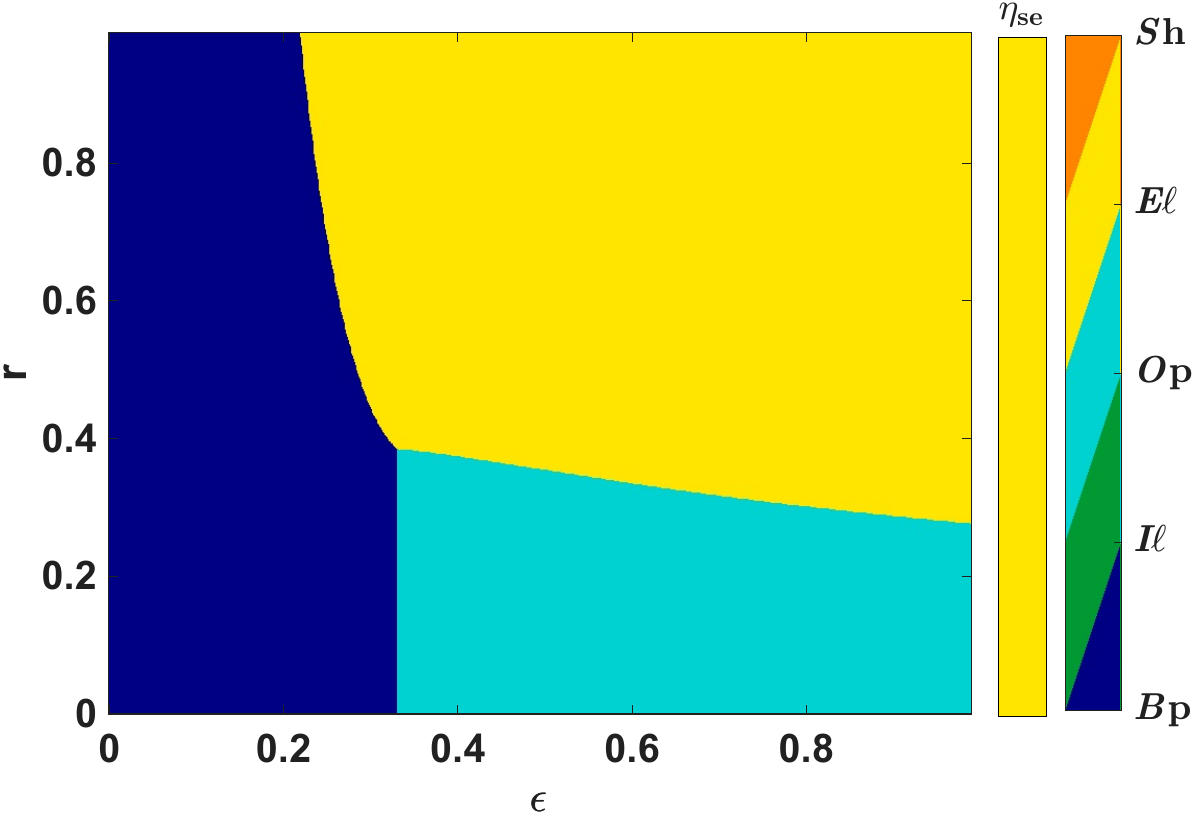}
    \caption{Non-Comparable I ($\Sco = 1200$)}
    \label{fig:noncomp1}
\end{subfigure}
\hfill
\begin{subfigure}[t]{0.4\textwidth}
    \centering
    \includegraphics[width=\linewidth,height=3.2cm]{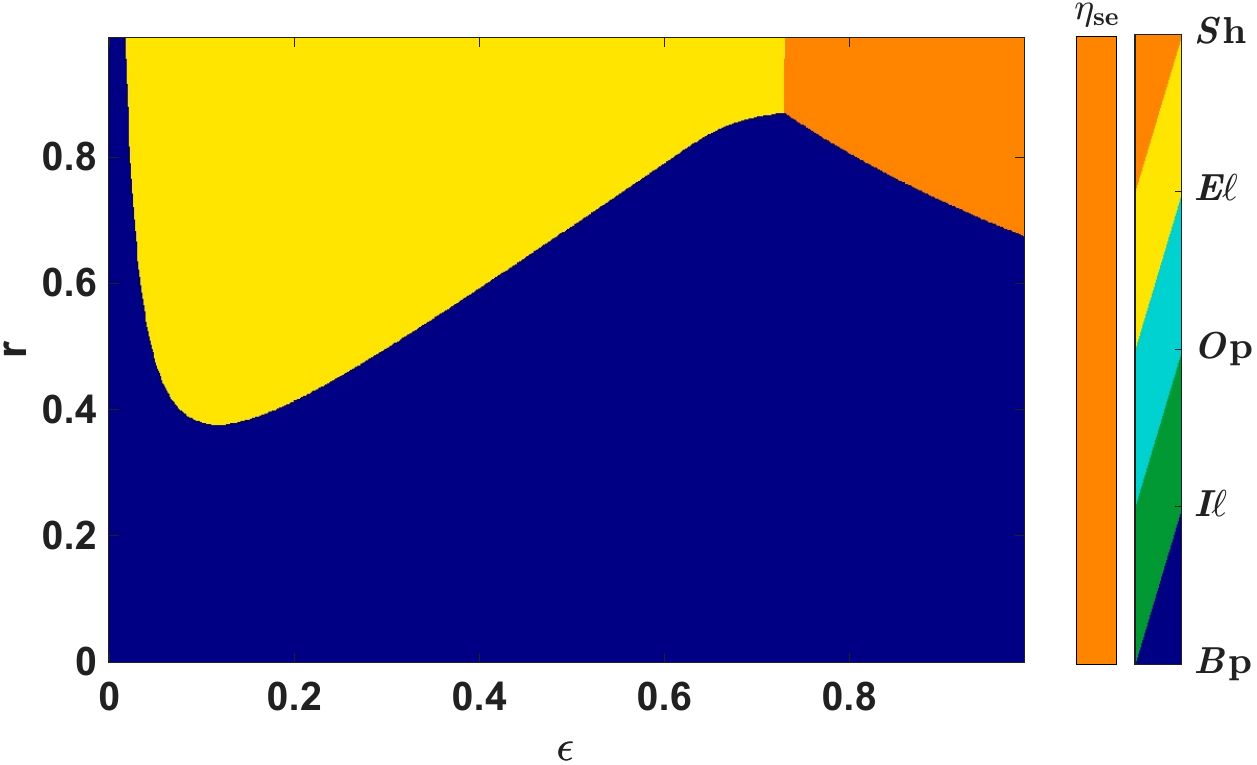}
    \caption{Non-Comparable II ($\Sco = -47$, so Bp)}
    \label{fig:noncomp2}
\end{subfigure}

\vspace{4mm}


\begin{subfigure}[t]{0.4\textwidth}
    \centering
    \includegraphics[width=\linewidth,height=3.2cm]{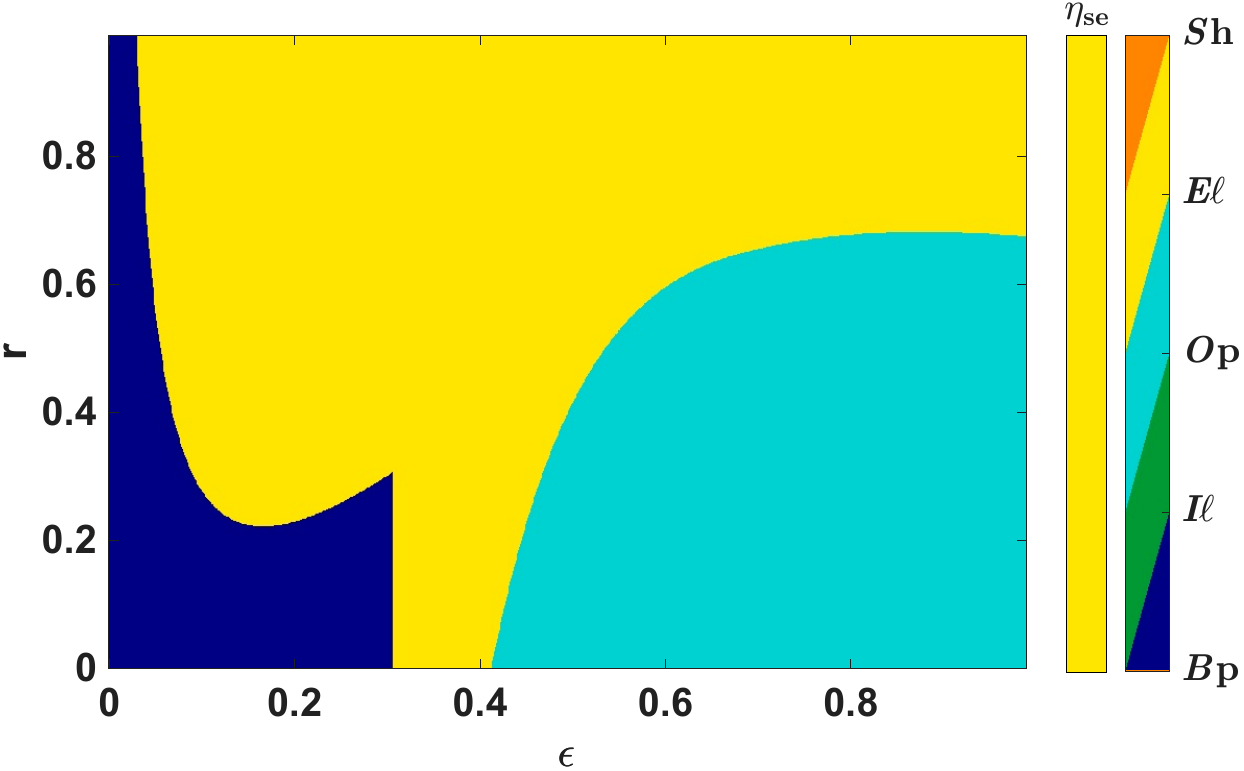}
    \caption{Superior In-house ($\Sco = 499.9967$)}
    \label{fig:superior}
\end{subfigure}
\hfill
\begin{minipage}[t]{0.32\textwidth}
\vspace{-3cm}
\begin{mdframed}
\scriptsize

\textbf{Non-comparable-I:}

In-house has smaller market but better reputation:
$
\bar d_M < \bar d_{Me},
\alpha_M < \alpha_{Me}.
$

\vspace{2mm}

\textbf{Non-comparable-II:}

In-house has bigger market but inferior reputation:
$
\bar d_M > \bar d_{Me},
\alpha_M > \alpha_{Me}.
$

\end{mdframed}

\end{minipage}

\caption{Optimal regimes for different values of  $r$ and $\varepsilon$.}
\label{fig_overall}

\end{figure}
\subsection{Optimal Regimes with varying $r$ and $\dbar_\sMe$}\label{subsec_num_2}

We next examine the surprising implication of Lemma~\ref{lem_compr} and Theorem~\ref{thm_all_r} (for smaller values of $r$)---if the out-house has bigger market potential it can actually be disadvantageous---it can be forced work  at break-even conditions. We fix $\bar d_{\sM}=800$, $\alpha_{\sMe}=0.00001$, $\alpha_{\sMi}=0.1$, $C_{\sMi}=4$, $C_{\sMe}=1$, $C_{\sS}=3$, and $O_{\sMi}=O_{\sMe}=O_{\sS}=10$, and plot the coalition's optimal regime as a function of the out-house market potential $\bar d_{\sMe}$ and the fallback rate $r$ for $\varepsilon\in\{0.35,0.65,0.95\}$  in Figures~\ref{fig:low_ess}--\ref{fig:high_ess}. 

At low essentialness ($\varepsilon=0.35$), the co-existence option dominates when the fallback for smaller values of rate $r$: the coalition finds it optimal to operate in Bp for small values of  $\bar d_{\sMe}$ and  shifts to $I\ell$ as $\bar d_{\sMe}$ increases; the coalition finds it beneficial to operate in Sh mode  once $r$ is sufficiently large. 

For intermediate values of essentialness   $\varepsilon=0.65$ in Figure \ref{fig:med_ess}, the Bp region shrinks sharply, with Op being optimal for  most values of $\dbar_\sMe$ when $r$ is small; however Sh or shutdown in-house 
is the best option once fallback rate $r$ is sufficiently large---interestingly this option is optimal for larger market potentials of out-house only for very high values of $r$. \textit{Thus once the essentialness factor is sufficiently high, the out-house is forced to operate mostly at break-even potential as its potential increases---this is  because of the combined effect of  monopoly of supplier in upper echelon and it's ability to enter downstream market via in-house.}

  At high essentialness ($\varepsilon=0.95$) in Figure \ref{fig:high_ess}, Bp almost disappears and the coalition's optimal choice alternates only between Op (below a fallback threshold) and Sh (above it). 
  Thus even at higher values of $\varepsilon$, as the out-house's market potential increases, the coalition forces it to operate at par rather than ceding profits, unless the fall-back rate is too high. Such high fall-back rates may not be realistic, and hence one can again conclude  that higher market potential is not a good news for out-house in the presence of a monopolistic supplier with in-house production house, unless the product is of luxury catagory.

\begin{figure}[h]
    \centering

    \begin{subfigure}[t]{0.32\textwidth}
        \centering
        \includegraphics[width=\linewidth,height=3.5cm]{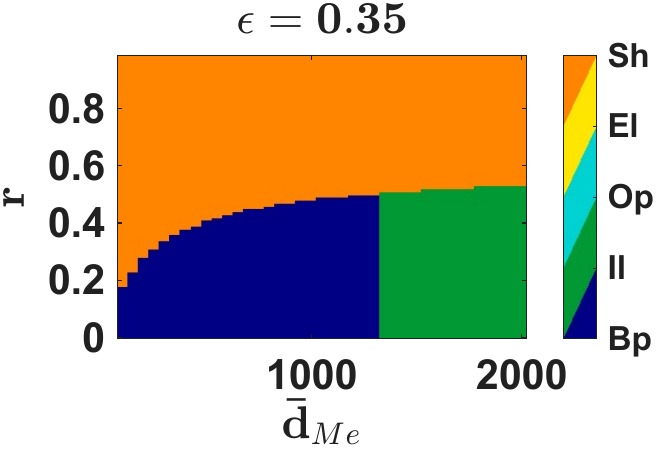}
        \caption{Less Essentialness}
        \label{fig:low_ess}
    \end{subfigure}
    \hfill
    \begin{subfigure}[t]{0.32\textwidth}
        \centering
        \includegraphics[width=\linewidth,height=3.5 cm]{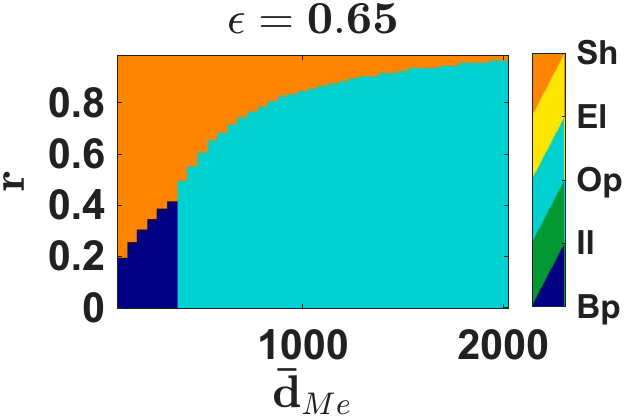}
        \caption{Medium Essentialness}
        \label{fig:med_ess}
    \end{subfigure}
    \hfill
    \begin{subfigure}[t]{0.32\textwidth}
        \centering
        \includegraphics[width=\linewidth,height=3.5cm]{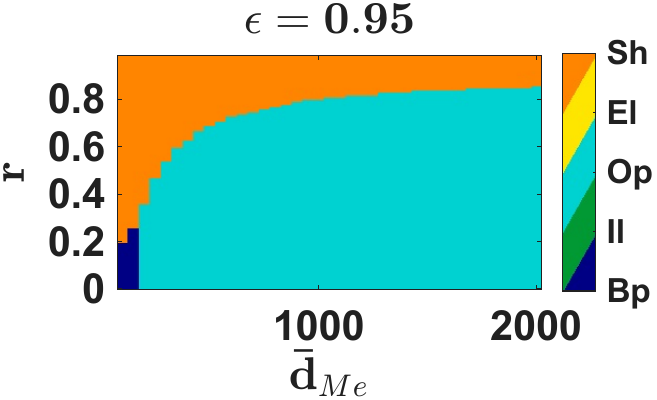}
        \caption{High Essentialness}
        \label{fig:high_ess}
    \end{subfigure}

    \caption{Optimal regimes with varying $\dbar_\sMe$ and $r$ when $\alpha_\sMe \to 0$}
    \label{fig:all_ess}
\end{figure}

\subsection{Optimal Regimes when $r=\varepsilon$}\label{subsec_num_3}

There is a natural correlation between the essentialness factor $\varepsilon$ and $r$,  the fallback rate---both these parameters represent a kind of fallback or substituting nature of the customers to an alternate production unit. We now consider a third case study, where we set $r=\varepsilon$, to obtain more focused insights for the scenarios that reflect the said correlation. 

We again investigate how the coalition's optimal regime changes with the out-house manufacturer's characteristics $(\dbar_\sMe, C_{\sMe}, \alpha_{\sMe})$ and with $r$ (which now equals $\varepsilon$) in Figures \ref{fig:low_ess_1}--\ref{fig:high_ess_1}.  The remaining parameters are set at: $\bar d_{\sM}=800$, $\alpha_{\sMi}=0.01$, $C_{\sMi}=4$, $C_{\sS}=3$, and $O_{\sMi}=O_{\sMe}=O_{\sS}=10$.

We first fix $\alpha_{\sMe}=0.001$ and $C_{\sMe}=1$ and vary the out-house market potential $\dbar_\sMe$ and $r$ in Figure~\ref{fig:low_ess_1}. For small values of $r$, once again the co-existence-oriented regimes dominate: the coalition operates in the $Bp$ regime when $\bar d_{\sMe}$ is small and transitions to the loss-making regime I$\ell$ as $\dbar_\sMe$ increases. As $r$ further increases, the coalition   relies on a more strategic market discipline Op, where it curbs the out-house to operate at par. For even  large values of $r$, co-existence ceases to be optimal and the coalition finds it optimal  to  shutdown the in-house (Sh regime),  irrespective of the out-house market potential.

Next, fixing $\alpha_{\sMe}=0.001$ and $\dbar_\sMe =1000$, we vary the out-house production cost $C_{\sMe}$ and $r$ in Figure~\ref{fig:med_ess_1}. The regime boundaries are nearly insensitive to changes in $C_{\sMe}$. As $r$ increases, the coalition sequentially transitions through the $Bp$, I$\ell$, $Op$, and $Sh$ regimes, indicating that the fallback intensity plays a much more significant role than the out-house production cost in determining the coalition's optimal strategy.

Finally, we fix $C_{\sMe}=1$ and $\bar d_{\sMe}=1000$ and vary the out-house manufacturer's price sensitivity $\alpha_{\sMe}$ and $r$  in Figure~\ref{fig:high_ess_1}. In contrast to the production cost, the coalition's optimal regime is highly sensitive to $\alpha_{\sMe}$. For very small values of $\alpha_{\sMe}$, the coalition transitions from I$\ell$ to $Bp$ and then to $Op$ as $r$ increases. However, even a modest increase in $\alpha_{\sMe}$ causes the elimination regime E$\ell$ to rapidly dominate the parameter space. The shutdown regime $Sh$ appears only for high values of $r$ when the out-house manufacturer has a very strong reputation.
\begin{figure}[H]
    \centering

    \begin{subfigure}[t]{0.32\textwidth}
        \centering
        \includegraphics[width=\linewidth,height=3cm]{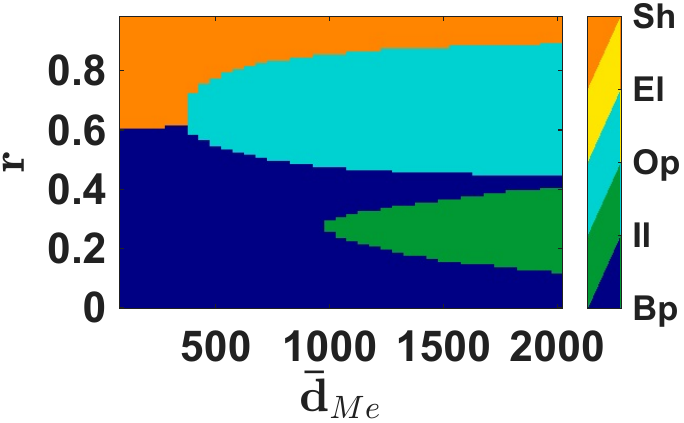}
        \caption{Varying Out-house Market Potential}
        \label{fig:low_ess_1}
    \end{subfigure}
    \hfill
    \begin{subfigure}[t]{0.32\textwidth}
        \centering
        \includegraphics[width=\linewidth,height=3cm]{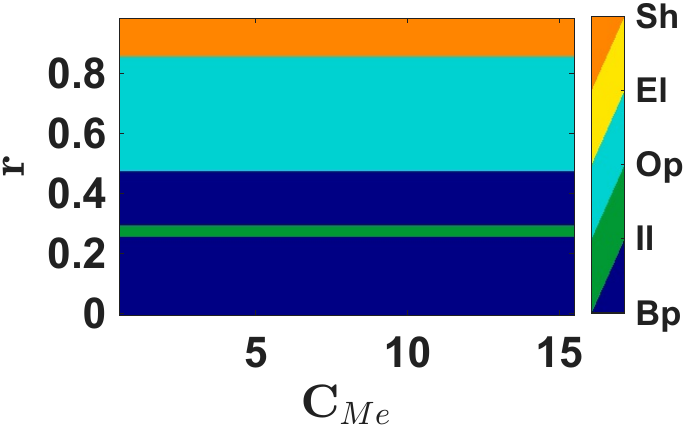}
        \caption{Varying Out-house Production Cost}
        \label{fig:med_ess_1}
    \end{subfigure}
    \hfill
    \begin{subfigure}[t]{0.32\textwidth}
        \centering
        \includegraphics[width=\linewidth,height=3cm]{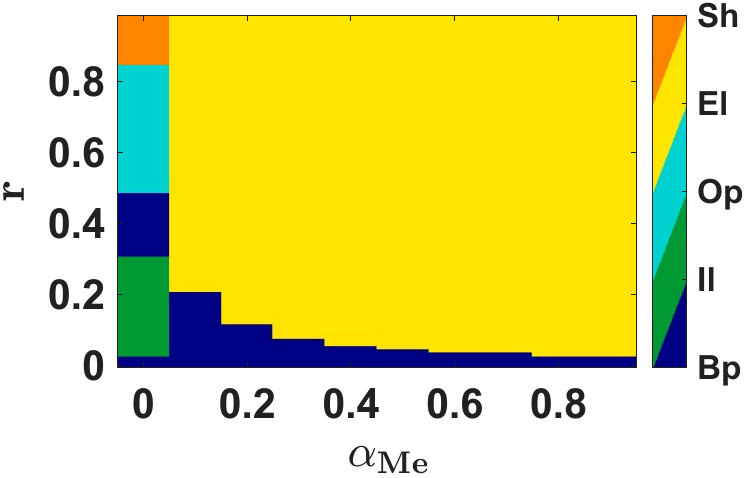}
        \caption{Varying Out-house Price Sensitivity}
        \label{fig:high_ess_1}
    \end{subfigure}

    \caption{Optimal regimes with $r= \varepsilon$}
    \label{fig:all_ess}
\end{figure}

\vspace{-3mm}

\section{Conclusions}
\label{sec_conclusions}

This paper studies  a partially vertically integrated supply chain in which a supplier simultaneously operates an in-house production unit and supplies raw material to an independent out-house manufacturer that competes with it in the downstream. Using a Stackelberg game framework with dedicated yet cross-influenced customer bases, we characterize how the coalition of supplier and in-house manufacturer can jointly set wholesale and retail prices to shape the entire downstream market to their advantage based on system parameters, rather than merely coordinate it.

Our analysis shows that the presence of an in-house unit fundamentally alters the competitive dynamics: it gives the supplier an instrument to discipline the out-house manufacturer  by forestalling its downstream monopoly. It can optimally select among co-existence with profits for both, force the opponent to operate at break-even point, in-house at losses, shutdown of in-house, or outright exclusion of the rival---depending on customer loyalty, product essentialness, demand fallback (in the absence of a particular production-unit), and the relative strengths of the two units. We derive closed-form optimal prices and utilities for each of these regimes and provide a numerical procedure that identifies the coalition's optimal choice for any given set of market and system parameters.

For two limiting, yet important, cases, we also provide  provable optimal characterizations. When the customer loyalty is high and the products are not so essential (i.e., for the case with luxury products), we prove that profitable co-existence is optimal for the coalition. At the opposite extreme, when essentialness is high, we prove that the optimal regime collapses to a choice between at most two alternatives, governed by two closed-form relative-strength scores: for low demand fallback, the coalition either forces the out-house manufacturer to operate at break-even or sustains profitable co-existence at the maximum retail price; for high fallback, it either shuts down its own in-house unit or eliminates the out-house manufacturer entirely. Outside these two limiting regimes, our numerical study---validated against the theoretical asymptotics---shows that the transition between regimes is generally smooth and monotone in essentialness factor and fallback rate, with one notable exception: when the in-house unit is distinctly inferior to its rival, the coalition can find it optimal to operate the in-house unit at a loss over an intermediate range of essentialness, using it purely as a competitive instrument rather than a profit center.

The most surprising outcome of this study is that a higher market potential can become detrimental to the out-house---the coalition forces the otherwise superior out-house  to operate at par when the market potential of the latter is higher than  a threshold and optimally prefers profitable operation of both units otherwise---this is true especially when the out-house is superior in terms of price-sensitivity (or market reputation) and production costs.

This work opens several avenues for further research. A natural extension is to study the same encroachment problem when the supplier is no longer a monopoly in the upstream market. Another promising direction is to embed the model in a dynamic setting incorporating inventory decisions, demand fluctuations, learning, and inter-period strategic adaptation. 
Yet another interesting study is to quantify the incremental value of the vertical encroachment by comparing supplier profit with and without an in-house production capability.

\ignore{
\section{Impact of Encroachment on SC Agents}
Talk about total utility in all the three cases, co-existence, shut down in-house, and eliminate downstream...also out-house manufacturer is most benefited when?

\section{ Case Study when costs are normalized to zero} 
\begin{cor}
 When essentialness is high and manufacturers have same price sensitivity parameters, it is always beneficial to operate at par.      
\end{cor}
}



 \bibliographystyle{elsarticle-harv}   \bibliography{ref}


\appendix
\section*{Appendices}
\setcounter{section}{1}

 
\section{ }
\label{sec_Appendix_AA}

\subsection*{\bf Proof of Theorem \ref{thm_Fco_positive}:}
\label{Proof_thm_1}
From  \eqref{Eqn_phi_st_p}, 
 $\phi(p) < 0$, when $\varepsilon \alpha_\sMi p > \dbar_\sMe + 2 \varepsilon \dbar_\sMi - \alpha_\sMe C_\sMe$. For such $p$,    $(p,q) \notin {\cal F}_{co}^+$ for any   $q$ (see \eqref{Eqn_Fco_plus}). 
 Also, from \eqref{Eqn_psi_q}, $p \le \psi(q)$ if and only if $q \ge \psi^{-1} (p)$.
 Thus a more direct  representation of ${\cal F}_{co}^+$   \eqref{Eqn_Fco_plus} is given by:
\begin{eqnarray}\label{eqn_cal_F_co_+}
    {\cal F}^+_{co} &=&  \left \{ (p, q)    :  0 \le   p \le \bar{p} (q)  \mbox{ and } \max\left \{0, \psi^{-1}(p) \right \} \le  q \le {\bar q}(p)  \right \}  \mbox{ with }  \hspace{10mm} \\
    {\bar p} (q) &:=&
    \min \left \{\pmax, \psi(q), \phi^{-1}(0)  \right  \}  \nonumber \\
    &= & \min \left \{ \frac{\dbar_\sMi + \varepsilon \dbar_\sMe}{\alpha_i}, \psi(q), \frac{\dbar_\sMe + 2\varepsilon\dbar_\sMi -\alpha_\sMe C_\sMe}{\varepsilon\alpha_\sMi}   \right  \}. \label{Eqn_bar_p}  \mbox{ and }  \\
    {\bar q}(p) &:=& \min \left \{\theta (p), \phi (p) \right  \} \stackrel{a}{=} \left \{ 
    \begin{array}{lll}
       \theta (p)  &  \mbox{ if } p  \le  p_{sw} = \frac{\dbar_\sMi}{\alpha_\sMi} + \frac{\sqrt{\alpha_\sMe O_\sMe}}{\varepsilon\alpha_\sMi} \\
      \phi (p)    & \mbox{ else. }
      \label{Eqn_bar_q_p}
    \end{array} 
    \right . 
\end{eqnarray}
(by direct computations using \eqref{Eqn_opt_policy_Mj}, \eqref{Eqn_psw} and \eqref{Eqn_phi_st_p} one can verify equality `$a$'). 

The function
$U_\sV$ is continuous and ${\cal F}_{co}^+$ is bounded (as $\pmax < \infty$), thus we have an optimizer for the problem in the LHS of \eqref{eqn_vc_opt_f_co_+}.

  Define $p$-sections $\S_p := {\cal F}_{co}^+ \cap \{(p, q): q \ge 0\}$ lines for each $p \le \pmax$. From~\eqref{eqn_cal_F_co_+},\vspace{-3mm}
\begin{eqnarray}
    \label{Eqn_Sp}
 \S_p  =   [ l(p), \ {\bar q} (p) ], \mbox{ with }  l(p) :=  \max\left \{0, \psi^{-1}(p) \right \} 
  \end{eqnarray}
  which is  an empty set $\emptyset$, when left point $l(p) > {\bar q}(p)$, and otherwise an interval.
The idea is to find sub-optimizers in each 
  $\S_p$ and then find the global optimizer. 
  Towards this goal, first note that
     the function $U_\sV$ in  ${\cal F}_{co}^+$ matches with  the   `unconstrained' function ${\cal U}$ given in equation \eqref{eqn_util_co-exist_uc}, which can be rewritten as (we are repeating the definitions of $\{w_i\}_i$ constants for ease of reading):
\begin{eqnarray}\label{eqn_Util_w}
   {\cal U}(p,q) \ = \ w_1 p^2 + w_2 pq + w_3 q^2 + w_4 p + w_5 q + w_6, \mbox{ with }  \hspace{26mm}&& \\
    \begin{array}{llll}
  &  w_1 \ = \ \frac{ -\alpha_\sMi \left (2- \varepsilon^2  \right )}{2} \hspace{4mm}\nonumber  %
    &   w_4 \ = \  \frac{2\dbar_\sMi + \varepsilon\dbar_\sMe + \varepsilon\alpha_\sMe C_\sMe  - \varepsilon\alpha_\sMi C_\sS + \alpha_\sMi(2-\ \varepsilon^2)\left(C_\sMi + C_\sS \right) }{2}  \\ 
   &  w_2 \ =  \  \frac{\varepsilon\left(\alpha_\sMi + \alpha_\sMe\right)}{2}   
   & w_5  \ = \  -\frac{\varepsilon\alpha_\sMe \left(C_\sMi + C_\sS\right)}{2} + \frac{\left(\dbar_\sMe - \alpha_\sMe C_\sMe + \alpha_\sMe C_\sS\right)}{2}   \\
   & w_3 \ = \  -\frac{\alpha_\sMe}{2} 
   &   w_6 \ = -\left(\dbar_\sMi + \frac{\varepsilon\left(\dbar_\sMe + \alpha_\sMe C_\sMe \right)}{2}\right)\left(C_\sMi + C_\sS \right)  \\
   & & \hspace{20mm}  - \left(\frac{\dbar_\sMe - \alpha_\sMe C_\sMe}{2}\right)C_\sS.
    \end{array} \nonumber 
\end{eqnarray}    
The   second derivative $\nicefrac{\partial^2 {\cal U}}{\partial^2 q} = w_3 < 0$ for all $(p,q)$. 
 %
     Thus for any $p$ with $\S_p \ne \emptyset$,  the     sub-optimizer of  the sub-optimization problem $ \max_{q : (p,q)\in \S_p} {\cal U}(p,q)$ is  unique by strict concavity  and equals,
  \begin{eqnarray}\label{eqn_q_star_p}
      q^*(p) := 
     \max\{l(p),  \min \{ h(p),  {\bar q}(p)  \},  \mbox{ where } h(p) := - \frac{w_2 p + w_5} {2 w_3}  
  \end{eqnarray} 
     is the  `unconstrained'  optimizer  of  ${\cal U} (p, \cdot)$ over $\{  q  \in {\cal R} \}$, 
     $l(p)$  and ${\bar q}(p)$ are  the   boundary points   of $\S_p$ (see~\eqref{eqn_cal_F_co_+}, \eqref{Eqn_Sp}).   

 Using \eqref{Eqn_bar_p} and \eqref{Eqn_psi_q}, define,  
\begin{eqnarray}
         \bar p &:=& \bar p(0) = \min \left \{ \pmax, \ \psi(0), \ \phi^{-1} (0) \right  \}     \nonumber \\
         &
     = & \hspace{-5mm}
     \min \left \{ \frac{\dbar_\sMi+\varepsilon \dbar_\sMe }{\alpha_\sMi},   \frac{   2 \dbar_\sMi + \varepsilon  \dbar_\sMe + \varepsilon \alpha_\sMe  C_\sMe }{\alpha_\sMi (2-\varepsilon^2)},  \frac{\dbar_\sMe + 2\varepsilon\dbar_\sMi -\alpha_\sMe C_\sMe}{\varepsilon\alpha_\sMi} \right  \}.  \hspace{9mm}
\label{Eqn_pbar}
 \end{eqnarray}
 From \eqref{Eqn_phi_st_p}, \eqref{Eqn_psi_q} and \eqref{eqn_cal_F_co_+},  when  $p \le  {\bar p} (0)$,  we have  $\phi(p) \ge   0$   and so $   {\bar q}(p) \ge  0$ (as from \eqref{Eqn_feasible_Regioin_Mj},  $\theta(p') >0$ for any $p'$, under {\bf A.1}) and $\psi^{-1}(p) \le 0$.
 Hence $l(p) = 0$ and $   {\bar q}(p) \ge  0$ in \eqref{Eqn_Sp} and 
 \textit{hence $\S_p \ne \emptyset$  for all    $p \le  {\bar p} (0)$.} Also observe $\bar q(p) > 0$  for all $p < \bar p$,  as $\phi$ is linear, and by definition of $\bar p$ in \eqref{Eqn_pbar}, 
 $\bar q(\bar p) = 0$.
 
 Further   $h(p) >0$ for all $p$ by {\bf A}.2 and thus  \eqref{eqn_q_star_p} equals (recall $q^*(p)$ is only defined for $p$ with $\S_p \ne \emptyset$):
\begin{eqnarray}
    q^*(p) = \left \{
    \begin{array}{ll}
  \min \{ h(p),  {\bar q}(p) \}       &  \mbox{ if } p \le \bar p \\ \\
\max \{ l(p),   \min \{ h(p),  {\bar q}(p) \}   \}         & \mbox{ else, i.e., if,  } {\bar p}  < p  < \pmax \mbox{ and } \S_p \ne \emptyset   \hspace{-12mm}
    \end{array}
    \right .  \nonumber  \\
    \label{Eqn_qstar}
\end{eqnarray}
%
%
We now analyze the last line of \eqref{Eqn_qstar}.
In the sub-case with $\bar p=\bar p (0) < \psi(0)$ we have $\S_p = \emptyset$ as: 
 \begin{itemize}
     \item either   ${\bar p} = \pmax$ in \eqref{Eqn_bar_p} and then clearly $\S_p = \emptyset$ for all $p > \bar p$; 
     \item   or 
 $\phi(\bar p) = 0$ and so $\phi(p) < 0$ (and so $\bar q(p) < 0$) for all $p > \bar p$,  and then again 
 $\S_p = \emptyset$ for all $p$.
 \end{itemize}
And  when $\bar p = \psi(0)$, for all $p > \bar p$ we have $l(p) = \psi^{-1}(p) > 0$ --- hence in the last line of \eqref{Eqn_qstar} we have $l(p) > 0$ and so $q^*(p) > 0$. Thus, in all, we have $q^*  (p) > 0 $ for all $p$ with $\S_p \ne \emptyset$, except for the corner case\footnote{We always have $h(\bar p) > 0$.  
When $\bar p < \phi^{-1}(0)$, we have $\phi(\bar p) > 0$ from \eqref{Eqn_phi_st_p} and further:\\   
i) if $\bar p = \psi(0)$, then  $\bar{q}(\bar p) = \phi (\bar p) > 0$, as $  \psi(0) > p_{sw}$ using \eqref{Eqn_psw} and \eqref{Eqn_bar_q_p}; and \\   ii)  if $\bar p = \pmax$, then also $\bar{q}(\bar p) > 0$ as  $  \theta(p_{mx}) > 0 $.
} when $\bar p = \phi^{-1}(0)$  for which $q^*(\bar p) =\bar q(\bar p) = 0$. 
 
From \eqref{Eqn_pbar} and {\bf A}.1  we have $\bar p > 0$, thus  
 there exists at  least one $p$ such that $\S_p \ne \emptyset$,
  and hence:
$$
\max_{(p, q) \in {\cal F}_{co}^+} U_\sV(p,q) = \max_{p \le \pmax, \S_p \ne \emptyset}  U_\sV(p, q^*(p) ). 
$$
In other words,   the global optimizer of $U_\sV$ in ${\cal F}_{co}^+$ is among, 
    \begin{eqnarray} 
\nonumber
    \mathbb {L}^* &:= &\{ (p, q) : 0\le p \le \pmax, \S_p \ne \emptyset, q = q^*(p) \} \\
     &=& \bigg  \{ (p, q) :  0 \le  p \le   {\bar p}, \ \  \ \ \ \ \ \ \ q= q^*(p) = \min \{ h(p),  {\bar q}(p) \} \bigg  \}   \nonumber \\
&& \cup \bigg \{ (p, q) : {\bar p}  < p \le \pmax, \ \S_p \ne \emptyset, \ \ \  q = q^*(p) \bigg  \} . \label{Eqn_L_star} 
 \end{eqnarray}
 Also since $q^*(p) > 0$ for all $\S_p \ne \emptyset$ with $p \ne \bar p$ (and when $\bar p = \phi^{-1} (0)$), we have (depending upon $\bar p = \phi^{-1} (0)$ or $\bar p \ne \phi^{-1} (0)$)
 $$
 {\mathbb L}^* = {\mathbb L}^* \cap \left ( \{ (p, q) : q > 0\} \cup \{ (\bar p, 0) \} \right ) \mbox{ or } {\mathbb L}^* \cap \left ( \{ (p, q) : q > 0\}  \right ). 
 $$

   Define the  following mapping, using \eqref{eqn_Util_w} and function $h$,  whose optimizer over $\{p: \S_p \ne \emptyset\}$   can  be a potential optimal point:
 \begin{eqnarray}
 \label{Eqn_omega}
\omega(p) &:=& {\cal U}(p, h(p)) = w_1 p^2 -  \frac{w_2^2 p +w_5 w_2 }{2  w_3  }   p  +\frac{ (w_2 p +w_5)^2}{4  w_3  }    \\
&&
+ w_4 p - \frac{w_2 w_5  p +w_5^2}{2  w_3  }  + w_6. \nonumber
\end{eqnarray}
The first derivative of    $\omega$ 
 at $p = 0$ is given by:
 \ignore{
 {\color{red} $$
\left . \frac{d \omega }{ d p}\right .
=   2 w_1 p - \frac{ 2 w_2^2p  + w_5 w_2}{2w_3} + \frac{ 2w_2(w_2 p +w_5)}{4  w_3  }  + w_4  - \frac{w_2 w_5}{2w_3}   
=  2 w_1 p -  \frac{  w_2^2p  }{2w_3} + w_4  - \frac{w_2 w_5}{2w_3} ,
 $$
 Thus 
 $$
 p_{co}^* = \frac{ 2w_3 w_4 - w_2 w_5} {w_2^2 - 4 w_1 w_3} \mbox{ and }  q_{co}^* = - \frac{ w_2 \frac{ 2w_3 w_4 - w_2 w_5} {w_2^2 - 4 w_1 w_3}  + w_5 } {2 w_3} = -\frac{ 2w_3 w_4 w_2 - 4 w_1 w_3  w_5 }{ 2 w_3 (w_2^2 - 4 w_1 w_3) }
 $$

 Now 
 $$
 p_{co}^* - \psi (q_{co}^* )  = \frac{ 2w_3 w_4 - w_2 w_5} {w_2^2 - 4 w_1 w_3} -  
\frac{1}{\alpha_\sMi (2-\varepsilon^2)} \left ( 2 \dbar_\sMi + \varepsilon  \dbar_\sMe + \varepsilon \alpha_\sMe  (C_\sMe - \frac{ 2w_3 w_4 w_2 - 4 w_1 w_3  w_5 }{ 2 w_3 (w_2^2 - 4 w_1 w_3) } ) \right ) = 
 $$
 }
 }
 $$
\left . \frac{d \omega }{ d p}\right |_{p=0} 
= \frac{2w_3w_4 - w_2w_5}{2w_3},
 $$
 which is 
  positive as $w_3 < 0$ and  $w_2,w_4,w_5$ terms are positive under {\bf A}.1-2. Thus $(0, h (0)) $ can never be the global optimizer in ${\cal F}_{co}^+$, even if $q^*(p) = h(p)$ near $p = 0$. Like wise 
\begin{eqnarray*}
     \left .
 \frac{ \partial {\cal U}(p, \theta(p) )  }{\partial p }  \right |_{p = 0} \hspace{-2mm} &=&   w_2  \theta(0) +w_4 + (2 w_3 \theta(0)   + w_5) \theta'(0)  \\
 &=&  \frac{1}{2}( \varepsilon (\alpha_\sM + \alpha_\sMe)  - 2 \varepsilon \alpha_\sM) \theta(0))  +w_4  + w_5  \frac{\varepsilon \alpha_\sM}{\alpha_\sMe}
 > 0 
\end{eqnarray*}
   when  $h(0) =-\frac{w_5}{2w_3} > \theta(0)
  $ implying $(2w_3 \theta(0) + w_5) > 0$  (as $w_3 < 0$) and since the derivative $\theta'(0) = \nicefrac{\varepsilon \alpha_\sM}{\alpha_\sMe} > 0 $ (see \eqref{Eqn_opt_policy_Mj}).

Now consider the corner case with $\bar p = \phi^{-1}(0)$ for which $q^*(\bar p) = 0$.  For this case,  by direct substitution into  \eqref{Eqn_opt_policy_Mj},  one can verify $\pe^*(\bar p, 0) = {\pe}_{mx} = \pe^*(\bar p, q) $ for any  (but sufficiently  small) $q > 0$  and that  the demands at such $(\bar p, q)$ as well as 
$(\bar p, 0)$ are positive:

\vspace{-4mm}
{\small\begin{eqnarray*}
    (\dbar_\sM - \alpha_\sM \bar p + \varepsilon \alpha_\sMe\pe^*(\bar p, 0))^+ = (\dbar_\sM - \alpha_\sM \bar p + \varepsilon\alpha_\sMe \pe^*(\bar p, q))^+ = (\dbar_\sM - \alpha_\sM \bar p + \varepsilon \alpha_\sMe {\pe}_{mx}) \\
    (\dbar_\sMe + \varepsilon \alpha_\sM \bar p -   \pe^*(\bar p, 0))^+ = (\dbar_\sMe + \varepsilon \alpha_\sM \bar p - \alpha_\sMe \pe^*(\bar p, 0))^+ = (\dbar_\sMe + \varepsilon \alpha_\sM \bar p - \alpha_\sMe {\pe}_{mx})
\end{eqnarray*}} and thus from \eqref{eqn_util_co-exist_given_pq}  we have, $U_\sV(\bar p, 0) <  U_\sV (\bar p, q)$. 
  
  
 Thus, in all, 
 $(0, q^*(0))$ or $(\bar p, 0)$ can't be global optimizers  of $U_\sV$ in ${\cal F}_{co}^+$   and  thus we have:
\begin{eqnarray}\label{eqn_L_star_intersect}
 {\mathbb L}^* = {\mathbb L}^* \cap \{ (p, q) :  q  > 0\} \cap \{ (p, q) :  p  > 0\}. 
 \end{eqnarray}


 For further analysis, we consider the 
  second derivative of the mapping $\omega$   \eqref{Eqn_omega}, which  equals,
\begin{eqnarray}
  \frac{d^2 \omega }{d p^2} =\frac{4w_1w_3 - w_2^2}{4w_3}.  \label{Eqn_sec_derivative}  
\end{eqnarray}

The second derivative   is 
either negative or positive --  thus the mapping $\omega$ is either concave or convex. This implies either of the two possibilities:

\begin{itemize}
    \item  If $
\left \{ (p, h(p) ) : \S_p \ne \emptyset \right \} \cap \mathbb{L}^*  = \emptyset
$, then from \eqref{Eqn_L_star},  the global optimizer is among (basically the global attractor of $U_v$ across each $p$-section is at the respective appropriate boundary point)
 \begin{eqnarray*}  
\mathbb {L}^*     &=& \bigg  \{ (p, q) :  0 <  p \le   {\bar p}, \ \  \ \ \ \ \ \ \  q =    {\bar q}(p)  \bigg  \} \\
&&\hspace{-13mm} \cup \bigg \{ (p, q) : {\bar p}  < p \le \pmax, \S_p \ne \emptyset, \  \  q =  l(p) \indc {h(p) < l(p)} + {\bar q}(p) \indc{h(p) > {\bar q}(p)}   \bigg  \}.  
 \end{eqnarray*}
Hence it is at 
one of the four border lines  $\mathbb{L}_1$-${\mathbb L}_4$ (observe here that $\psi (.)$ is inverse function of  $l(.)$ with $p > \bar p$). This completes the proof of the theorem for this sub-case, also by \eqref{eqn_L_star_intersect}.

\item  In the other case, some section-wise optimizers $\{(p, h(p))\}$ intersect with $\mathbb{L}^*$, and one can again have two sub-cases:  

\begin{itemize}
    \item   the global optimizer of ${\cal U}$  is in the interior of ${\cal F}_{co}^+$ and this happens if and only if  $(p_{co}^*, q_{co}^*)$ of \eqref{Eqn_pco_qco} is in the interior of $   {\cal F}_{co}^+$  because of the following reasons:
    \begin{enumerate}[(g.1)]
        \item observe $p_{co}^*$   solves $\nicefrac{d\omega}{dp}  = 0$ or equivalently $(4 w_1 w_3  -w_2^2) p_{co}^* + 2w_3 w_4 - w_5 w_2 = 0$,     $q_{co}^* = h(p_{co}^*)$ 
        
        \item  and hence $(p_{co}^*, q_{co}^*)$ becomes the global optimizer of ${\cal U}$  when $\omega$ is concave; and further

        \item  when $\omega$  is convex, from \eqref{Eqn_sec_derivative} the solution of  $\nicefrac{d\omega}{dp}  = 0$ is a negative point and hence $ (p_{co}^*, q_{co}^*)$  is not in the interior of ${\cal F}_{co}^+$.

    \end{enumerate}
     
      In this case  $(p_{co}^*, q_{co}^*)$ of \eqref{Eqn_pco_qco} is also the global optimizer of $U_v$  in~${\cal F}_{co}^+$.

\item  the global optimizer of ${\cal U}$ is outside   ${\cal F}_{co}^+$.

Then 
by concavity/convexity of $\omega$  the global optimizer of $U_\sV$ is on one of the  boundaries, $\{q= {\bar q}(p) = \min \{\theta(p), \phi(p) \} \}$ or $\{p = \min\{\pmax, \psi(q)\}\}$, hence is at one of the four boarder lines ${\mathbb L}_1$- ${\mathbb L}_4$. 
\end{itemize}
This completes the proof for this sub-case also (also by \eqref{eqn_L_star_intersect})  and hence the theorem. \eop 
\end{itemize}

One of the side results that can be derived from the proof of  the above theorem is provided below. This lemma 
  establishes the conditions 
under which the coalition  never prefers interior-Bp regime and is useful in deriving the comparison analysis: 
\begin{lemma}\label{lem_comp}
\textit{If}
$
(8-6\varepsilon^2)\alpha_\sMi \alpha_\sMe
- \varepsilon^2(\alpha_\sMi^2 + \alpha_\sMe^2) < 0,
$
\textit{then the $\Vi$ coalition finds it optimal to operate either in I$\ell$, Op, or Mp regimes.}
\end{lemma}

 {\bf Proof of Lemma \ref{lem_comp}:} Under the given condition, the second derivative in \eqref{Eqn_sec_derivative} of Appendix (while proving Theorem \ref{thm_Fco_positive}) is positive, which implies the optimal is on one of the boundaries, excluding $\{q=0\}$ and $\{p=0\}$ lines.  \eop

\subsection*{{\bf Proof of Lemma \ref{lem_compr}}:}
\label{Proof_Lemma_2}
By  {\bf A.1},  there exists a $\bar \varepsilon < 1$ such that 
 the condition in 
\eqref{eqn_il_not_opt}
is satisfied    
 and so I$\ell$ is not optimal for $\varepsilon >\bar \varepsilon$.
 Thus the coalition $\V$ finds
it optimal to operate either in   Op or Mp regimes
 by  Lemma \ref{lem_comp} for all  $\varepsilon >\bar \varepsilon$ (if required with a bigger $\bar \varepsilon$ but a  $\bar \varepsilon<1$), as  the term 
$$
(8-6\varepsilon^2)\alpha_\sMi \alpha_\sMe - \varepsilon^2(\alpha_\sMi^2 + \alpha_\sMe^2) \to - (\alpha_\sMi - \alpha_\sMe)^2  <  0.  
$$

For all   $\varepsilon > \bar  \varepsilon $ because of \eqref{eqn_il_not_opt},   we have  
$l_{_{Mp}} = 0$ (see also \eqref{Eqn_psi_inv_pmax}). Choose $\bar  \varepsilon $ further big if required such that, for all $\varepsilon > \bar  \varepsilon $,  we have   $\pmax = \nicefrac{(\dbar_\sMi + \varepsilon\dbar_\sMe)}{\alpha_\sMi} > p_{sw}$ (see   \eqref{Eqn_psw} and by {\bf A.1}). 
Hence  as $\varepsilon \to 1$, from  second row of \eqref{eqn_r_max_simplified}  and from  \eqref{eqn_q_star_p} of \ref{sec_Appendix_AA},  we have that:    

\vspace{-3mm}
{\small\begin{eqnarray}
 \label{Eqn_rmp_hpax_lim}   
r_{_{Mp}} &\stackrel{\mbox{\tiny when $\varepsilon \ge \bar \varepsilon$}}{=}&  \phi(\pmax)  \to \frac{ (\dbar_\sMi - \alpha_\sMe C_\sMe)}{\alpha_\sMe} 
\mbox{ and }  \\
h(\pmax) &\to& \frac{ (\frac{\alpha_\sMi + \alpha_\sMe }{2})(\frac{\dbar_\sMi+ \dbar_\sMe}{\alpha_\sMi}) -\frac{\alpha_\sMe(C_\sMi + C_\sMe)}{2} + \frac{\dbar_\sMe}{2}}{\alpha_\sMe}, \mbox{ and thus } \nonumber
 \\
    \label{Eqn_hpmax_minus_rmp}
&& \hspace{-35mm} z\ := \ \lim_{\varepsilon\to 1} \left ( h(\pmax) - r_{_{Mp}} \right )  = \frac{ \frac{ \alpha_\sMe - \alpha_\sMi}{\alpha_\sMi}\dbar_\sMi
+ \frac{2\alpha_\sMi + \alpha_\sMe}{\alpha_\sMi}\dbar_\sMe
+   \alpha_\sMe(C_\sMe - C_\sMi) }{2\alpha_\sMe} . \hspace{4mm}
\end{eqnarray}}
Observe the sign of $z $ equals that of the  LHS in \eqref{Eqn_cond_forat_max} and hence $z< 0$ iff \eqref{Eqn_cond_forat_max} is false. 

{\bf Case 1:}  If  in \eqref{Eqn_hpmax_minus_rmp}
  $z< 0$, then 

\vspace{-3mm}
{\small$$U^*_{_{Mp}} = U_\sV (\pmax, h (\pmax) )  \mbox{ as } l_{_{Mp}} = 0 < h(\pmax) \mbox{  and by  \eqref{Eqn_Ustar_Op_when_z_negative}, }  U^*_{_{O_p}} = U_\sV (\pmax, \theta(\pmax) ). $$}%
Also by optimality of $h(\pmax)$, along $p=\pmax$ line, we have 
$$
U^*_{_{O_p}}=U_\sV (\pmax, \theta(\pmax) ) \le U_\sV (\pmax, h(\pmax) ) = U^*_{_{Mp}}  .
$$

{\bf Case 2:} Now if in \eqref{Eqn_hpmax_minus_rmp}
  $z\ge 0$, then 
  $
  U^*_{_{Mp}} = U_\sV (\pmax, r_{_{M_p}} ).
  $
Further since $\pe^{*}(\pmax, r_{_{Mp}} ) = \pe^{*}(\pmax, \theta(\pmax))  = \nicefrac{(\dbar_\sMe + \dbar_\sMi)}{\alpha_\sMe}$ at limit,  using \eqref{Eqn_rmp_hpax_lim} and  \eqref{Eqn_psw}:

\vspace{ -4 mm}
{\small\begin{eqnarray*}
\lim_{\varepsilon \to 1}  \left (  U_\sV (\pmax, \theta(\pmax) ) - U_\sV (\pmax, r_{_{Mp}}  ) \right ) \hspace{-59mm} \nonumber \\
   &=&  \lim_{\varepsilon \to 1}  \left [
     \left(\dbar_\sMe + \alpha_\sMi \pmax - \alpha_\sMe \pe^{*}(\pmax,\theta(\pmax))\right)\left(\theta(\pmax) -  r_{_{Mp}}   \right) \right ]\label{eqn_util_diff} \\
       &=&  
 \dbar_\sMe\left ( \frac{\dbar_\sMe +  \dbar_\sMi - \alpha_\sMe C_\sMe  }{\alpha_\sMe} -\frac{   O_\sMe}{  \dbar_\sMe}  -\frac{\dbar_\sMi- \alpha_\sMe C_\sMe }{\alpha_\sMe  } \right )  \ 
           = \ \frac{ (\dbar_\sMe)^2 }{\alpha_\sMe}  -  O_\sMe > 0, 
   \end{eqnarray*}}
and thus  we have $U^*_{_{O_p}} \ge U_\sV (\pmax, \theta(\pmax) ) \ge U^*_{_{M_p}}$.   \eop

\begin{lemma}\label{lem_op_max_price}
Under the hypothesis of Theorem \ref{thm_coex},   there exists a $\bar\varepsilon > 0$ such that for all $\varepsilon \le \bar\varepsilon $: \\
(a) $\textit{Mp}$ regime is empty,  and 
(b) $U_{_{Sh}}^* > U_{_{I\ell}}^*$.
\end{lemma}
\textbf{Proof of Lemma \ref{lem_op_max_price}:}
From \eqref{Eqn_psw}, 
$
\lim_{\varepsilon \to 0} \ p_{mx} = \nicefrac{\dbar_\sMi} {\alpha_\sMi}  \mbox{ and }   p_{sw} \to \infty, 
$
thus  $p_{mx} \le  p_{sw}$ for all $\varepsilon \le \bar \varepsilon$ (for some $\bar \varepsilon > 0$). Now from \eqref{Eqn_psi_inv_pmax} and \eqref{eqn_r_max_simplified}, we have the following as $\varepsilon \to 0$,
\begin{eqnarray*}
 l_{_{Mp}} &\to& \max \left \{ \frac{\dbar_\sMe - \alpha_\sMe C_\sMe}{\alpha_\sMe}, 0  \right  \}  =   \frac{\dbar_\sMe - \alpha_\sMe C_\sMe}{\alpha_\sMe} \mbox{ by assumption } {\bf A.1}, \mbox{ and, }\\
 r_{_{Mp}}  &\to& \frac{\dbar_\sMe - \alpha_\sMe C_\sMe - 2\sqrt{\alpha_\sMe O_\sMe}}{\alpha_\sMe}.
\end{eqnarray*}
Thus  $\lim_{\varepsilon \to 0} ( l_{_{Mp}} - r_{_{Mp}} ) > 0$ and hence the Mp regime is empty by Theorem~\ref{Thm_all_in_one}; this completes the proof of part (a). 

{\bf Part (b):} 
From \eqref{eqn_tilde_q_star}, \eqref{Eqn_psi_inv_less_phi_at_pmax} and \eqref{Eqn_feasible_Regioin_Mj}, we have the following limits as $\varepsilon\to 0$:
\begin{eqnarray*}
\tilde{q}^{*} &\to& \frac{\dbar_\sMe - \alpha_\sMe C_\sMe + \alpha_\sMe C_\sS}{2\alpha_\sMe}, \ 
\psi^{-1}(p_{mx})  \ \to \  \frac{\dbar_\sMe - \alpha_\sMe C_\sMe}{\alpha_\sMe}, \mbox{ and } \\
\theta(p_{mx}) &\to&  \frac{\dbar_\sMe - \alpha_\sMe C_\sMe - 2\sqrt{\alpha_\sMe O_\sMe}}{\alpha_\sMe}, \mbox{ and } \lim_{\varepsilon\to 0} \left ( \theta(p_{mx})  -   \psi^{-1}(p_{mx}) \right ) < 0.
\end{eqnarray*}
Thus from \eqref{eqn_pi_il}, $r_{_{I\ell}} = \theta(p_{mx})$ and $\tilde{q}^{*} < \theta(p_{mx})$,  for all  $\varepsilon  \le  \bar \varepsilon$  for some  $\bar \varepsilon > 0$. Thus  for all such $\varepsilon$, we have 
\begin{eqnarray*}
    U_{_{I\ell}}^* = U_{_{I\ell}} (\pmax, \tilde{q}^{*} ) \mbox{ and }  \lim_{\varepsilon \to 0}  U_{_{I\ell}}^*  = \frac{ \left (\dbar_\sMe  -\alpha_\sMe (C_\sMe + C_\sS) \right )^2 }{8 \alpha_\sMe} - O_\sMi - O_\sS.
\end{eqnarray*}
Next by  Lemma \ref{lem_shut}, as $\varepsilon \to 0$,   
\begin{eqnarray}\label{eqn_sh_eps_zero}
 U_{_{Sh}}^*
\to
\frac{\left(\dbar_\sMe - \alpha_\sMe(C_\sMe + C_\sS)\right)^2}
{8\alpha_\sMe}
- O_\sS,   
\end{eqnarray}
and hence part~(b) follows, as $O_\sM > 0$.  \eop

\medskip

\medskip

\subsection*{\bf Proof of Theorem \ref{thm_coex}:}
\label{ref_thm_4}
The idea is to   compare the utilities at different    regimes with that at co-existence,  at limit $\varepsilon \to 0$; we will show that
all the limits are strictly less than $  \lim_{\varepsilon \to 0} U^*_{co}.$
 Then by continuity,  we will have inequalities in the same direction for any $\varepsilon$ small enough.

We begin with the Bp co-existence regime.  From \eqref{Eqn_ws}, as 
$\varepsilon \to 0$,  we have:

\vspace{-3mm}
{\small\begin{eqnarray}
    w_1 &\to& \hspace{-2mm}-\alpha_\sMi,  \  \
    w_2 \to 0,  \  \ 
    w_3 \to -\frac{\alpha_\sMe}{2}, \nonumber\\
   w_4 &\to&    (\dbar_\sMi +  \alpha_\sMi(C_\sMi + C_\sS)) ,  \mbox{ and, } 
    w_5 \to  \frac{\left(\dbar_\sMe - \alpha_\sMe C_\sMe + \alpha_\sMe C_\sS\right)}{2}. \ \hspace{4mm}\nonumber
\end{eqnarray}}
By substituting the above limits  in \eqref{Eqn_pco_qco}, we have the following as  $\varepsilon \to 0$,
\begin{eqnarray*}
 p^{*}_{co} \to \frac{\left(\dbar_\sMi + \alpha_\sMi(C_\sMi + C_\sS)  \right)}{2\alpha_\sM}, \  \  q^{*}_{co} \to \frac{\left(\dbar_\sMe - \alpha_\sMe(C_\sMe + C_\sS)  \right)}{2\alpha_\sMe}.
\end{eqnarray*}
By substituting the above  in \eqref{eqn_util_co-exist_uc}, we have
\begin{eqnarray}\label{eqn_coex_eps_to_zero}
\lim_{\varepsilon \to 0} U^*_{co}
 =   \frac{\left(\dbar_\sMi - \alpha_\sMi (C_\sMi +C_\sS) \right )^2  }{4 \alpha_\sMi} +\frac{\left(\dbar_\sMe - \alpha_\sMe (C_\sMe +C_\sS) \right )^2  }{8 \alpha_\sMe} - O_\sMi  - O_\sS. \hspace{4mm}   
\end{eqnarray}
By Lemma \ref{lem_op_max_price}.(b),   and   using \eqref{eqn_sh_eps_zero}
$$
\lim_{\varepsilon \to 0} U^{*}_{_{_{I\ell}}}  < \lim_{\varepsilon \to 0} U^{*}_{_{Sh}}   < \lim_{\varepsilon \to 0} U{^*}_{co}.  
$$ 
Next from Lemma \ref{lem_elim}, 
$$
\lim_{\varepsilon \to 0} U^{*}_{_{E\ell}} = \frac{\left(\dbar_\sMi - \alpha_\sMi (C_\sMi +C_\sS) \right )^2  }{4 \alpha_\sMi} - O_\sS - O_\sMi  < \lim_{\varepsilon \to 0} U^*_{co}.  
$$ 
Thus, by Lemma \ref{lem_op_max_price}.(a), it suffices to show that 
$\lim_{\varepsilon \to 0}  U^*_{_{Op}} < \lim_{\varepsilon \to 0} U^*_{co}   $, which is the last part of this proof.  Towards  computing  $\lim_{\varepsilon \to 0} U^*_{_{Op}}$, first observe  using \eqref{Eqn_psw}  that  $\lim_{\varepsilon \to 0} p_{sw} > \lim_{\varepsilon \to 0} p_{mx}$; thus using \eqref{Eqn_opt_util_par} $ \lim_{\varepsilon \to 0} U^*_{_{Op}} = 
   \lim_{\varepsilon \to 0} U_{\sV} (p^{1,*} \theta(p^{1,*}))$. Also, observe the following by assumption {\bf A.1},
 \begin{eqnarray*}
     \lim_{\varepsilon \to 0}  \left ( \frac{(C_\sMi + C_\sS)}{2} 
    + \frac{\left( \dbar_\sMi + \varepsilon\dbar_\sMe - \varepsilon\sqrt{\alpha_\sMe O_\sMe} + \frac{\varepsilon\alpha_\sMi \sqrt{\alpha_\sMe O_\sMe}}{\alpha_\sMe} \right)}{2\alpha_\sMi(1-\varepsilon^2)} \right )  &<&  \lim_{\varepsilon \to 0} p_{mx}.
 \end{eqnarray*}
Thus by \eqref{eqn_u2}, \vspace{-4mm}
\begin{eqnarray*}
      \lim_{\varepsilon \to 0} p^{1,*} = \frac{\dbar_\sMi }{2\alpha_\sMi} + \frac{C_\sMi + C_\sS}{2}.
\end{eqnarray*}
After substituting
$\lim_{\varepsilon \to 0} p^{1,*}$  as $p$ 
in   \eqref{eqn_u1} and using \eqref{Eqn_opt_util_par}, we obtain the limit of the optimal Op utility:

\vspace{-3mm}
{\small{
\begin{eqnarray}
 \lim_{\varepsilon \to 0} U^*_{_{Op}} &=& 
   \lim_{\varepsilon \to 0} U_{\sV} (p^{1,*} \theta(p^{1,*})) \nonumber  \\
   &=& - O_\sMi - O_\sS  + \frac{\left(\dbar_\sMi -\alpha_\sMi(C_\sMi + C_\sS)\right)^2}{4\alpha_\sMi} \nonumber \\
   &&  + \  \frac{\sqrt{\alpha_\sMe O_\sMe}\left(\dbar_\sMe - \alpha_\sMe (C_\sMe + C_\sS) - 2\sqrt{\alpha_\sMe O_\sMe} \right)}{\alpha_\sMe}. \hspace{10mm}\label{eqn_par_eps_to_zero}
\end{eqnarray}}

Now, comparing  \eqref{eqn_coex_eps_to_zero} and \eqref{eqn_par_eps_to_zero}, we 
obtain, 
\begin{eqnarray*}
 \lim_{\varepsilon \to 0} (U^*_{co} - U^{*}_{_{Op}}) &=& \frac{\left(\dbar_\sMe - \alpha_\sMe (C_\sMe +C_\sS) \right )^2}{8\alpha_\sMe}\\ &-& \frac{8\sqrt{\alpha_\sMe O_\sMe}(\dbar_\sMe - \alpha_\sMe (C_\sMe + C_\sS) - 2\sqrt{\alpha_\sMe O_\sMe}) }{8\alpha_\sMe}\\
 &=& \frac{\left( \dbar_\sMe - \alpha_\sMe (C_\sMe +C_\sS) - 4\sqrt{\alpha_\sMe O_\sMe} \right)^2}{8\alpha_\sMe} > 0.
\end{eqnarray*}
Thus, we have, 
$
\lim_{\varepsilon \to 0} U^{*}_{co} > \lim_{\varepsilon \to 0} U^{*}_{_{Op}}.  $
\eop

{\bf Proof of Theorem \ref{thm_all_r}:}
\label{ref_thm_5}
From Theorem \ref{Thm_all_in_one} and Lemma \ref{lem_compr}, there exists an $\tilde\varepsilon < 1$ such that for all $\varepsilon > \tilde{\varepsilon}$, we have $ U^{*}_{co} = \max \{U^{*}_{_{Op}}, U^{*}_{_{Mp}}\}$. Further, from the proof of Lemma \ref{lem_compr}, we have ($\lim_{\varepsilon \to 1} h(\pmax), \lim_{\varepsilon \to 1} \theta(\pmax) $ are computed and $\lim_{\varepsilon \to 1} \pmax = \nicefrac{(\dbar_\sMi+\dbar_\sMe)}{\alpha_\sMi}$ and substituting the same into \eqref{eqn_util_co-exist_given_pq}) 
\begin{eqnarray}
  \lim_{\varepsilon \to 1} U^{*}_{_{Op}} &=&\lim_{\varepsilon \to 1} U_\sV(\pmax,\theta(\pmax)) =\dbar_\sMe\left ( \frac{\dbar_\sMe +  \dbar_\sMi - \alpha_\sMe C_\sMe  }{\alpha_\sMe} -\frac{   O_\sMe}{  \dbar_\sMe}  \right )  \label{eqn_op_1}  \hspace{4mm} \\  
   \lim_{\varepsilon \to 1} U^{*}_{_{Mp}} &=& \lim_{\varepsilon \to 1} U_\sV(\pmax,h(\pmax)) = \frac{  (2\dbar_\sM+ \dbar_\sMe) \left(\dbar_\sM +\dbar_\sMe -\alpha_\sM(C_\sS + C_\sM) \right)}{2  \alpha_\sM} \nonumber\\ &+& \dbar_\sMe \left( \frac{\dbar_\sMi+ 2\dbar_\sMe - \alpha_\sMe(C_\sMe + C_\sS) }{2\alpha_\sMe}\right )
    - O_\sM - O_\sS. \label{eqn_mp_1}
\end{eqnarray}
Thus $\lim_{\varepsilon \to 1}  U^{*}_{co} < \infty $. 
Now obtain the following using Lemmas \ref{lem_shut}-\ref{lem_elim} for any given $r \in [0, 1]$:
\begin{eqnarray}
\Omega_{_{Sh}} (r) := \lim_{\varepsilon \to 1} U^{*}_{_{Sh}} &=& \frac{\left(\dbar_\sMe +  \dbar_\sMi - \alpha_\sMe(1-r)(C_\sMe+ C_\sS)\right)^2}{8\alpha_\sMe(1-r)} - O_\sS. \label{eqn-omega_sh} \\  
\Omega_{_{El}} (r) :=  \lim_{\varepsilon \to 1} U^{*}_{_{El}} &=& \frac{\left(\dbar_\sMi + \dbar_\sMe - \alpha_\sMi(1-r)(C_\sMi+ C_\sS)\right)^2}{4\alpha_\sMi(1-r)} - O_\sMi - O_\sS. \hspace{2mm}\label{eqn-omega_el} \hspace{6mm}
\end{eqnarray}
Also consider the derivative of the above two terms which are positive by {\bf{A.1}}:
\begin{eqnarray*}
    \frac{\partial \Omega_{_{Sh}} (r)}{\partial r} &=& \frac{16\alpha_\sMe^2(1-r)(C_\sMe+C_\sS)(\dbar_\sMe + \dbar_\sMi -\alpha_\sMe(1-r)(C_\sMe + C_\sS))}{64\alpha_\sMe^2(1-r)^2}\\
    &+& \frac{8\alpha_\sMe (\dbar_\sMe + \dbar_\sMi -\alpha_\sMe(1-r)(C_\sMe + C_\sS))^2}{64\alpha_\sMe^2(1-r)^2} > 0, \\
    \frac{\partial \Omega_{_{El}} (r)}{\partial r} &=& \frac{8\alpha_\sMi^2(1-r)(C_\sMi+C_\sS)(\dbar_\sMe + \dbar_\sMi -\alpha_\sMi(1-r)(C_\sMi + C_\sS))}{16\alpha_\sMi^2(1-r)^2}\\
    &+& \frac{4\alpha_\sMi (\dbar_\sMe + \dbar_\sMi -\alpha_\sMi(1-r)(C_\sMi + C_\sS))^2}{16\alpha_\sMi^2(1-r)^2} > 0.
\end{eqnarray*}
As seen from the above, both $\Omega_{_{Sh}}$ and $\Omega_{_{E\ell}}$ are strictly increasing in $r$;  further   
 $\lim_{r \to 1}\Omega_{_{Sh}} (r) = \infty$ and  $\lim_{r \to 1}\Omega_{_{El}} (r) = \infty$.
Thus there exists a $\bar r  < 1$ such that 
\begin{eqnarray}
\lim_{\varepsilon \to 1}  U^{*}_{co} &<& \min \{ \Omega_{_{Sh}} (r)  \Omega_{_{El}} (r) \} \mbox{ for all } r >  \bar r, \mbox{ and }  \nonumber  \\ 
\lim_{\varepsilon \to 1}  U^{*}_{co}  &>&  \min \{ \Omega_{_{Sh}} (r), \Omega_{_{El}} (r) \} \mbox{ for all } r <  \bar r. \label{Eqn_comparison}
\end{eqnarray}

{\bf Case 1:
Fix any   $r  < \bar r$}.
From \eqref{Eqn_comparison}, 
 fo any $r < \bar r$, the overall optimizer is among the co-existence regime. Thus this part of the theorem follows by Lemma~\ref{lem_compr}, with $\epsilon_r = {\tilde \varepsilon}$, defined at the beginning of the proof, for every $r < \bar r$.

{\bf Case 2:
Fix any   $r  > \bar r$}. If $\Omega_{_{Sh}} (r)-  \Omega_{_{El}} (r)  > 0$ given by \eqref{eqn_ineq_elim_shut}, then  there exists an $\epsilon_r $ such that 
$U^*_{_{Sh}} (r, \varepsilon) > U^* _{_{El}} (r, \varepsilon)$, for all $\varepsilon \ge \epsilon_r$ and thus 
$$
U^*_\sV = U^*_{_{Sh}} \mbox{ for all }  \varepsilon \ge \epsilon_r, \mbox{ and  all } r \ge \bar r.  
$$
and thus the optimal configuration is to shut down the in-house. 
In similar lines, when 
$\Omega_{_{Sh}} (r)-  \Omega_{_{El}} (r)  < 0$ for all $\varepsilon \ge \epsilon_r$ (for some appropriate $\epsilon_r$),  E$\ell$ is the optimal among all configurations.  \eop

\subsection*{\underline{Some computations for Op regime as $\varepsilon \to 1$}}
With $\varepsilon \to 1$, we have $\lim_{\varepsilon \to 1} (\pmax - p_{sw} ) > 0 $, and  so by \eqref{eqn_u2} $  p^{1,*} \to \lim_{\varepsilon \to 1}   p_{sw}  = \frac{\dbar_\sM + \sqrt{ \alpha_\sMe O_\sMe}}{\alpha_\sM}$. 

Further towards understanding $p^{2,*}$  of \eqref{Eqn_p2_star}, at limit $\varepsilon\to 1$, we first observe

\vspace{-3mm}
{\small \begin{eqnarray*}
    \frac{
\dbar_{\sMi}(1+\varepsilon^{2})
+ \varepsilon\dbar_{\sMe}
+ \alpha_{\sMi}(C_{\sMi}+C_{\sS})
+ \frac{\varepsilon\alpha_{\sMi}}{\alpha_{\sMe}}
(\dbar_{\sMe}+\varepsilon\dbar_{\sMi}-\alpha_{\sMe} C_{\sMe}-\alpha_{\sMe} C_{\sS})
}{2\alpha_{\sMi}} - \pmax  \hspace{-120mm}
  \\ &\to& 
     \frac{ (2\alpha_{\sMe} + \alpha_{\sMi} )
\dbar_{\sMi} 
+ (\alpha_{\sMe} + \alpha_{\sMi} ) \dbar_{\sMe}
+ \alpha_{\sMi} \alpha_\sMe (C_{\sMi} - C_{\sMe})
}{2\alpha_{\sMi} \alpha_\sMe }   -  \frac{\dbar_\sMe + \dbar_\sM } {\alpha_\sM}  \\
 &=& \frac{   \alpha_{\sMi} 
\dbar_{\sMi} 
+ ( - \alpha_{\sMe} + \alpha_{\sMi} ) \dbar_{\sMe}
+ \alpha_{\sMe}\alpha_{\sMi}(C_{\sMi} - C_{\sMe})
}{2\alpha_{\sMi} \alpha_\sMe } .
\end{eqnarray*}}

{\bf Case 1:} When $z$ in \eqref{Eqn_hpmax_minus_rmp} is negative or when 

\vspace{-3mm}
{\small$$
\alpha_\sMe \alpha_\sMi C_\sMi + \alpha_\sMi \dbar_\sMi > \alpha_{\sMe}  \dbar_{\sMi} +  ( 2\alpha_{\sMi} + \alpha_{\sMe} ) \dbar_{\sMe} + \alpha_\sMe \alpha_\sMi C_\sMe > \alpha_{\sMe}  \dbar_{\sMe}  + \alpha_\sMe \alpha_\sMi C_\sMe
$$}
we have 
$$
\lim_{\varepsilon \to 1} p^{2,*} = \pmax, 
$$
and this implies, by \eqref{Eqn_p2_star}, that the optimizer of $p \mapsto U_\sV (p, \theta(p) )$ in  range $[p_{sw}, \pmax]$ is at $\pmax$ while the unconstrained optimizer of the objective function  in  \eqref{Eqn_second_opt_Par} is above   $\pmax$.  
Further since $p_{sw} < \pmax$ at limit, 
  by strict concavity of $p \mapsto U_\sV (p, \theta(p) )$ on  $[p_{sw}, \pmax]$  and by continuity of  $p \mapsto U_\sV (p, \theta(p) )$ on $[0, \pmax]$,  we have
\begin{eqnarray}
  \label{Eqn_Ustar_Op_when_z_negative}
U^*_{Op} = U_\sV (\pmax, \theta(\pmax) )  \mbox{(when $z$ of \eqref{Eqn_hpmax_minus_rmp} is negative)}, 
\end{eqnarray}
also as optimal  of $U_\sV (p, \theta(p) )$ over the range $[0, p_{sw}]$ is at $p_{sw}$.

\end{document}